%% file: tpg-ssy.tex
\pdfoutput=1
\documentclass[11pt]{article}
\usepackage[margin=1in]{geometry}

\usepackage[utf8]{inputenc}
\usepackage[T1]{fontenc}
\usepackage{amsmath,amssymb,amsthm}
\usepackage{graphicx}
\usepackage{url}
\usepackage{booktabs}
\usepackage{nicefrac}
\usepackage{microtype}
\usepackage{caption}
\usepackage{xcolor}
\usepackage{enumitem}
\usepackage{setspace}
\usepackage{algorithm}
\usepackage{algpseudocode}
\usepackage[caption=false,font=footnotesize]{subfig}
\usepackage{natbib}
\bibpunct[, ]{(}{)}{,}{a}{}{,}

\usepackage[breaklinks=true]{hyperref}
\hypersetup{
  colorlinks=true,
  linkcolor=blue,
  citecolor=blue,
  urlcolor=blue,
  pdfauthor={Ramesh Johari, Tianyi Peng, Wenqian Xing},
  pdftitle={Estimation of Treatment Effects Under Nonstationarity via the Truncated Policy Gradient Estimator},
  pdfkeywords={treatment effect estimation, nonstationarity, policy gradient, Markov decision processes, temporal interference}
}
\theoremstyle{plain}
\newtheorem{theorem}{Theorem}
\newtheorem{lemma}{Lemma}
\newtheorem{proposition}{Proposition}
\newtheorem{corollary}{Corollary}
\newtheorem{assumption}{Assumption}
\theoremstyle{definition}
\newtheorem{definition}{Definition}
\newtheorem{remark}{Remark}

\input{macros}

\title{Estimation of Treatment Effects Under Nonstationarity via the Truncated Policy Gradient Estimator}
\author{%
Ramesh Johari\thanks{Management Science and Engineering Department, Stanford University, \texttt{rjohari@stanford.edu}}
\and
Tianyi Peng\thanks{Columbia Business School, Columbia University, \texttt{tianyi.peng@columbia.edu}}
\and
Wenqian Xing\thanks{Management Science and Engineering Department, Stanford University, \texttt{wxing@stanford.edu}}
}
\date{October 2026}

\begin{document}

\maketitle

\begin{abstract}
Randomized experiments (or A/B tests) are widely used to evaluate interventions in dynamic systems such as recommendation platforms, marketplaces, and digital health. In these settings, interventions affect both current and future system states, so estimating the global average treatment effect (GATE) requires accounting for temporal dynamics, which is especially challenging in the presence of {\em nonstationarity}; existing approaches suffer from high bias, high variance, or both. In this paper, we address this challenge via the novel Truncated Policy Gradient (TPG) estimator, which replaces instantaneous outcomes with short-horizon outcome trajectories. 
The estimator admits a {\em policy gradient} interpretation: it is a truncation of the first-order approximation to the GATE, yielding explicit bias and variance bounds in nonstationary Markovian settings.
For a fixed truncation size, we establish a central limit theorem and, under an additional mean-stability condition, the consistency of a variance estimator for nonstationary single-trajectory data.
We validate our theory with two real-world case studies. 
The results show that relative to existing approaches, a well-calibrated TPG estimator can achieve a favorable balance between bias and variance in nonstationary settings, highlighting the value of the policy-gradient perspective for designing effective estimators under complex dynamics.

\end{abstract}

\begingroup
\makeatletter
\let\orig@addcontentsline\addcontentsline
\renewcommand{\addcontentsline}[3]{%
  \def\tempa{#1}\def\tempb{toc}%
  \ifx\tempa\tempb
  \else
    \orig@addcontentsline{#1}{#2}{#3}%
  \fi}
\makeatother

\section{Introduction}
Randomized controlled experiments (``A/B tests'') are essential tools for evaluating the impact of interventions across dynamic, technology-driven environments such as recommendation systems, online marketplaces, and digital health platforms. Typically, these interventions do not merely influence immediate outcomes but also propagate effects over subsequent system states. This leads to a breakdown of the standard assumption of unit-level independence, a phenomenon known as a {\em carryover effect} \citep{bojinov2023design,hu2022switchback} or \textit{temporal interference} \citep{glynn2020adaptive}.  In such settings, it is well established that applying a naive estimator (e.g., the difference-in-means or Horvitz-Thompson estimators \citep{horvitz1952generalization}) can result in substantial bias as it ignores these intertemporal dependencies.  For instance, in a digital health platform, randomization of incoming patients to treatment and control affects the workload of clinical providers, which in turn impacts the care experience of subsequent patients.  Similarly, in a ride-sharing platform, user assignments to treatment and control affect the availability of drivers: a booking decision by one user influences the system state---e.g., driver distribution or availability---which in turn affects the experience and behavior of subsequent users.

Crucially, these real-world systems are also typically {\em nonstationary}: their dynamics can vary in a time-dependent manner, driven by factors such as evolving user preferences, time-dependent driver dynamics, and external shocks (e.g., seasonal trends or viral events).  In the presence of nonstationarity, the problem of experimental design and estimation becomes even more challenging -- a major issue for practitioners. 
In this paper, our central contribution is a simple and practical approach that leverages a policy-gradient perspective to yield an estimator that provides a favorable balance of bias and variance relative to existing approaches, in settings with both carryover effects and nonstationarity.

One way to model such dynamic, nonstationary environments is via a \emph{Markovian} system, represented by a tuple sequence $(X_t, Z_t, Y_t)$ for a finite time horizon $t = 1, \dotsc, T$. Here, $X_t$ denotes the system state at time $t$ (e.g., the number and locations of available drivers), and $Z_t \in \{0, 1\}$ indicates whether the user arriving at time $t$ is assigned to the control or treatment group. The variable $Y_t$ captures the reward resulting from the user's behavior after receiving the assignment (e.g., $Y_t = 1$ if the user books a trip). Conditioned on the state $X_t$ and action $Z_t$, the system evolves according to:
\begin{align*}
    X_{t+1} \sim P_{Z_t}^t(X_t, \cdot),
\end{align*}
where $P_0^t$ and $P_1^t$ denote the (unknown) transition kernels under control and treatment at time $t$. The superscript $t$ reflects the system's nonstationarity.

{Platforms typically want to compare} two policies: the \emph{treatment policy}, which assigns $Z_t = 1$ for all $t$, and the \emph{control policy}, which assigns $Z_t = 0$ for all $t$. The estimand of interest is the expected average reward under each policy over a fixed finite time horizon, commonly referred to as the \emph{global average treatment effect} (GATE). Since it is infeasible to observe both policies simultaneously, data are typically collected under an exploratory policy (the {\em experiment})---for instance, a Bernoulli randomization policy that assigns each $Z_t$ independently to $0$ or $1$ with equal probability, as in a typical A/B test setting.

This yields a challenging causal inference problem: \textit{estimation of the GATE from a single trajectory of the nonstationary system under the experiment.}  A variety of approaches, of differing levels of sophistication, exist for this problem:
\begin{itemize}[left=0pt]
  \item \textbf{Naive estimator.} A naive approach is to simply ignore the interference and take the difference in average outcomes between treatment and control groups (e.g., a difference-in-means or Horvitz-Thompson estimator \citep{horvitz1952generalization}).
Though straightforward, the bias of such an estimator can be as large as the treatment effect itself \citep{glynn2020adaptive,farias2022markovian,hu2022switchback},  which motivates the line of research aimed at addressing interference in the first place.

  \item \textbf{Stationary OPE.} Recent research has explored more sophisticated techniques inspired by reinforcement learning, including off-policy evaluation and the recently proposed difference-in-Q's (DQ) estimator \citep{farias2022markovian,farias2023correcting,uehara2022review,shi2023dynamic}. However, most of these methods assume a \textit{stationary} environment. When applied to our setting, they tend to exhibit high bias and substantial variance (see Sections~\ref{sec:policy_gradient} and~\ref{sec:numerical}), as they rely heavily on assumptions such as the stationary observability of system states or the availability of multiple independent trajectories---none of which hold in our setting. 

  \item \textbf{More complex designs.} An orthogonal line of research investigates more elaborate experimental designs, such as switchback experiments \citep{hu2022switchback}, to address nonstationarity. However, these designs can be challenging to implement in practice and often require careful tuning of parameters such as interval length. This motivates our investigation into whether improvements are possible within the {more commonly utilized} Bernoulli randomized {design in} A/B testing.
\end{itemize}
Perhaps surprisingly, we find that a simple modification of the {difference-in-means (DM)\footnote{Strictly speaking, this is a naive Horvitz--Thompson estimator, but we use the more common language of ``difference-in-means,'' understanding that when $T$ is large, the gap between the two naive estimators is minimal.}} estimator yields a promising low-bias, low-variance alternative.
Specifically, rather than comparing {\em instantaneous} outcomes (as in the DM estimator), we compare {\em $k$-step accumulated rewards}, with $k$ fixed (and chosen to be small relative to the experiment horizon).

To motivate our approach, we show that this difference of $k$-step sums can be viewed as an approximate {\em policy gradient}: namely, the gradient of the cumulative reward with respect to the probability of assignment to treatment.  We thus refer to our estimator as the {\em Truncated Policy Gradient} (TPG) estimator (see Section~\ref{sec:est_non} for the formal definition).   By a Taylor approximation argument, this policy gradient approximates the GATE.  We show that we can interpret the $k$-step sums for treatment and control respectively as an appropriately defined $Q$-value function over a truncated horizon; thus the TPG estimator is a difference in these ``truncated'' $Q$-values, as is the case in the DQ estimator in a stationary setting described above.

 We find that it offers several appealing properties:
\begin{itemize}[left=0pt]
    \item \textbf{Theoretical.} We show that the TPG estimator achieves substantial variance reduction, at the cost of a controlled mixing bias, under mild assumptions on the environment's mixing time.
    Our analysis is based on a first-order policy-gradient formulation tailored to nonstationary environments.
    We also establish a CLT and develop a \emph{practical}, \emph{consistent} variance estimator for single-trajectory nonstationary data.
    
    \item \textbf{Practical.} The estimator is simple to implement: 
    \begin{enumerate}
        \item It is based entirely on \emph{observed rewards} and does not require knowledge of the state space;
        \item It can be computed from a single sample path, without requiring multiple i.i.d. trajectories;
        \item It allows for post-experiment inference and analysis. {\color{black}Under the assumptions in Section~\ref{sec:var_est}, the variance estimator supports confidence intervals for the mean of the fixed-$k$ estimator.}
        Furthermore, we propose a heuristic by which the truncation size $k$ can be evaluated post-experiment, enabling practitioners to adaptively select the truncation level (see, e.g., Appendix~\ref{appendix:algo}).
    \end{enumerate}

    \item \textbf{Empirical.} {In Section \ref{sec:numerical}}, we empirically evaluate the TPG estimator using two realistic simulations: a queueing model based on nonstationary patient arrivals to a hospital emergency department, and a large-scale ride-sharing system based on NYC taxi data. Both settings reflect real-world service system dynamics.
    In this highly nonstationary environment, we compare its performance to standard OPE methods and a recent switchback-based approach \citep{hu2022switchback}. With appropriate calibration, our estimator consistently exhibits substantially lower bias and variance, demonstrating its robustness and practical utility in real-world settings.
\end{itemize}

\subsection{Related Literature}
\label{sec:related_literature}

\paragraph{Experimental design and estimation under network interference.}
Our work is related to the literature on experimental design and estimation under interference.
In network experiments, \citet{aronow2017estimating} provide a general framework for causal inference based on treatment assignment mechanisms, exposure mappings, and estimands. \citet{eckles2017design} study graph-based randomization designs to reduce interference bias, while \citet{viviano2023causal} develop cluster-based designs that balance the bias and variance of global treatment effect estimation.
On the estimation side, \citet{peng2025differences} propose the differences-in-neighbors estimator, which uses local approximations to achieve a favorable bias--variance trade-off. A common challenge in these settings is that the treatment effect under experimental randomization can differ from the global treatment effect under uniform treatment assignments. Our work considers a similar problem in dynamic systems, where interference propagates through time rather than across a network.

\paragraph{Interference in online marketplaces.}
Another related line of work studies experimental design and estimation in online marketplaces, where interference arises from interactions between market participants. \citet{johari2022experimental} investigate how market balance affects the bias of customer-side and listing-side randomization and propose two-sided experimental designs. \citet{wager2021experimenting} develop a mean-field framework for estimating gradients of equilibrium performance through small randomized perturbations. \citet{johari2024does} study how interference affects treatment effect estimation and hypothesis testing, showing that reliable launch decisions depend on the structure of the treatment. \citet{wu2024switchback} further investigate temporal interference induced by forward-looking demand in pricing experiments.
These studies highlight the distinction between the effect of a randomized experimental intervention and that of a sustained policy change. Our work addresses this distinction by using a policy-gradient approximation to estimate the global average treatment effect (GATE).

\paragraph{Markovian interference and difference-in-Q estimators.}
Markovian models provide a natural framework for studying treatment effects in dynamic systems with evolving states. \citet{glynn2020adaptive} develop efficient adaptive experimental designs for estimating steady-state reward differences using maximum likelihood. \citet{hays2025double} develop double machine learning methods for shared-state interference, combining observed shared states with nuisance-function estimation and orthogonal scores. Their framework supports inference on average direct effects and, under additional structural assumptions, regression-adjusted estimation of the GATE in switchback experiments. \citet{farias2022markovian} introduce the difference-in-Q's (DQ) estimator, which approximates the global treatment effect through local value-function contrasts and establishes approximation error bounds in terms of the intervention size. \citet{farias2023correcting} further extend this approach through policy-value Taylor expansions and develop a Monte Carlo implementation for experiments at Douyin. \citet{li2026experimenting} exploit queueing structure to improve estimation under stochastic congestion and recover policy gradients from unit-level randomization. Our work is closely related to this literature in using local approximations to estimate global treatment effects.
In particular, we establish the connection between the DQ estimator and the policy gradient under balanced randomization. We further extend this perspective to nonstationary Markovian systems and develop the truncated policy gradient (TPG) estimator, which does not require state observations. Our analysis separates the local approximation error from the truncation error induced by omitted carryover effects.

\paragraph{Switchback experiments and nonstationarity.}
Switchback experiments provide an alternative approach to addressing temporal interference through experimental design. \citet{bojinov2023design} study optimal switchback designs and inference under restrictions on the duration of carryover effects. \citet{hu2022switchback} consider nonstationary Markovian systems with potentially infinite carryover effects and establish bias bounds under uniform mixing conditions.
They further study how burn-in and bias correction can improve estimation. Our setting is closely related to theirs, as we also allow time-varying transition kernels and use uniform mixing conditions to control the influence of past treatments. However, instead of controlling carryover effects through switchback designs, we study estimation under Bernoulli randomization and use truncation to balance bias and variance. Other studies address nonstationarity through different assumptions. \citet{simchi2023non} consider treatment effect estimation and welfare optimization under linear trends, while \citet{wu2024nonstationary} develop temporal stratification and time-grouped randomization to account for changing outcome distributions. Our framework jointly captures nonstationarity and treatment-induced changes in future states.

\paragraph{Off-policy evaluation and partial observability.}
Our work is also related to off-policy evaluation (OPE), which estimates the value of a target policy using data collected under a different policy. \citet{uehara2022review} provide a comprehensive review of OPE methods and their statistical properties. \citet{thomas2016data} combine model-based estimation with importance weighting, while \citet{su2020adaptive} develop adaptive procedures for selecting estimators with different bias--variance trade-offs. \citet{shi2023dynamic} study dynamic treatment effects and sequential testing using observations from a single evolving system.
Under partial observability, \citet{hu2023off} propose truncating importance weights to recent history and establish a trade-off between truncation bias and policy overlap. \citet{uehara2023future} develop future-dependent value functions to address latent states using informative proxies. Our estimator similarly uses truncation to reduce dependence on distant treatment assignments, but differs from conventional OPE methods by estimating a local policy gradient to approximate the finite-horizon GATE. Consequently, our error analysis accounts for both the approximation of the global treatment effect by a local gradient and the truncation of long-term carryover effects.

\paragraph{Policy gradients and dynamic causal effects.}
Recent work further investigates the connection between policy gradients and causal effect estimation. \citet{lai2026estimating} study the identification and doubly robust estimation of dynamic marginal policy effects under sequential unconfoundedness, without requiring full state observability. \citet{qiu2026long} apply this perspective to study how service reliability affects long-term outcomes. \citet{cao2026structured} develop orthogonal estimation methods for Q-function contrasts and state abstractions that preserve decision-relevant information. These studies demonstrate the usefulness of estimating local policy responses in dynamic systems. Our work focuses on randomized experiments and develops a truncated policy-gradient estimator for approximating the global treatment effect under nonstationarity, with explicit bounds on both approximation and truncation errors.

An earlier version of this work appeared in \cite{pmlr-v306-johari26a}. The present version sharpens the variance bounds, extends the analysis to general Bernoulli allocation, introduces a reward-corrected estimator, and provides a self-contained consistency proof for the HAC variance estimator.

\subsection{Organization}
The remainder of this paper is organized as follows.
Section~\ref{sec:pre} introduces the nonstationary Markovian model, the estimand GATE, and the Bernoulli randomized experimental design. Section~\ref{sec:est_non} presents the TPG estimator and our main bias and variance bounds. Section~\ref{sec:policy_gradient} studies the untruncated policy gradient (PG) estimator and establishes its connection to the policy gradient theorem. Section~\ref{sec:trunc_policy_gradient} develops the truncated policy-gradient interpretation and analyzes the resulting bias--variance trade-off. Section~\ref{sec:var_est} establishes a central limit theorem and a consistent HAC variance estimator under additional regularity conditions. Section~\ref{sec:numerical} evaluates the TPG estimator through a two-state nonstationary MDP simulation and two realistic case studies involving queueing and ride-sharing systems.  Additional proofs and simulation details are provided in the appendices.\footnote{All theoretical results have been formalized in Lean and are available at \url{prove2.me}.}

\section{Preliminaries}\label{sec:pre}

In this section, we introduce the formal framework used throughout the paper. 
Section~\ref{sec:POMDP} presents a nonstationary Markovian model of system dynamics under treatment and control.
In Section~\ref{sec:exp_est}, we define the global average treatment effect (GATE) and a Bernoulli randomized experimental design. 

\subsection{Notation}
We adopt standard Landau notation, writing $f(x) = O\pbrac{g(x)}$ if there exists a constant $C > 0$, independent of $x$, such that $|f(x)| \leq C |g(x)|$ for all sufficiently large $x$. We use $[T]$ to denote the index set $\{1, \ldots, T\}$, and $\Law{Y}$ to denote the probability distribution of a random variable $Y$.

\subsection{A dynamic model of treatment}\label{sec:POMDP}

\paragraph{Time horizon.}  We consider a discrete, finite time horizon indexed by $t = 1, \ldots, T$.

\paragraph{State and action.}  We consider a finite state space $\mc{X}$ and a binary action space $\mc{Z} = \{ 0, 1 \}$.  %

\paragraph{Transition kernel.}  For $x \in \mc{X}$ and $z \in \mc{Z}$, we let $P_z^{t}(x, x')$ denote the transition probability from $x$ to $x'$ if action $z$ is chosen at time $t$, i.e., $\P(X_{t+1} = x' | X_t = x, Z_t = z) = P^{t}_z(x, x')$.

We assume the following {\em mixing time} property on the transition kernels. A similar uniform mixing-time bound was assumed by \cite{farias2022markovian} and \cite{hu2022switchback} in the context of switchback experiments in a Markovian setting. More broadly, mixing time assumptions are common in the literature on bandits and reinforcement learning \citep{van1998learning, even2004experts, hu2023off}, as well as in the study of general Markov processes \citep{doukhan1994mixing}.

\begin{assumption}[Mixing time]\label{assumption:mixing-time}
There exists a constant $\gamma$, $0 < \gamma < 1$, such that, for any time $t\in [T]$, any action $z \in \mc{Z}$, and any pair of distributions $f$ and $f'$ over the state space $\mathcal{X}$, the following holds
\begin{equation*}
\left\| f^{\prime} P_z^{t} - f P_z^{t} \right\|_{\mathrm{TV}} \leq \gamma \left\|f^{\prime}-f\right\|_{\mathrm{TV}},
\end{equation*}
where $ f P_z^{t} $ denotes the distribution on $ \mathcal{X} $ defined by $(f P_z^{t})(x') := \sum_{x \in \mathcal{X}} f(x) P_z^{t}(x, x')$,
and $\|\cdot\|_{\mathrm{TV}}$ denotes the total variation distance. {\color{black}For probability row vectors $f$ and $f'$, we use the convention $\|f-f'\|_{\mathrm{TV}}=\tfrac12\|f-f'\|_1$.}
\end{assumption}

\paragraph{Markov process.}  We denote the distribution of the initial state $X_1$ by $\rho$.  For all subsequent time steps, we assume the state transition dynamics obey the Markov property.  
Formally, for any time $t$ and any sequences $\{X_1 = x_1, \ldots, X_{t-1} = x_{t-1}, X_t = x\}$ of states and $\{Z_1 = z_1, \ldots, Z_{t-1} = z_{t-1}, Z_t = z\}$ of actions, we assume the Markov property holds, i.e.,
\begin{multline*}
\P(X_{t+1} = x' | X_1 = x_1, Z_1 = z_1, \ldots, X_{t-1} = x_{t-1}, Z_{t-1} = z_{t-1}, X_t = x, Z_t = z) \\= \P(X_{t+1} = x' | X_t = x, Z_t = z) = P^{t}_z(x, x').
\end{multline*}
Any sequence of actions 
$\{z_1, \ldots, z_T\} \in \mc{Z}^{T}$, 
together with the initial distribution $X_1 \sim \rho$, completely determines the distribution of the (Markov) process $\{X_1, \ldots, X_T\}$.

\begin{remark}
{\color{black}Assumption~\ref{assumption:mixing-time} implies that the influence of an action at time $t$ on the state at time $t+h$ decays at rate $O(\gamma^{h-1})$ for $h\geq1$.}  {Informally, despite nonstationarity, this property yields systems that gradually ``forget" initial conditions or past actions. Similar nonstationary mixing conditions have been considered in \cite{hu2022switchback,wager2024causal}}.
\end{remark}

\paragraph{Rewards.}  For $x \in \mc{X}$ and $z \in \mc{Z}$, we let $R(\cdot | x, z)$ denote the probability distribution of the \emph{real-valued rewards} obtained when action $z$ is taken in state $x$.  
We assume for simplicity that rewards are almost surely bounded over all states and actions, i.e., 
there exists a constant $M$ such that $R([-M,M] | x, z) = 1$ for all $x \in \mc{X}, z \in Z$. We let $r(x, z) = \int y R(dy | x, z)$ denote the expected reward when action $z$ is used in state $x$.
For any time $t$ and any sequences $\{X_1 = x_1, \ldots, X_{t-1} = x_{t-1}, X_t = x\}$ of states and $\{Z_1 = z_1, \ldots, Z_{t-1} = z_{t-1}, Z_t = z\}$ of actions, we assume the reward obtained at time $t$, denoted $Y_t$, is distributed according to $R(\cdot | x, z)$, and is conditionally independent of all other randomness (including the past history) given $X_t$ and $Z_t$.

\subsection{Estimand and experimental design}\label{sec:exp_est}

In this section, we describe our estimand as well as a Bernoulli randomized experimental design to compare treatment and control.  As we will see, these are naturally modeled through the viewpoint of Markov policies from the theory of Markov decision processes (MDPs).

\paragraph{Markov policies: Treatment, control, and Bernoulli randomization.}  Informally, Markov policies are those that choose actions based only on the current state (and not based on history); they can be deterministic or randomized. Since we consider $\mc{Z} = \{0, 1\}$ (binary treatment), any Markov policy can be specified by the probability it chooses $z = 1$ in each state $x$.  

In this paper, we consider a set of Markov policies parameterized by $\theta \in [0,1]$.  Given $\theta$, the policy $\pi_\theta$ chooses $z = 1$ (resp., $z = 0$) with probability $\theta$ (resp., $1 - \theta$) in every state $x$, independent of all other randomness.  Note that $\pi_1$ is the policy that always chooses $z = 1$, which we refer to as the {\em treatment} policy; analogously, $\pi_0$ is the policy that always chooses $z = 0$, which we refer to as the {\em control} policy.  For any $\theta$ such that $0 < \theta < 1$, the policy assigns treatment in each time step according to an i.i.d.~Bernoulli($\theta$) random variable, so we refer to such a $\pi_\theta$ as a {\em Bernoulli randomization} policy.

We denote $P_{\theta}^{t}$ as the transition kernel under the Bernoulli randomization policy $\pi_{\theta}$ at time $t$, i.e., $P_{\theta}^{t} = \theta P_{1}^{t} + (1-\theta)P_{0}^{t}$.
Similarly, we denote $r_\theta(x)$ as the expected reward in state $x$ under policy $\pi_\theta$, i.e., $r_{\theta}(x) = \theta r(x, 1) + (1-\theta)r(x, 0)$, and $\mathcal{L}_\theta^t$ as the marginal distribution of the outcome $Y_t$ under $\pi_\theta$, i.e.:
\begin{equation}
\label{eq:law}
    \mathcal{L}_\theta^t=\operatorname{Law}\left(Y_t \mid X_1 \sim \rho, \, Z_u \sim \text{i.i.d.~Bernoulli}(\theta)  \text { for all } 1 \leq u \leq t\right).
\end{equation}
We write $\E_\theta$ for expectation under the full experimental trajectory law induced by $\rho$ and $\pi_\theta$. Conditional expectations of future outcomes are taken under the same law.

\begin{remark}
    Together with the specification above, given an initial distribution for the state $X_1$, every policy $\pi_\theta$ induces a Markov process as follows: at each time $t = 1, \ldots, T$, the action is chosen according to $Z_t \sim \textrm{Bernoulli}(\theta)$, and the next state is then chosen according to $X_{t+1} \sim P_{Z_t}^{t}(X_t, \cdot)$.  Thus our specification above is a special case of a {\em nonstationary Markov decision process} (MDP), with respect to the restricted set of policies defined above; since the transition dynamics vary arbitrarily over time, in general, no stationary distribution will exist for the resulting process (see, e.g., \cite{lecarpentier2019non, cheung2020reinforcement, hu2022switchback}).
\end{remark} 

\paragraph{Estimand: Global Average Treatment Effect (GATE).}
Given an initial state distribution $X_1 \sim \rho$, recall from \eqref{eq:law} that $\mc{L}_1^t$ (resp., $\mc{L}_0^t$) represents the marginal distribution of the outcome $Y_t$ at time $t$ under the treatment policy $\pi_1$ (resp., under the control policy $\pi_0$).  
With these laws in hand, the treatment effect at time $t$, denoted by $\tau_t$, is defined as $\tau_t := \E_{\mathcal{L}_{1}^t}\sbrac{Y_t} - \E_{\mathcal{L}_{0}^t}\sbrac{Y_t}$, and the global average treatment effect (GATE) $\tau$ over the time horizon $[T]$ is then defined as
\begin{equation*}
\tau := \frac{1}{T} \sum_{t=1}^T \tau_t.
\end{equation*}

\begin{remark}
    The estimand $\tau$ represents the average treatment effect over the specific time horizon $[T]$ under the initial state distribution $\rho$. 
    In general, since the underlying process follows a nonstationary MDP, $\tau$ {\color{black}need not converge to a stationary value.}
    Nevertheless, estimating $\tau$ remains meaningful, as it provides insight into the treatment effect over the experimental horizon---an interval that is typically of intrinsic interest and may be repeated in future deployments (see, e.g., \cite{hu2022switchback,wu2024nonstationary}).
\end{remark}

\paragraph{Data generation and estimation.}
The experiment starts from an initial state $X_1 \sim \rho$, evolves under a uniform Bernoulli randomization policy $\pi_{1/2}$ over the time horizon $[T]$, and generates a single data trajectory $\cbrac{(X_t, Z_t, Y_t): 1 \leq t \leq T}$. An {\em estimator} (computed at time $T$) is then any random variable adapted to $\sigma\pbrac{(X_t, Z_t, Y_t): 1 \leq t \leq T}$.

\section{Estimation under nonstationarity}\label{sec:est_non}

A naive approach to estimating $\tau$ under a Bernoulli randomization policy $\pi_{1/2}$ is to use the following difference-in-means (DM) estimator: 
\begin{equation}\label{eq:IPW_DM}
    \widehat \tau_0 := \frac{1}{T}\sum_{u=1}^T\left(\frac{1\left\{Z_u=1\right\}}{1/2}-\frac{1\left\{Z_u=0\right\}}{1/2}\right) Y_u. 
\end{equation}
The DM estimator is unbiased and exhibits low variance when the observations are i.i.d. \citep{horvitz1952generalization,robins1994estimation}. However, in dynamic (e.g., Markovian) environments, the DM estimator has been shown to incur substantial bias due to temporal {\em interference} across observations \citep{glynn2020adaptive,farias2022markovian,hu2022switchback}. Consequently, a variety of estimators have been developed to correct for temporal interference in dynamic settings, including the recently proposed difference-in-Q's (DQ) estimators \citep{farias2022markovian,farias2023correcting} and estimators based on off-policy evaluation (OPE) techniques \citep{shi2023dynamic,uehara2022review}. Despite these advances, most bias-corrected estimators still rely on stationarity assumptions and can fail in nonstationary environments, where bias correction is inherently more challenging.

\paragraph{Truncated Policy Gradient (TPG).}
To address the challenge of estimation under nonstationarity, we propose an estimator that naturally extends the naive DM estimator and, perhaps surprisingly, mitigates temporal interference bias even in nonstationary Markovian environments. Our approach is straightforward to implement: instead of adding up only the immediate outcome $Y_t$ at each time $t$ as in \eqref{eq:IPW_DM}, we sum the subsequent outcomes over a window of size $k$. Formally, for each $0 \leq k \leq T$, we define the estimator:
\begin{equation}\label{eq:trunc_IPW_est}
    \widehat \tau_k := \frac{1}{T}\sum_{u=1}^T\left(\frac{\mathbf{1}\left\{Z_u=1\right\}}{1/2}-\frac{\mathbf{1}\left\{Z_u=0\right\}}{1/2}\right)\sum_{t=u}^{\min (u+k, T)}  Y_t. 
\end{equation}
Note that when $k = 0$, the estimator $\widehat \tau_k$ reduces to the DM estimator, which completely ignores the nonstationary Markovian dynamics of the environment. In contrast, when $k \geq 1$, the estimator $\widehat \tau_k$ captures the Markovian dynamics by incorporating short-term future outcomes into the estimation.  We refer to $\widehat\tau_k$ as the \emph{truncated policy gradient estimator} with truncation size $k$; later in Section~\ref{sec:trunc_policy_gradient}, we show that $\E[\widehat\tau_k]$ represents the gradient of a hypothetical truncated policy.

\begin{remark}
    The TPG estimator $\widehat{\tau}_k$ requires only observed treatment assignments and rewards, without access to the underlying system states. This offers a practical advantage over methods that rely on state observability, such as many off-policy evaluation (OPE) estimators \citep{uehara2022review}.
    Moreover, when there is no interference, $\widehat \tau_k$ provides an unbiased estimation of $\tau$ under any truncation size $k$ (reduces to DM as $k=0$).
\end{remark}

The main results of this paper establish bias and variance bounds for the TPG estimator $\widehat{\tau}_k$ in nonstationary Markovian environments, as formalized in Theorem~\ref{them:trunc_DQ_bias}.

\begin{figure}[!t]
\centering
\includegraphics[width=.8\textwidth]{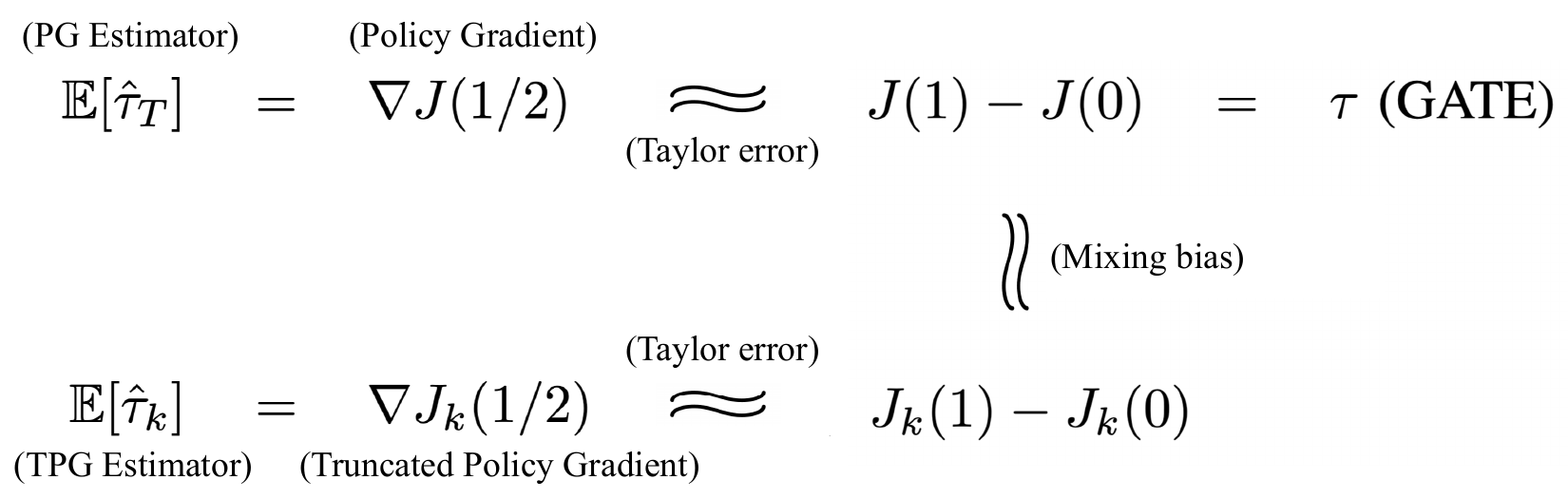}
\caption{Connections between the TPG estimator, PG estimator, policy gradients, and the GATE.\label{fig:figure3}}
\end{figure}

\begin{theorem}[Bias and variance bounds for the TPG estimator]\label{them:trunc_DQ_bias}
Under Assumption~\ref{assumption:mixing-time}, the bias of the TPG estimator $\widehat \tau_k$ with respect to the GATE $\tau$, for any truncation size $0\leq k \leq T$, is bounded by
\begin{equation}\label{eq:bias_bound}
 \abs{\mathbb{E}\sbrac{\widehat \tau_k} - \tau} = O\pbrac{\frac{1 - \gamma^k}{(1-\gamma)^2}\,\delta^2M + \frac{\gamma^k}{1-\gamma}\,\delta M},
\end{equation}
where $\delta$ is the {\em kernel deviation}, i.e., the maximum total variation distance between the transition kernels across all states and times:
$
    \delta:=\sup _{ 1 \leq t \leq T }\sup _{x\in \mathcal{X}}\|P_{1}^t(x, \cdot) - P_{0}^t(x, \cdot)\|_{\mathrm{TV}}.
$
The variance of $\widehat\tau_k$ is bounded by
\begin{equation}
    \operatorname{Var}(\widehat\tau_k)
= O\!\left(\frac{(k+1)^2M^2}{T} + \frac{\gamma (k+1)M^2}{T(1-\gamma)}\right). \label{eq:var_bound}
\end{equation}
\end{theorem}
For any fixed $\gamma<1$, the variance bound is $O\left(M^2(k+1)^2/T\right)$. 
{\color{black}
In particular, the proof in Appendix~\ref{app:thm} provides a constant-reward example that attains this order even in the absence of interference.}

Theorem~\ref{them:trunc_DQ_bias} shows that the bias and variance of $\widehat{\tau}_k$ are jointly governed by the truncation size $k$, the mixing rate $\gamma$, and the kernel deviation $\delta$. In particular, since the DM estimator corresponds to $k = 0$, the theorem demonstrates that the bias reduction benefit of the TPG estimator with $k > 0$ depends on the relationship between $\delta$, $\gamma$, and $k$.  
{\color{black}The bound shows that the ratio $\delta/(1-\gamma)$ governs the cumulative effect of treatment-induced changes in the transition dynamics over the mixing horizon. When this quantity is small, the second-order Taylor term is small relative to the natural first-order scale of the treatment effect on the dynamics. However, since these results provide only upper bounds, they cannot be directly used to optimize the estimator.}

{We derive this theorem over the next two sections.  We start in Section \ref{sec:policy_gradient} by first introducing a broader class of (untruncated) policy gradient (PG) estimators in the nonstationary environment described in Section \ref{sec:pre}. 
In particular, we show that the recent DQ estimator \citep{farias2022markovian} is in fact a PG estimator evaluated at $\pi_{1/2}$ under stationary MDPs.
The policy gradient can be viewed as a first-order Taylor approximation to the GATE (i.e., the difference in policy value for $\pi_1$ and $\pi_0$).  In this way, we obtain a bias bound for the untruncated PG estimator in the nonstationary setting that (1) is dominated by the Taylor error; and (2) unfortunately, scales quadratically in $T$ due to nonstationarity.}

Our analysis of the TPG estimator is then carried out in a parallel manner, and shows that this poor bias behavior can be mitigated. We show in Section \ref{sec:trunc_policy_gradient} that the TPG estimator can be interpreted as an estimator of a \emph{hypothetical truncated policy gradient}, which estimates the difference between two policies: first playing $\pi_{1/2}$, and then playing $\pi_1$ or $\pi_0$ respectively in the most recent $k$ time periods.  Because we approximate this difference using the hypothetical truncated policy gradient, a Taylor approximation error is again incurred, similar to the untruncated version; this is the first term in the bias bound \eqref{eq:bias_bound} in the theorem, which we refer to as the {\em Taylor error}. The second term in the bias bound arises by bounding the gap between the hypothetical truncated policy difference and the GATE, taking advantage of the mixing time assumption Assumption \ref{assumption:mixing-time}; thus, we refer to the second term as the {\em mixing bias}.  Crucially, we observe that due to truncation, the estimator $\widehat{\tau}_k$ has a bias that scales quadratically in $k$, rather than $T$ as for the untruncated PG estimator. Figure~\ref{fig:figure3} shows the connections among the TPG estimator, PG estimator, policy gradients, and the GATE.

\section{{The policy gradient estimator under nonstationarity}}
\label{sec:policy_gradient}

As a starting point in our analysis, we consider the untruncated policy gradient (PG) estimator, i.e., $k = T$.  We show how this estimator can be naturally interpreted as an estimation of the policy gradient for the Bernoulli randomization policy, and use this interpretation to quantify its bias via a Taylor approximation argument.  We start with some preliminary definitions.

\paragraph{Value function.} 
Define the {\em value function} under the Bernoulli randomized policy $\pi_\theta$ as the average outcome over the horizon $[T]$, i.e., in our setting with nonstationary Markovian environments,
\begin{equation}\label{eq:policy_value_func}
    J(\theta) := \frac{1}{T}\sum_{t=1}^{T} \E_{\mathcal{L}_{\theta}^t}\sbrac{Y_t}.
\end{equation}
Note with this definition that the GATE $\tau$ is equal to $J(1) - J(0)$.
We then define the $Q$-value of taking action $z$ at time $t$ under the experimental policy $\pi_{1/2}$\footnote{The policy-gradient results in this section also hold for experimental policies with $\theta\neq 1/2$, with the detailed extension provided in Appendix~\ref{app:general-randomization}.}
as the expected cumulative reward from time $t$ to $T$.
Formally, we write:
\begin{equation}\label{eq:$Q$-values}
Q_{{1/2}}^{t}(z) 
:=\sum_{u=t}^{T}\mathbb{E}_{\textcolor{black}{1/2}}\left[Y_u \mid Z_t = z\right].
\end{equation}
Note that this definition of the $Q$-value in nonstationarity is \emph{state-independent}. Unlike the standard $Q$-value in infinite-horizon stationary MDPs that depends on the state-action pair, our definition reflects the expected value from taking action $z$ at time $t$, conditioned on the initial state distribution $X_1 \sim \rho$.

\begin{prop}\label{prop:policy_gradient}
    The policy gradient of $ J(\theta)$ in \eqref{eq:policy_value_func} evaluated at the uniform randomized policy $\pi_{1/2}$ exists and is given by the average difference in $Q$-values defined in \eqref{eq:$Q$-values} over the horizon, i.e.,
$$
        \nabla J(1/2) = \frac{1}{T}\sum_{t=1}^T\left(Q_{1 / 2}^t(1)-Q_{1 / 2}^t(0)\right).
$$
\end{prop}
Proposition~\ref{prop:policy_gradient} shows that the policy gradient $\nabla J(1/2)$ corresponds to the average $Q$-value difference over the horizon $[T]$.\footnote{In stationary MDPs, $J(\theta)$ is known to be smooth under both average-reward and discounted-reward formulations (see \cite{sutton2018reinforcement}); our theorem shows it remains smooth for finite-horizon non-stationary MDPs.} In fact, the difference-in-Q's (DQ) estimator proposed in \cite{farias2022markovian} can be interpreted as a policy gradient in stationary average-reward MDPs (see Appendix~\ref{appedix:dq} for a detailed discussion).\footnote{\color{black}
This policy-gradient interpretation also extends to any Bernoulli randomization with $\theta\in(0,1)$, with the corresponding generalized bias and variance bounds provided in Appendix~\ref{app:general-randomization}. In Appendix~\ref{app:reward-correction}, we further introduce a reward-corrected (RC) estimator that applies an additional importance weight to the accumulated future rewards. This weighting keeps the future state dynamics under the experimental policy $\pi_\theta$, but changes the accumulated future reward from $r_\theta$ to $r_{1-\theta}$. The reward $r_{1-\theta}$ arises from a first-order approximation of the treatment and control state distributions around $\theta$ while keeping the total-treatment and total-control rewards $r_1$ and $r_0$ unchanged.
The RC estimator coincides with TPG at $\theta=1/2$. In the stationary, untruncated limit, its expectation recovers the general-allocation DQ expansion \citep{farias2022markovian,farias2023correcting} and connects to related work on network interference \citep{peng2025differences}.
}

Along these lines, the key step in our analysis of the untruncated PG estimator $\widehat{\tau}_T$ is to view it as an estimator for the policy gradient $\nabla J(1/2)$ in the nonstationary setting.
Applying a Taylor expansion around $J(1/2)$, the bias of $\widehat{\tau}_T$ relative to $\tau$ decomposes as
\begin{align*}
    \mathbb{E}\left[\widehat{\tau}_T\right] - \tau 
    &= \mathbb{E}\left[\widehat{\tau}_T\right] - \nabla J(1/2) + \nabla J(1/2) - (J(1) - J(0)) \notag\\
    &=  \underbrace{\left( \mathbb{E}[\widehat{\tau}_{T}] - \nabla J(1/2) \right)}_{\substack{\text{gradient estimation bias} \\ = 0}}
     + \underbrace{\sum_{n=2}^{\infty} \frac{1}{n!} \nabla^n J(1/2)\left[(-1/2)^n - (1/2)^n\right]}_{\text{Taylor error}},
\end{align*}
where the Taylor error corresponds to the higher-order bias in approximating $\tau$ using the policy gradient $\nabla J(1/2)$. 
Now observe by elementary calculation that the untruncated PG estimator is {\em unbiased} for the policy gradient $\nabla J(1/2)$, i.e., $\mathbb{E}[\widehat{\tau}_T] = \nabla J(1/2)$.  Therefore, the entire bias is due to the Taylor error. {\color{black}The following proposition establishes bias and variance bounds for the untruncated PG estimator without imposing Assumption~\ref{assumption:mixing-time}.}

\begin{prop}
\label{them:DQ_bias}
Without imposing Assumption~\ref{assumption:mixing-time}, the bias of
$\widehat{\tau}_T$ with respect to the GATE $\tau$ is bounded by
$$
        \abs{\mathbb{E}\sbrac{\widehat \tau_T} - \tau}
        =
        O\pbrac{T^2 \delta^2 M}.
$$
In fact, the preceding bound applies for any unbiased estimator $\widehat{\tau}$
of $\nabla J(1/2)=J'(1/2)$.
Moreover, the variance of $\widehat{\tau}_T$ satisfies
\begin{equation}
    \operatorname{Var}(\widehat{\tau}_T)\leq 2M^2(T+1)=O(TM^2).
    \label{eq:untrunc_variance_bound}
\end{equation}
\end{prop}
Proposition~\ref{them:DQ_bias} can be viewed as the untruncated counterpart of Theorem~\ref{them:trunc_DQ_bias} in the absence of Assumption~\ref{assumption:mixing-time}.
{\color{black}
Proposition~\ref{them:DQ_bias} and Theorem~\ref{them:trunc_DQ_bias} clarify the roles of mixing and truncation. In particular, the mixing condition controls the bias uniformly in $T$, while the variance of the untruncated estimator can still grow at rate $O(TM^2)$, even without interference (see the variance proof of Theorem~\ref{them:trunc_DQ_bias}). 
For a fixed mixing rate $\gamma<1$, truncation additionally reduces the variance to order $M^2(k+1)^2/T$ at the cost of an additional mixing bias that decays geometrically in $k$. The numerical results in Section~\ref{sec:numerical} illustrate this bias--variance trade-off.
}

\section{Estimation via the truncated policy gradient estimator}\label{sec:trunc_policy_gradient}

Proposition~\ref{them:DQ_bias} gives bias and variance bounds for the untruncated PG estimator (i.e., $k=T$) that grow quadratically and linearly in $T$ without a mixing-time condition.
In this section, we analyze the TPG estimator \eqref{eq:trunc_IPW_est} for general $k$ and show how exactly the mixing-time condition in Assumption~\ref{assumption:mixing-time} leads to more favorable bias and variance bounds.

We start by introducing the hypothetical truncated policy gradient in nonstationary MDPs and establish its connection to the TPG estimator $\widehat\tau_k$. 

\paragraph{Hypothetical truncated policy.}
{\color{black}For $a_t:=\max\{1,t-k\}$, let $\mathcal{L}_\theta^{t,k}$ denote the distribution of $Y_t$ under a hypothetical policy that follows the experimental policy $\pi_{1/2}$ through time $a_t-1$ and switches to $\pi_\theta$ at time $a_t$, so that assignments from $a_t$ through $t$ are generated under $\pi_\theta$, starting from $X_1\sim\rho$. When $a_t=1$, the switch occurs at the first time period, so the policy follows $\pi_\theta$ from time $1$ through $t$.}
$$
\mathcal{L}_\theta^{t,k}
= \operatorname{Law}\left(
    Y_t \,\middle|\,
    \begin{array}{l}
        X_1 \sim \rho, \\[4pt]
        \{Z_u\}_{u=1}^{\textcolor{black}{a_t-1}} \sim \text{i.i.d.\ Bernoulli}(1/2), \\[4pt]
        \{Z_u\}_{u=\textcolor{black}{a_t}}^{t} \sim \text{i.i.d.\ Bernoulli}(\theta)
    \end{array}
\right).
$$
The {\em truncated policy value function} $J_k(\theta)$ is defined as
\begin{equation}\label{eq:trunc_value_func}
    J_k(\theta) 
    := \frac{1}{T}{\sum_{t=1}^{T} \E_{\mathcal{L}_\theta^{t, k}}\sbrac{Y_t}}.
\end{equation}
We refer to the gradient of $J_k(\theta)$ as the \emph{truncated policy gradient}. 
For truncation size $0 \leq k \leq T$,  define the \emph{truncated $Q$-value} as
\begin{equation}\label{eq:trunc_Q_values}
Q_{{1/2}}^{t, k}(z) 
:=\sum_{u=t}^{\min(t+k,T)}\mathbb{E}_{\textcolor{black}{1/2}}\left[Y_u \mid Z_t = z\right].
\end{equation}

\begin{prop}\label{prop:trunc_policy_gradient}
    The truncated policy gradient $\nabla J_k(\theta)$ evaluated at the uniform randomized policy $\pi_{1/2}$ exists and is given by the average difference in truncated $Q$-values in \eqref{eq:trunc_Q_values} over the horizon, i.e.,
    \begin{equation*}
        \nabla J_k(1/2) = \frac{1}{T}\sum_{t=1}^T\left(Q_{1 / 2}^{t,k}(1)-Q_{1 / 2}^{t,k}(0)\right).
    \end{equation*}
\end{prop}
With the TPG estimator $\widehat \tau_k$ defined in \eqref{eq:trunc_IPW_est}, we have
$$
    \mathbb{E}\sbrac{\widehat \tau_k} =  \frac{1}{T}\sum_{t=1}^T \left(Q_{{1/2}}^{t, k}(1) - Q_{{1/2}}^{t, k}(0) \right) = \nabla J_k(1/2).
$$
This indicates that the TPG estimator $\widehat{\tau}_k$ can be viewed as an estimation of $\tau$ via an unbiased estimation of the truncated policy gradient $\nabla J_k(1/2)$. Accordingly, the bias of $\widehat{\tau}_k$ relative to $\tau$ decomposes as 
\begin{equation*}
    \mathbb{E}\left[\widehat \tau_k\right] - \tau = \underbrace{\nabla J_k(1/2) - (J_k(1) - J_k(0))}_{\text{{Taylor error w.r.t. }} J_k(1/2)} + \underbrace{(J_k(1) - J(1)) + (J(0) - J_k(0))}_{\text{mixing bias}},
\end{equation*}
which includes a {Taylor error} between $\nabla J_k(1/2)$ and $J_k(1) - J_k(0)$, and an additional mixing bias between the original value function $J(\cdot)$ and its truncated counterpart $J_k(\cdot)$.  We note that the {Taylor error} is of order $O(k^2)$ because, at time $t$, the treatment probability affects only the current reward and the most recent $\min\{k,t-1\}$ transitions, rather than the full trajectory (cf.~the $O(T^2)$ bias of $\widehat{\tau_T}$).

Figure~\ref{fig:figure3} illustrates the connections among the GATE, the policy gradients $\nabla J(1/2)$ and $\nabla J_k(1/2)$, and the corresponding untruncated PG estimator $\widehat{\tau}_T$ and TPG estimator $\widehat{\tau}_k$. The bias of the TPG estimator $\widehat \tau_k$ involves an additional mixing bias but tightens the worst-case Taylor error.
{\color{black}
The variance of $\widehat\tau_k$ can be bounded by combining Assumption~\ref{assumption:mixing-time} with a strong $\alpha$-mixing argument \citep{rosenblatt1956central}. For fixed $k$, each term depends only on a local reward window, while the dependence between sufficiently separated windows decays geometrically. Therefore, for fixed $k$ and $\gamma<1$, the variance of $\widehat\tau_k$ is of order $O(T^{-1})$, in contrast to the untruncated estimator whose variance can grow with $T$.
}

\paragraph{Choice of size $k$.}
Post-calibration methods such as Lepski's method \citep{lepski1997optimal, su2020adaptive} can select $k$ in a data-driven way using the variance estimates in Section~\ref{sec:var_est}, but should carefully handle the non-monotone bias bound.
{\color{black}As a practical guideline, we suggest starting with the standard DM estimator ($k=0$) and increasing $k$ until the change in the estimate is small relative to its estimated standard error.} 
Our numerical results show that, in most cases, increasing $k$ corrects the bias of the DM estimator and moves the estimate closer to the ground-truth GATE, thereby revealing the bias direction. A practical heuristic $k$-selection algorithm is provided in Appendix~\ref{appendix:algo}.

\section{Nonstationary CLT and variance estimation}\label{sec:var_est}

In this section, we establish the asymptotic normality of the TPG estimator $\widehat\tau_k$ for any finite truncation size $k$ and present a consistent nonstationary estimator of its asymptotic variance $\sigma_k^2$.
Define $w_t^{\mathrm{IPW}} := 2(2Z_t - 1)$, $B_t := \sum_{i=\max\{1,\, t-k\}}^t w_i^{\mathrm{IPW}}Y_t$, and $\bar B := T^{-1}\sum_{t=1}^T B_t$.
For any finite $k$, we consider the following assumptions to regularize the nonstationarity:
\begin{assumption}[Asymptotic variance]\label{ass:asymptotic_variance}
We assume that the following limit exists and is finite and positive:
\[
\sigma_k^2
:=
\lim_{T\to\infty}
\frac{1}{T}\,
\Var\!\left(\sum_{t=1}^T B_t\right)
\in (0,\infty).
\]
\end{assumption}

\begin{assumption}[Ces\`aro-$L^2$ mean stability]\label{ass:mean-drift}
Define $\mu_t := \E[B_t]$ and the averaged mean $\bar\mu_T := T^{-1}\sum_{t=1}^T \mu_t$.  
Assume that the Ces\`aro-$L^2$ deviation satisfies
$$
\Delta_T^2 := \frac{1}{T}\sum_{t=1}^T \bigl(\mu_t - \bar\mu_T\bigr)^2 = o(T^{-2/3}).
$$
\end{assumption}
Assumptions~\ref{ass:asymptotic_variance} and~\ref{ass:mean-drift} impose regularity conditions on the system's nonstationarity. Assumption~\ref{ass:asymptotic_variance} requires the normalized population variance to converge to a positive limit. Assumption~\ref{ass:mean-drift} requires the time-specific means of $B_t$ to approach their horizon average at the stated rate. 
Under these conditions, Theorem~\ref{thm:truncDQ_CLT_HAC_concise} establishes the CLT for the TPG estimator and presents a Newey--West (NW) \emph{nonstationary heteroskedasticity- and autocorrelation-consistent} (HAC) variance estimator. 

\begin{theorem}\label{thm:truncDQ_CLT_HAC_concise}
Fix a finite truncation size $k \geq 0$. Under Assumptions~\ref{assumption:mixing-time}, \ref{ass:asymptotic_variance}, and \ref{ass:mean-drift}, define $\widehat V_t := B_t - \bar B$,

$$
\widehat\Gamma_\ell := \frac{1}{T}\sum_{t=1}^{T-\ell} \widehat V_t \widehat V_{t+\ell}, 
\quad
w^{\mathrm{NW}}_\ell := 1-\frac{\ell}{L_T+1},
$$
{\color{black}where $L_T$ is a deterministic integer sequence satisfying $L_T\to\infty$ and $L_T=O(T^{1/3})$, with $0\leq L_T<T$ (e.g., $L_T=\min\{T-1,\lfloor T^{1/3}\rfloor\}$).}
We construct the nonstationary HAC variance estimator as

$$
\widehat\Omega_T := \widehat\Gamma_0 + 2\sum_{\ell=1}^{L_T} w^{\mathrm{NW}}_\ell\,\widehat\Gamma_\ell.
$$
Then, as $T \to \infty$,

$$
\sqrt{T}\,(\widehat\tau_k - \E[\widehat\tau_k]) \;\Rightarrow\; \mathcal{N}(0,\sigma_k^2),
\quad
\widehat\Omega_T \;\xrightarrow{p}\; \sigma_k^2.
$$
\end{theorem}

\paragraph{Proof sketch.}
The proof of Theorem~\ref{thm:truncDQ_CLT_HAC_concise} has two parts.
First, we establish the CLT by representing the experimental trajectory as a temporally nonhomogeneous Markov chain and applying the CLT of \citet{arlotto2016central}. Assumption~\ref{assumption:mixing-time} gives a uniform lower bound on the corresponding ergodic coefficients, while fixed $k$ and bounded rewards ensure that the local summands are uniformly bounded. Together with Assumption~\ref{ass:asymptotic_variance}, these conditions yield asymptotic normality.
Second, we prove consistency of the HAC variance estimator. We first consider an ideal Bartlett HAC estimator centered by the time-specific means $\E[B_t]$. Using boundedness and geometric mixing, we control its sampling error uniformly over lags and show that its expectation converges to the asymptotic variance. Since the time-specific means cannot be estimated separately from a single trajectory, we then use Assumption~\ref{ass:mean-drift} to show that replacing $\E[B_t]$ by the sample mean $\bar B$ changes the HAC estimator by $o_p(1)$. This yields consistency of the feasible estimator $\widehat\Omega_T$.

{\color{black}
\paragraph{Inference.}
For a fixed truncation size $k$, Theorem~\ref{thm:truncDQ_CLT_HAC_concise} and Slutsky's theorem yields
$$
 \frac{\sqrt T\,\bigl(\widehat\tau_k-\E[\widehat\tau_k]\bigr)}{\sqrt{\widehat\Omega_T}}
 \Rightarrow\mathcal N(0,1).
$$
Therefore, the estimated standard error is $\widehat\sigma_k:=\sqrt{\widehat\Omega_T/T}$.
A Wald interval centered at $\widehat\tau_k$ is asymptotically valid for the GATE only when the bias is also controlled, e.g., when $\sqrt T\bigl(\E[\widehat\tau_k]-\tau\bigr)\to0$, or when the interval explicitly accounts for the bias.
}

\section{Numerical results}\label{sec:numerical}
In this section, we evaluate the performance of our TPG estimator through a series of experiments. These include: (i) a two-state nonstationary Markov decision process, (ii) a queueing simulator calibrated using real-world nonstationary patient arrival data adapted from \cite{li2026experimenting}, and (iii) a large-scale NYC ride-sharing simulator with real data adapted from \cite{peng2025differences}. All code and experiment details are publicly available.\footnote{\url{https://github.com/wenqian-xing/TPG-Estimator}}

\begin{figure}[htbp]
\centering

\begin{minipage}[t]{0.40\textwidth}
  \centering
  \vspace{0pt}
  \begin{tabular}{ccc}
    \toprule
    $k$ & MAE (\%) & STD \\
    \midrule
    0 (DM) & 48.706 & 0.122 \\
    1 & 27.11 & 0.233 \\
    3 & 23.17 & 0.447 \\
    5 & 28.18 & 0.657 \\
    10 & 47.34 & 1.169 \\
    50 & 209.67 & 5.198 \\
    100 & 410.65 & 10.178 \\
    $T$ (PG) & 11610.03 & 288.382 \\
    \bottomrule
  \end{tabular}
  \captionof{table}{\color{black}
    Evaluation of $\widehat\tau_k$ across truncation sizes $k$, averaged over varying uniform mixing bounds $\gamma$. Reported metrics are MAE (mean absolute error, in \%) and STD (standard deviation). Average true ATE:
    $\tau = 2.05$.
  }
  \label{tab:truncDQ}
\end{minipage}
\hfill
\begin{minipage}[t]{0.56\textwidth}
  \centering
  \vspace{0pt}
  \includegraphics[width=\linewidth]{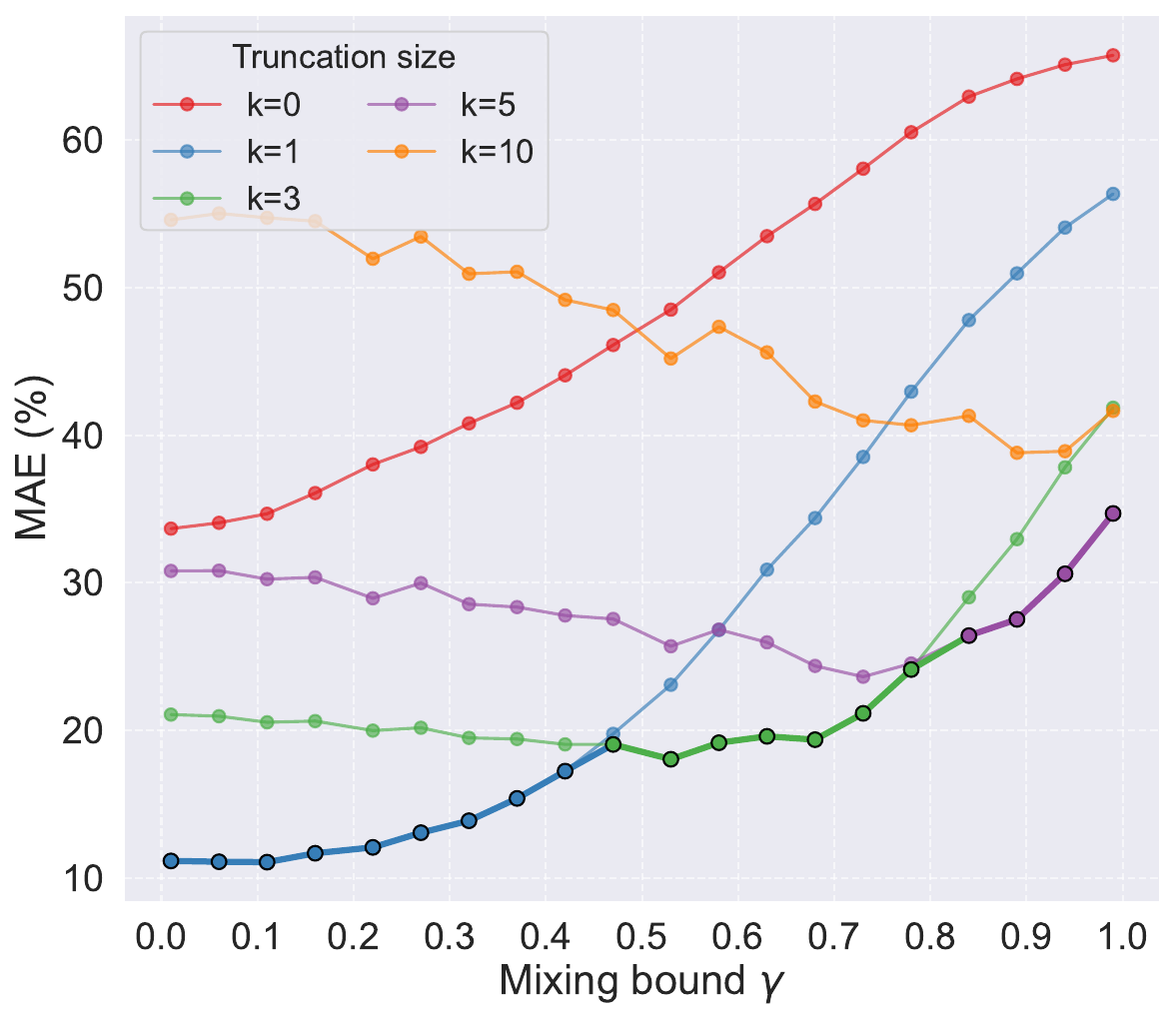}

  \captionof{figure}{\color{black}
    MAE of $\widehat\tau_k$ with varying uniform mixing bounds $\gamma$.
    The lowest-MAE choice of $k$ at each $\gamma$ is bolded.
  }
  \label{fig:truncDQ}
\end{minipage}

\end{figure}

\paragraph{Two-state nonstationary MDP.}

{\color{black}
We simulate a two-state nonstationary MDP to evaluate the bias and variance of the TPG estimator under uniform mixing bounds $\gamma$ ranging in $[0.01,0.99]$.
The transition kernels are generated through a mean-reverting Gaussian process, with clipping and projection to ensure that $\gamma$ is a uniform mixing bound for both treatment and control kernels and their Bernoulli mixtures.
For each of the 20 values of $\gamma$, we simulate 1,000 independent experimental trajectories with a horizon of $T=5000$, holding the kernel sequence fixed across trials. The ground truth GATE $\tau$ is computed exactly from the state distributions under the all-treatment and all-control policies. Further details on the simulation are provided in Appendix~\ref{Appendix:two_state}.

In Table~\ref{tab:truncDQ}, we observe that the DM estimator (i.e., $k=0$) exhibits substantial error, with a mean absolute error (MAE) of $48.71\%$. The untruncated PG estimator ($k=T$) exhibits very large absolute error and variance.
In contrast, the TPG estimator substantially reduces the error even with $k=1$, achieving an MAE of $27.11\%$, while maintaining low variance. Among the reported truncation sizes, $k=3$ achieves the smallest MAE of $23.17\%$. Figure~\ref{fig:truncDQ} further illustrates the MAE across different truncation sizes and mixing bounds. The best-performing choice of $k$ varies with $\gamma$, but remains small relative to $T=5000$, illustrating the trade-off between truncation bias and variance.
}

\paragraph{Case Study I: {Queueing simulator based on real-world nonstationary patient arrivals.}}  %

To test the real-world effectiveness of our estimator, we implement a queueing simulator with nonstationary arrivals using a similar setup to \cite{li2026experimenting}.  The simulator is a queueing birth-death chain with inhomogeneous, state-dependent Poisson arrival rates based on a nonstationary estimate of patient arrival rates to an emergency department (both based on day of week, and time of day), obtained using data from the SEEStat Home Hospital database \citep{seestat_data}; see Figure 9, Section 5.2, and Appendix B.1 of \cite{li2026experimenting} for details.  In the simulator, customers are less likely to join when the state is higher (i.e., the emergency department is busier).  While the arrival rate is state-dependent in the simulator, the service rate is constant in all states.  We implement a discrete-time version of this simulator, with a constant proportional change in the arrival rates between treatment and control (e.g., a change in efficiency of the emergency department). Full details are provided in Appendix \ref{Appendix:real_case}. We run this simulator for 40,320 time steps (corresponding to four weeks of real time with 1-minute-long time steps), with 500 simulated experiment trajectories. The corresponding switchback experiment uses 1-hour intervals.

\begin{figure}[htbp]
\centering

\begin{minipage}[t]{0.40\textwidth}
  \centering
  \vspace{0pt}
  \begin{tabular}{cc}
    \toprule
    $k$ & Coverage (\%) \\
    \midrule
    2 & 2.0 \\
    3 & 44.2 \\
    4 & 93.2 \\
    \textbf{5} & \textbf{94.6} \\
    6 & 85.6 \\
    7 & 78.8 \\
    \bottomrule
  \end{tabular}
  \captionof{table}{Empirical coverage of 95\% CI using the nonstationary HAC variance estimator in Case Study I. The choice $k=5$ is selected by the $k$-selection algorithm in Appendix~\ref{appendix:algo}, which yields near-nominal 95\% coverage.}
  \label{tab:ci_coverage_k}
\end{minipage}
\hfill
\begin{minipage}[t]{0.56\textwidth}
  \centering
  \vspace{0pt}
  \includegraphics[width=0.85\linewidth]{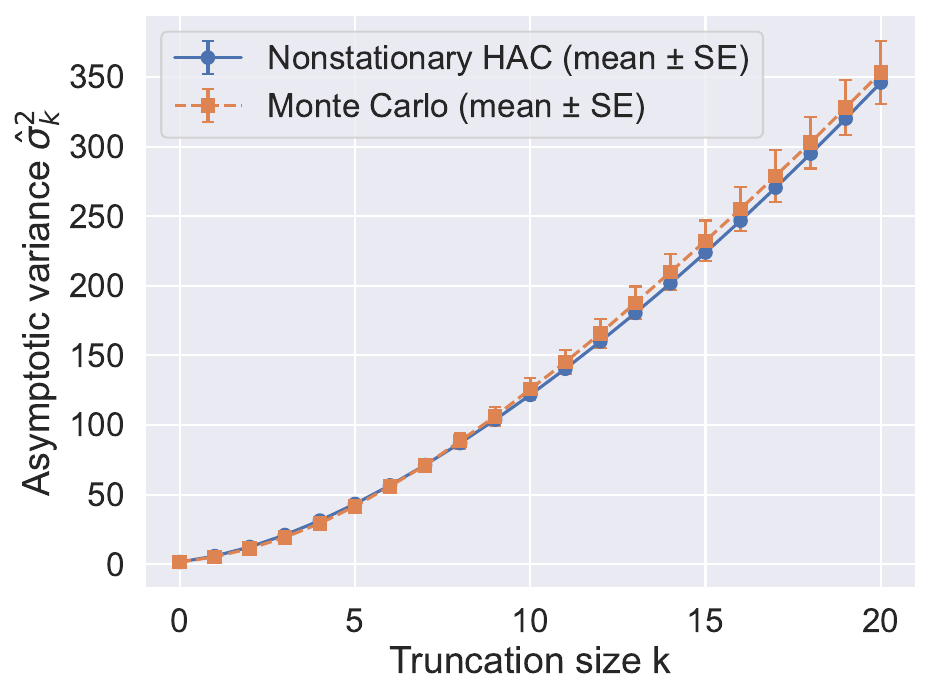}

  \captionof{figure}{Variance estimates for the TPG estimator $\widehat{\tau}_k$ in a realistic patient queueing system in Section~\ref{sec:numerical}.}
  \label{fig:hac_mc_var}
\end{minipage}

\end{figure}

We compare the TPG estimator with the stationary DQ and LSTD-OPE~\citep{farias2022markovian}, and with the bias-corrected (BC) estimator (with burn-in) under the Bernoulli \textit{switchback design}~\citep{hu2022switchback}. See  Figure~\ref{fig:case_study}. LSTD-OPE yields a mean of $-2.39$, outside the plotted range. Both DM and stationary DQ exhibit substantial bias. By contrast, the TPG estimator with $k=4$ or $5$ (chosen by Appendix~\ref{appendix:algo}) closely approximates the GATE. The BC estimator with a 20-minute burn-in also performs well, but it depends on a more complex switchback design that typically requires prior knowledge to choose the interval.
Table~\ref{tab:ci_coverage_k} and Figure~\ref{fig:hac_mc_var} further validate the nonstationary HAC variance estimator in this case study via Monte Carlo simulation: it closely matches the simulated variance, and the resulting 95\% CI attains near-nominal coverage at $k=5$.

\begin{figure}[tbp]
\centering
\includegraphics[width=1\linewidth]{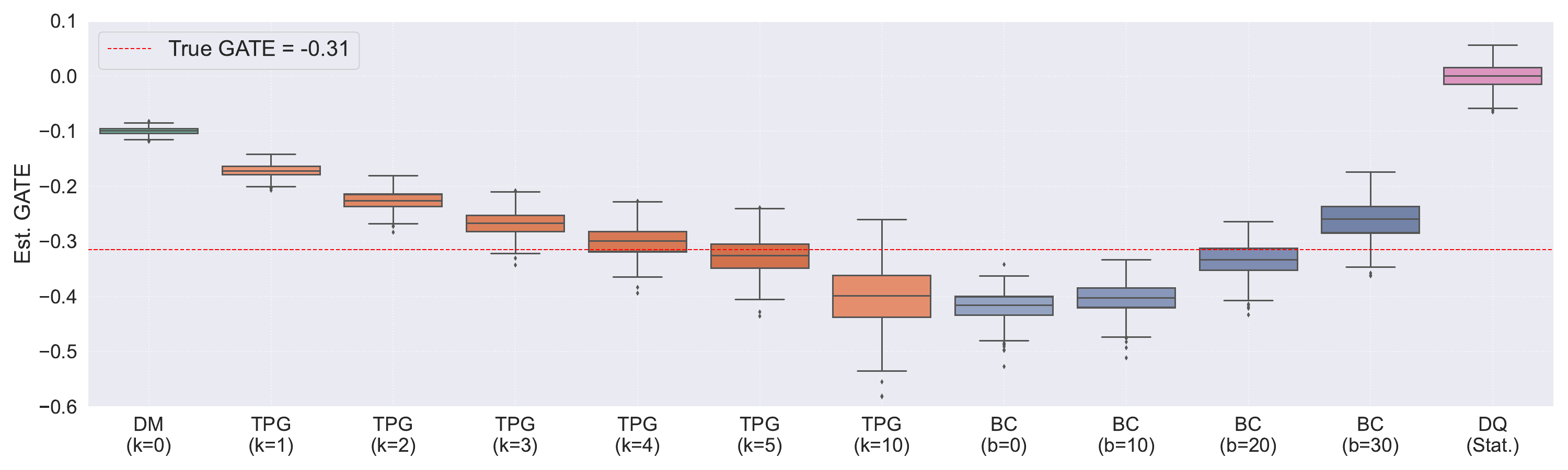}
\caption{Estimation under real-world nonstationary patient arrivals.\label{fig:case_study}}
\end{figure}

\paragraph{Case Study II: NYC ride-sharing simulator.}
We further adapt a large-scale NYC ride-sharing simulator using real-world data from \cite{peng2025differences}. Rider arrivals and destination locations are generated based on historical records provided by New York City's Taxi and Limousine Commission \cite{TLC}, restricted to trips within Manhattan. In this setting, we evaluate the platform's pricing policy, where the fare offered to a rider is computed as a fixed rate per unit of time multiplied by the estimated trip duration from pickup to destination. The treatment policy applies a higher rate per unit of time than the control policy. Upon receiving a price quote, the rider can either accept the trip---triggering a dispatch of the nearest available vehicle and generating a reward equal to the trip fare---or reject it, resulting in no dispatch and zero reward. 
We run the simulator for 500,000 rider arrivals, i.e., approximately 142 hours of simulated time, across 100 simulated experiment trajectories.
\begin{figure}[!tb]
\centering
\includegraphics[width=1\linewidth]{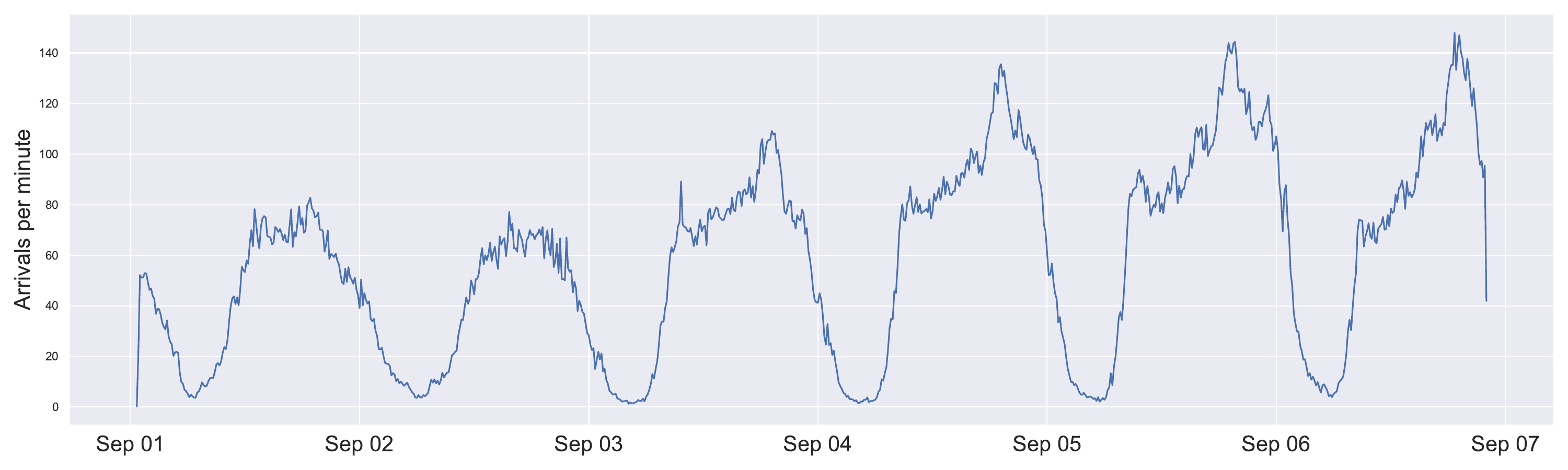}
\caption{Nonstationary arrival rate (per minute) in the NYC ride-sharing dataset.\label{fig:events}}
\end{figure}
As shown in Figure~\ref{fig:events}, rider arrivals in the NYC ride-sharing data from September 2024 exhibit complex, non-stationary patterns \citep{TLC}, presenting a significant challenge for treatment effect estimation.
Further details on the simulator and the rider choice model are provided in Appendix~\ref{Appendix:real_case_NYC}.

\begin{figure}[!tb]
\centering
\includegraphics[width=1\linewidth]{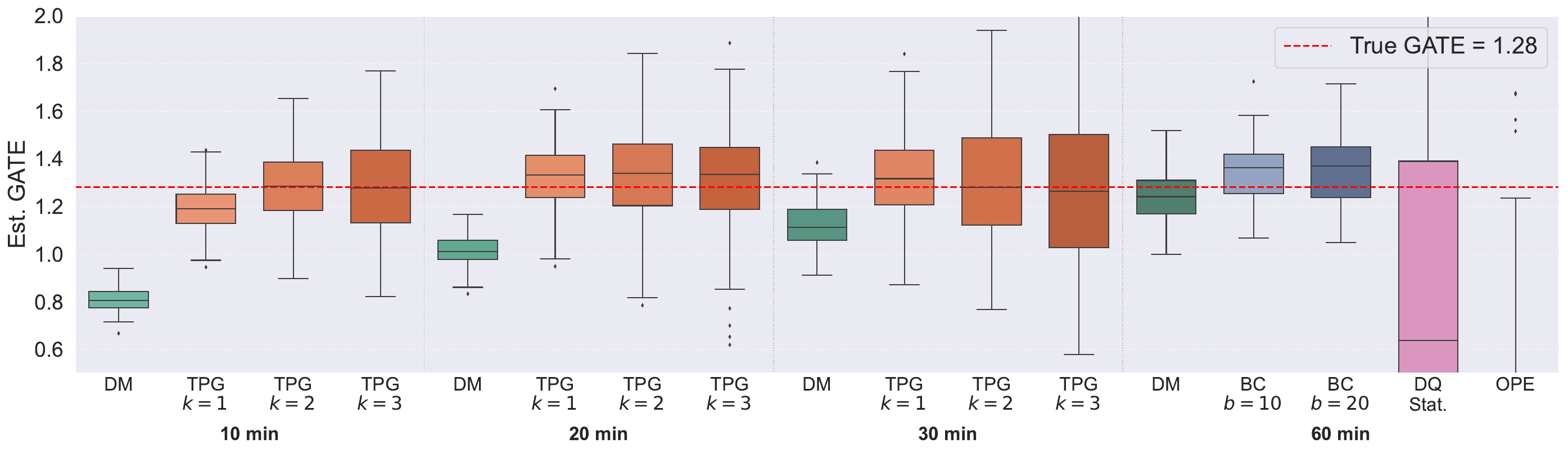}
\caption{Estimation under NYC ride-sharing simulation with real data.\label{fig:case_study_NYC}}
\end{figure}

In this experiment, {mirroring common practice in industry}, we implement a Bernoulli switchback design, in which each fixed time interval (e.g., 10 or 20 minutes) is independently assigned to treatment or control with probability $1/2$. The system state for each interval is defined as the average number of available drivers.
Figure~\ref{fig:case_study_NYC} presents the results of various estimators for this setting. Our first finding is that the DM estimator can perform well but is highly sensitive to the choice of switchback interval. For example, intervals of 10, 20, and 30 minutes all exhibit large bias, while the 60-minute interval performs well. This means the experimenter must identify the appropriate interval in advance. In contrast, the TPG estimators (which are all extensions of DM) consistently reduce bias across a wide range of switchback intervals. This suggests that TPG can serve as an effective postprocessing tool to compensate for imperfect interval selection during the design phase, at the cost of slightly increased variance. We also evaluated several state-of-the-art estimators {for switchback experiments}. Under the 1-hour interval, the BC estimator offers no improvement over the naive DM and shows slight overestimation when the burn-in period $b$ is set to 10 or 20 minutes. Meanwhile, the stationary DQ and OPE estimators struggle to learn reliable stationary $Q$-values, resulting in estimates with both high bias and high variance.
In addition, we can observe the bias--variance trade-off described after Proposition~\ref{them:DQ_bias}, i.e., increasing the truncation size $k$ leads to a significant increase in variance, whereas the mixing bias induced by choosing a smaller $k$ appears to be modest. In particular, in the ride-sharing case study, the TPG estimator with $k=2$ or $k=3$ already produces estimates that are very close to the GATE.

This case study also demonstrates that TPG estimators can be applied not only to the naive Bernoulli A/B design but also to more general designs, such as the Bernoulli switchback design. This provides additional flexibility for using TPG in practical systems. It also raises an important question regarding how to co-design the experiment interval (which controls the mixing rate $\gamma$) and the truncation size 
$k$ together to optimize treatment effect estimation. We leave this as a promising open question for future research.

\section{Conclusion}

In this paper, we introduced the TPG estimator for estimating the GATE from a single trajectory in nonstationary Markovian environments. The estimator is simple to implement, requires only observed assignments and rewards, and does not require knowledge of the underlying state space. In particular, its policy-gradient interpretation provides a direct and intuitive way to understand how truncation balances the bias from ignoring long-run effects against the variance from incorporating increasingly long reward trajectories.

Under a mixing condition, we establish explicit bias and variance bounds for the TPG estimator, showing how its performance depends on the truncation size, mixing rate, and treatment-induced change in the system dynamics.
For a fixed truncation size, we further established a central limit theorem and consistency of a Bartlett HAC variance estimator under additional variance-limit and mean-stability conditions. 
These results enable confidence intervals for the estimator mean.

Our numerical experiments in nonstationary queueing and ride-sharing settings illustrate this bias--variance trade-off. In both studies, we show that relatively short truncation windows can substantially improve over methods that ignore temporal interference or rely on stationarity. We also propose a simple post-experiment heuristic for choosing the truncation size. Overall, these results show how the policy-gradient perspective can lead to practical estimators for experiments with a single dependent trajectory in dynamic, nonstationary systems.

\endgroup

\section*{Acknowledgments}
The authors are especially grateful to Andrew Zheng for his contributions to the ride-sharing simulator, to Shuangning Li for generously sharing details of the nonstationary patient arrival queueing simulator, and to Stefan Wager for insightful discussions on the policy gradient theorem.
We also gratefully acknowledge the SEE Research Team at the Technion---Israel Institute of Technology for providing the queueing data used in our simulations. 
This work was supported in part by Stanford Data Science, the Stanford Human-Centered AI Institute, and the National Science Foundation under Grant 2205084.

\bibliographystyle{informs2014}
\bibliography{references}

\clearpage

\appendix
\begingroup
\renewcommand{\contentsname}{Appendix Contents} %
\tableofcontents
\endgroup
\clearpage
\input{appendix}

\end{document}

%% file: macros.tex
\newenvironment{prop}{\begin{proposition}}{\end{proposition}}

\let\leq\le
\let\geq\ge

\newcommand{\abs}[1]{{\left\lvert {#1} \right\rvert}}
\newcommand{\pbrac}[1]{{\left( {#1} \right)}}
\newcommand{\cbrac}[1]{{\left\lbrace {#1} \right\rbrace}}
\newcommand{\sbrac}[1]{{\left[ {#1} \right]}}

\newcommand{\mc}[1]{\ensuremath{\mathcal{#1}}}

\renewcommand{\P}{\ensuremath{\mathbb{P}}}

\newcommand{\E}{\ensuremath{\mathbb{E}}}

\newcommand{\Var}
{\ensuremath{\operatorname{Var}}}

\newcommand{\Law}[1]{\ensuremath{\mathsf{Law}(#1)}}

%% file: appendix.tex
\section{Notation and Technical Lemmas}\label{appendix:notation}
In this section, we present the key notation along with several technical results and lemmas used throughout the proofs.

\subsection{Time-Varying Transition Kernels}
We denote by $P_{\theta}^{a\rightarrow b}$ the transition kernel induced by the Bernoulli-randomized policy $\pi_{\theta}$ from time $a$ to time $b$ for $1 \leq a \leq b \leq \textcolor{black}{T+1}$, defined as
\begin{equation*}
    P_{\theta}^{a \rightarrow b} := 
    \begin{cases}
        \prod_{u=a}^{b-1} P_{\theta}^u, & \text{if } a < b, \\
        I, & \text{if } a = b,
    \end{cases}
\end{equation*}
where $I$ denotes the identity matrix. Additionally, note that $P_{\theta}^{a \rightarrow a+1} = P_{\theta}^a$ for all $a\in [T]$.

\subsection{Time-Varying Kernel Deviations}
For each time step $t \in [T]$, we define the difference in transition kernels as $\Delta^t := P_1^t - P_0^t$, where $P_1^t$ and $P_0^t$ are the transition kernels under the always-treated and always-controlled policies, respectively. We quantify the maximum deviation across time and states using the worst-case total variation norm:
\begin{equation*}
    \delta^t := \sup_{x \in \mathcal{X}} \left\| P_1^t(x, \cdot) - P_0^t(x, \cdot) \right\|_{\mathrm{TV}}, \quad \delta := \sup_{1 \leq t \leq T} \sup_{x \in \mathcal{X}} \left\| P_1^t(x, \cdot) - P_0^t(x, \cdot) \right\|_{\mathrm{TV}}.
\end{equation*}
Additionally, we define the deviations from the uniform randomization policy $\pi_{1/2}$ at each time $t$ as $\Delta_+^t := P_1^t - P_{1/2}^t$ and $\Delta_-^t := P_0^t - P_{1/2}^t$, where $P_{1/2}^t = (P_1^t + P_0^t)/2$. It follows directly that
\begin{equation*}
    \sup_{1 \leq t \leq T} \sup_{x \in \mathcal{X}} \left\| \Delta_+^t(x, \cdot) \right\|_{\mathrm{TV}} 
    = \sup_{1 \leq t \leq T} \sup_{x \in \mathcal{X}} \left\| \Delta_-^t(x, \cdot) \right\|_{\mathrm{TV}} 
    = \sup_{1 \leq t \leq T} \sup_{x \in \mathcal{X}} \frac{1}{2}\left\| \Delta^t(x, \cdot) \right\|_{\mathrm{TV}} \,\leq\, \delta.
\end{equation*}

\subsection{Derivative of the Transition Kernels}
Recall that under the Bernoulli-randomized policy $\pi_\theta$, the transition kernel and expected reward at time $t$ are given by
\begin{equation*}
    P_\theta^t =  (1-\theta)P_0^t + \theta P_1^t, \quad r_\theta = (1-\theta)r_0 + \theta r_1.
\end{equation*}
Differentiating with respect to the treatment probability $\theta$, we have
\begin{equation*}
    \frac{d}{d\theta} P_\theta^t = P_1^t - P_0^t = \Delta^t, \quad \frac{d}{d\theta}r_\theta = r_1 - r_0.
\end{equation*}
By applying the product rule for matrix-valued functions for all  $1 \leq a < b \leq T$, we have
\begin{align*}
\frac{d}{d\theta} P_\theta^{a \rightarrow b} &= \frac{d}{d\theta}  \prod_{u=a}^{b-1} P_{\theta}^u = \sum_{u=a}^{b-1}\left(\prod_{v=a}^{u-1} P_\theta^v\right)\left(\frac{d}{d \theta} P_\theta^u\right)\left(\prod_{v=u+1}^{b-1} P_\theta^v\right) = \sum_{u=a}^{b-1} P_\theta^{a \rightarrow u} \Delta^u P_\theta^{(u+1) \rightarrow b},
\end{align*}
and $\frac{d}{d\theta} P_\theta^{a \rightarrow a} = 0$ for all $a \in [T]$.

\subsection{Strong Mixing Property}
We introduce the strong ($\alpha$) mixing property \citep{rosenblatt1956central} and relevant lemmas used in our variance analysis.

\begin{lemma}[Lemma 3 in \cite{doukhan1994mixing}]\label{lemma3_mixing}
    Let $X$ and $Y$ be bounded measurable random variables with respect to sigma-algebras $\mathcal{U}$ and $\mathcal{V}$, respectively. Denote $\|X\|_\infty=\inf \{M \geq 0: \mathbb{P}(|X|>M)=0\}$. Then,
    \begin{equation*}
        \abs{\operatorname{Cov}(X, Y)} \leq 4 \, \alpha(\mathcal{U}, \mathcal{V})  \|X\|_\infty \|Y\|_\infty,
    \end{equation*}
    where
    $$
    \alpha(\mathcal{U}, \mathcal{V}) = 
\sup_{\substack{A\in\mathcal{U}, B\in \mathcal{V}}}
\bigl|\mathbb{P}(A\cap B)-\mathbb{P}(A)\mathbb{P}(B)\bigr|.
    $$
\end{lemma}

\begin{lemma}\label{lemma:alpha-mixing}
Under Assumption~\ref{assumption:mixing-time}, the process $\{(X_t, Z_t, Y_t)\}_{t=1}^T$ is $\alpha$-mixing \citep{rosenblatt1956central} under Bernoulli randomization policy $\pi_\theta$.  Specifically, for any $1 \leq a \leq b \leq T$, define the sigma-algebra
\begin{equation*}
    \mathcal{F}_a^b :=\sigma\bigl((X_t, Z_t, Y_t)\colon a\leq t\leq b\bigr),
\end{equation*}
which represents the information generated by the states, actions, and outcomes from time $a$ through $b$.
{\color{black}For $1\leq h<T$, define the $\alpha$-mixing coefficient as follows, and set $\alpha(h)=0$ for $h\geq T$:}
\begin{equation*}
\alpha(h) =\sup_{1\leq s\leq T-h}\sup_{\substack{A\in\mathcal{F}_1^s\\B\in\mathcal{F}_{s+h}^T}}
\bigl|\mathbb{P}(A\cap B)-\mathbb{P}(A)\mathbb{P}(B)\bigr| \leq \gamma^{h-1}.
\end{equation*}
\end{lemma}
\begin{proof}[Proof of Lemma~\ref{lemma:alpha-mixing}]
    Fix any lag $h \geq 1$ and any split point $1 \leq s \leq T - h$. Let
\begin{equation*}
    \mathcal{F}_1^s = \sigma\left((X_1, Z_1, Y_1), \ldots, (X_s, Z_s, Y_s)\right), \,  \mathcal{F}_{s+h}^T = \sigma\left((X_{s+h}, Z_{s+h}, Y_{s+h}), \ldots, (X_T, Z_T, Y_T)\right).
\end{equation*}

By definition, we have
\begin{equation*}
    \alpha(h) =\sup_{1\leq s\leq T-h}\sup_{\substack{A\in\mathcal{F}_1^s\\B\in\mathcal{F}_{s+h}^T}}
\bigl|\mathbb{P}(A\cap B)-\mathbb{P}(A)\mathbb{P}(B)\bigr|.
\end{equation*}
    We show that $\bigl|\mathbb{P}(A\cap B)-\mathbb{P}(A)\mathbb{P}(B)\bigr| \leq \gamma^{h-1} $ for any $A\in \mathcal{F}_1^s$ and $B\in \mathcal{F}_{s+h}^T$.
    
    First, we observe
\begin{equation*}
    \mathbb{P}(A\cap B) = \mathbb{E}\left[ \mathbf{1}\{A\}\mathbf{1}\{B\}\right] = \mathbb{E}\left[\mathbb{E}\left[ \mathbf{1}\{A\}\mathbf{1}\{B\}\mid \mathcal{F}_1^s\right]\right] = \mathbb{E}\left[  \mathbf{1}\{A\} \mathbb{P}(B\mid \mathcal{F}_1^s)\right] 
    ,
\end{equation*}
    and
\begin{equation*}
    \mathbb{P}(A)\mathbb{P}(B) = \mathbb{E}\left[\mathbf{1}\{A\}   \mathbb{P}(B)\right].
\end{equation*}
Then, we have
    \begin{align*}
    \abs{\mathbb{P}(A\cap B) - \mathbb{P}(A)\mathbb{P}(B)} 
    &\leq \mathbb{E}\left[\mathbf{1}\{A\}\abs{ \left(  \mathbb{P}(B\mid \mathcal{F}_1^s)-  \mathbb{P}(B)\right) }\right] \\
    &\overset{\dagger}{\leq} \mathbb{P}(A) \sup _{\substack{x, x' \in \mathcal{X} \\ z, z' \in \mathcal{Z}}}\left|\mathbb{P}(B \mid X_s = x, Z_s = z)-\mathbb{P}\left(B \mid X_s = x', Z_s = z'\right)\right|\\
    &\overset{\dagger}{=} \mathbb{P}(A) \sup _{\substack{x, x' \in \mathcal{X} \\ z, z' \in \mathcal{Z}}} 
    \Bigg| 
    \sum_{y \in \mathcal{X}} \mathbb{P}(B \mid X_{s+h}=y) 
    \Big( \mathbb{P}(X_{s+h}=y \mid X_s=x,  Z_s=z) \\
    &\hspace{4.85cm} - \mathbb{P}(X_{s+h}=y \mid X_s=x',  Z_s=z') 
    \Big)
    \Bigg| \\
    &\overset{\ddagger}{\leq} \mathbb{P}(A) \sup _{\substack{x, x' \in \mathcal{X} \\ z, z' \in \mathcal{Z}}} \left\| (e_{x}P_z^s- e_{x'}P_{z'}^s) \prod_{t=s+1}^{s+h-1}P_{\textcolor{black}{\theta}}^t \right\|_{\mathrm{TV}}\\
    &\overset{\ast}{\leq} \mathbb{P}(A) \gamma^{h-1}\\
    &\leq \gamma^{h-1},
    \end{align*}
where $e_x$ denotes the probability distribution vector with a one in the $x$th position and zero elsewhere.
Here, ($\dagger$) follows from the Markov property and the independence of the reward, ($\ddagger$) follows from the dual representation of total variation, and ($\ast$) follows from the mixing-time property (Assumption~\ref{assumption:mixing-time}).
\end{proof}

\begin{corollary}\label{corollary:Gamma-alpha-mixing}
Under Lemma~\ref{lemma:alpha-mixing} and Bernoulli randomization policy $\pi_\theta$, define
\begin{equation*}
\Gamma_t
:=2\,(2Z_t-1)\sum_{u=t}^{\min(t+k,T)}Y_u.
\end{equation*}
{\color{black}The variable $\Gamma_t$ is measurable with respect to $\mathcal F_t^{\min(t+k,T)}$. We set $\alpha_\Gamma(h)=0$ for $h\geq T$. For $1\leq h<T$, its $\alpha$-mixing coefficient satisfies}
\begin{align*}
    \alpha_{\Gamma}(h) &= \sup_{1\leq s\leq T-h}
\sup_{\substack{A\in\sigma(\Gamma_{1:s})\\B\in\sigma(\Gamma_{s+h:T})}}
\bigl|\mathbb{P}(A\cap B)-\mathbb{P}(A)\mathbb{P}(B)\bigr| \\
&\leq
\sup_{1\leq s\leq T-h}
\sup_{\substack{A\in\mathcal{F}_1^{\textcolor{black}{\min(s+k,T)}}\\B\in\mathcal{F}_{s+h}^T}}
\bigl|\mathbb{P}(A\cap B)-\mathbb{P}(A)\mathbb{P}(B)\bigr| 
\\
&\le
\begin{cases}
1, & 1\le h\le k+1,\\[1mm]
\alpha(h-k), & h\ge k+2,
\end{cases}  \\
&\le
\begin{cases}
1, & 1\le h\le k+1,\\[1mm]
\gamma^{h-k-1}, & h\ge k+2,
\end{cases}\\
&=
\gamma^{\max\{0,h-k-1\}} .
\end{align*}
\end{corollary}

\section{Heuristic Truncation Size \texorpdfstring{$k$}{k}-Selection Algorithm}\label{appendix:algo}

\begin{algorithm}[H]
  \caption{Heuristic $k$-Selection via Stability}
  \small
  \label{alg:simple-k}
  \begin{algorithmic}
    \State {\bfseries Input:} Candidate set $\mathcal{K}=\{0,1,2,\dots,K_{\max}\}$; stability threshold $\alpha\ge 0$
    \State {\bfseries Output:} Selected $k^\ast$
    \State $k^\ast \leftarrow \varnothing$
    \State $(\widehat\tau_0,\widehat\sigma_0)\leftarrow$ TPG (or DM) and nonstationary HAC estimator under truncation size $0$
    \For{$k=1$ {\bfseries to} $K_{\max}$}
      \State $(\widehat\tau_k,\widehat\sigma_k)\leftarrow$ TPG and nonstationary HAC estimator under truncation size $k$
      \If{$|\widehat\tau_k - \widehat\tau_{k-1}| \le \alpha\, \widehat\sigma_k$}
        \State $k^\ast \leftarrow k$
        \State {\bfseries break}
      \EndIf
    \EndFor
    \If{$k^\ast = \varnothing$}
      \State $k^\ast \leftarrow 0$ \hfill \texttt{\# Use the naive DM estimator}
    \EndIf
    \State {\bfseries return} $k^\ast$
  \end{algorithmic}
\end{algorithm}

Algorithm~\ref{alg:simple-k} is a practical heuristic in the spirit of Lepski's method \citep{lepski1997optimal}. {\color{black}It selects the first $k\geq1$ such that $|\widehat\tau_k-\widehat\tau_{k-1}|\leq\alpha\widehat\sigma_k$, where $\widehat\sigma_k=\sqrt{\widehat\Omega_T/T}$ is the estimated standard error for truncation size $k$. 
If no candidate satisfies this criterion, the algorithm returns $k=0$ (i.e., the vanilla DM estimator).
} {\color{black}The optimal selection $k^\ast$ is intended to balance bias reduction and sampling variability. However, we note that the quadratic dependence on $k$ in Theorem~\ref{them:trunc_DQ_bias} is a variance bound for each fixed $k$ and does not imply that variance is monotone in $k$ or provide a theoretical guarantee for the data-dependent choice $k^\ast$.}

\begin{table}[htbp]
\centering
\caption{Sensitivity of post-$k$ selection to the hyperparameter $\alpha$.\label{tab:post_k_alpha_robustness}}
\begin{minipage}{0.85\textwidth}
\centering
\begin{tabular}{cccc}
\toprule
$\alpha$ & Confidence level & Median $k$ & RMSE \\
\midrule
$1.036$ & $70\%$ & $5$ & $0.043$ \\
$1.282$ & $80\%$ & $4$ & $0.039$ \\
$1.645$ & $90\%$ & $4$ & $0.040$ \\
$1.960$ & $95\%$ & $3$ & $0.050$ \\
\bottomrule
\end{tabular}
\end{minipage}
\par\smallskip
\begin{minipage}{0.85\textwidth}\footnotesize
\textit{Notes.} The threshold $\alpha$ is the standard normal critical value
associated with each two-sided confidence level. RMSE is computed against the
true GATE $\tau$ in Case Study I across simulation replications.
\end{minipage}
\end{table}
\paragraph{Choice of $\alpha$.}
{\color{black}We set $\alpha=1$ as the default value, corresponding to a stability threshold of one estimated standard error. Increasing $\alpha$ relaxes the stopping criterion. The normal critical values in Table~\ref{tab:post_k_alpha_robustness} provide alternative threshold calibrations, but their nominal confidence levels should be interpreted cautiously, e.g., not as post-selection coverage guarantees.} Table~\ref{tab:post_k_alpha_robustness} reports a robustness check across these calibrations for Case Study I in Section~\ref{sec:numerical}. The selected median $k$ and the resulting RMSE remain relatively stable, suggesting that the procedure is not highly sensitive to the precise choice of $\alpha$.

\section{DQ Estimator \texorpdfstring{\citep{farias2022markovian}}{} in a Stationary Average-Reward MDP}\label{appedix:dq}

In \cite{farias2022markovian}, a {\em difference-in-Q's} (DQ) estimator is proposed to estimate the average treatment effect (ATE) in a stationary, average-reward Markovian setting. This approach was later extended to finite-horizon and discounted-reward settings in \cite{farias2023correcting}.  Recall that in an average-reward MDP, for a state $x$ and action $z$, the $Q$-value $Q(x,z|\theta)$ is defined as the \emph{excess} reward over the long-run average reward obtained starting from state $x$ with initial action $z$, under policy $\pi_\theta$. The DQ estimator {\color{black}under balanced randomization uses an estimate} of the $Q$-function {\color{black}$Q(\cdot | 1/2)$} and computes the difference in steady-state $Q$-values between the treatment and control actions.
The DQ estimator has been shown to achieve a favorable bias-variance trade-off in stationary MDPs, especially when $Q$-values can be {\color{black}efficiently estimated from data} (e.g., with LSTD-$\lambda$ \citep{bradtke1996linear}). 

By the policy gradient theorem for stationary average‐reward MDPs \citep{sutton2018reinforcement}, when $J(\cdot)$ denotes the steady-state value function, its gradient $\nabla J(\theta)$ can be written as a stationary average of score-weighted Q-values, which in our case (since actions are binary) can be expressed as a difference in $Q$-values. Thus we can view the DQ estimator as a form of policy-gradient estimator {\color{black}under balanced randomization} (i.e., $\theta=1/2$) in the stationary setting. We demonstrate this in detail below.

{\color{black}We note that the derivation below assumes access to the true $Q$-function (or an unbiased estimator). In practice, estimating the $Q$-function from data can introduce an additional source of bias.}

\paragraph{Policy gradient in an average-reward MDP.}
In an average-reward MDP, the $Q$-value represents the cumulative excess reward of a state-action pair $(x,z)$ relative to the long-run average reward under a Bernoulli-randomized policy $\pi_\theta$. Specifically, it is defined as
\begin{equation*}
    Q(x, z|\theta):=\lim _{T \rightarrow \infty} \mathbb{E}_{\pi_\theta}\left[\sum_{i=1}^{T}\left( Y_i - \lambda_\theta \right) \mid X_1 = x, Z_1= z\right],
\end{equation*}
where $\E_{\pi_\theta}$ denotes the expectation under policy $\pi_\theta$, and  $\lambda_\theta$ denotes the long-run average return under policy $\pi_\theta$, i.e., $\lambda_\theta := \lim _{T \rightarrow \infty}\frac{1}{T} \E_{\pi_\theta}\sbrac{\sum_{i=1}^{T}Y_i}$. In this setting, the $Q$-value depends on the state and action, but no longer on time.

Under the ergodicity assumption, let $\rho_\theta$ denote the stationary distribution of the states induced by policy $\pi_\theta$. The steady-state value function $J(\theta)$ is then defined as
\begin{equation*}
    J(\theta) := \sum_{x\in\mathcal{X}}\rho_{\theta}(x)\sum_{z\in\mathcal{Z}} \pi_\theta(z|x)r(x, z). 
\end{equation*}

Then, the average treatment effect (ATE) in \cite{farias2022markovian} is defined as $\tau := J(1) - J(0)$. 
By the policy gradient theorem for average-reward MDPs (see Page 333 in \cite{sutton2018reinforcement}), we have
\begin{align*}
    \nabla J(\theta) &= \sum_{x\in\mathcal{X}}\rho_{\theta}(x)\sum_{z\in\mathcal{Z}} \nabla \pi_\theta(z|x)Q(x, z|\theta)\\
    &= \sum_{x\in\mathcal{X}}\rho_{\theta}(x) \left(Q(x, 1|\theta) - Q(x, 0|\theta)\right),
\end{align*}
{\color{black}which corresponds to the population DQ functional in \cite{farias2022markovian} when $\theta=1/2$.}

{
\color{black}

\section{TPG under General Bernoulli Randomization}
\label{app:general-randomization}

In this section, we extend the TPG estimator to a Bernoulli randomization policy $\pi_\theta$ for any $\theta\in(0,1)$. As in the balanced case ($\theta = 1/2$), the estimator is unbiased for the corresponding truncated policy gradient. For a fixed truncation size $0\leq k\leq T$, define the (Bernoulli assignment) score function
\[
s_\theta(z):=\frac{z}{\theta}-\frac{1-z}{1-\theta}.
\]
The corresponding TPG estimator is
\begin{equation}
\widehat\tau^{\mathrm{TPG}}_{\theta,k}
:=
\frac1T\sum_{u=1}^T s_\theta(Z_u)
\sum_{t=u}^{\min\{u+k,T\}}Y_t.
\label{eq:general-theta-tpg}
\end{equation}
Note that when $\theta=1/2$, this reduces to the estimator $\widehat\tau_k$ in \eqref{eq:trunc_IPW_est}.

Following Section~\ref{sec:trunc_policy_gradient}, define the \emph{truncated $Q$-value} under $\pi_\theta$ as
\begin{equation*}
Q_\theta^{u,k}(z)
:=
\sum_{t=u}^{\min\{u+k,T\}}
\E_\theta[Y_t\mid Z_u=z].
\end{equation*}
Let $\mu_q^t:=\rho^\top P_q^{1\rightarrow t}$ denote the row distribution of $X_t$ under policy $\pi_q$. We define the corresponding \emph{hypothetical truncated policy value} as
$$
J_{k,\theta}(q)
:=
\frac1T\sum_{t=1}^T
\mu_\theta^{\max\{1,t-k\}}
P_q^{\max\{1,t-k\}\rightarrow t}r_q.
$$
For each reward time $t$, this policy first follows $\pi_\theta$ through time $\max\{1,t-k\}-1$, and then follows $\pi_q$ from $\max\{1,t-k\}$ through $t$. We hold the experimental allocation $\theta$ fixed when differentiating with respect to $q$. In particular, $J_{k,1/2}(q)$ is the truncated policy value $J_k(q)$ defined in \eqref{eq:trunc_value_func}.

\begin{prop}[TPG under general Bernoulli randomization]
\label{prop:general-theta-tpg}
For every $\theta\in(0,1)$ and $0\leq k\leq T$,
\begin{equation}
\E_\theta[\widehat\tau^{\mathrm{TPG}}_{\theta,k}]
=
\frac1T\sum_{u=1}^T
\bigl[Q_\theta^{u,k}(1)-Q_\theta^{u,k}(0)\bigr]
=
\left.
\frac{\partial}{\partial q}J_{k,\theta}(q)
\right|_{q=\theta}.
\label{eq:general-tpg-q}
\end{equation}
For $k=T$, $J_{T,\theta}(q)=J(q)$, and hence
$
\E_\theta[\widehat\tau^{\mathrm{TPG}}_{\theta,T}]
=
J'(\theta).
$
\end{prop}

Proposition~\ref{prop:general-theta-tpg} extends the policy-gradient interpretation in the main text to any interior experimental allocation. For finite $k$, the estimator is unbiased for the hypothetical truncated policy gradient. When $k=T$, it is unbiased for $J'(\theta)$. We next study its bias relative to the GATE under a general experimental allocation.

Define
$
p_\theta:=\theta(1-\theta),
$
and
\begin{equation*}
\Psi_k(\gamma)
:=
\sum_{j=0}^{k-2}(j+1)\gamma^j,
\qquad
\Psi_0(\gamma)=\Psi_1(\gamma)=0.
\end{equation*}

\begin{prop}[Bias and variance bounds under general Bernoulli randomization]
\label{prop:general-theta-tpg-bias-variance}
Under Assumption~\ref{assumption:mixing-time}, for every fixed $\theta\in(0,1)$ and integer $0\leq k\leq T$,
\begin{align}
\left|
\E_\theta[\widehat\tau^{\mathrm{TPG}}_{\theta,k}]
-\tau
\right|
&\leq
\frac{2\delta M\gamma^k}{1-\gamma}
+
4\bigl[\theta^2+(1-\theta)^2\bigr]
\delta^2M\Psi_k(\gamma)
+
\frac{2|1-2\theta|\delta(1-\gamma^k)}
{1-\gamma}
\|r_1-r_0\|_\infty,
\label{eq:general-theta-tpg-bias-bound}\\
\operatorname{Var}_\theta
\left(\widehat\tau^{\mathrm{TPG}}_{\theta,k}\right)
&\leq
\frac{M^2(k+1)}{Tp_\theta}
\left(
2k+3+\frac{4\gamma}{1-\gamma}
\right).
\label{eq:general-theta-tpg-var-bound}
\end{align}
When $k\geq T-1$, the first term in \eqref{eq:general-theta-tpg-bias-bound} can be omitted.
\end{prop}

Using $\|r_1-r_0\|_\infty\leq2M$ and
\[
\Psi_k(\gamma)
\leq
\frac{1-\gamma^k}{(1-\gamma)^2},
\]
the bias bound can be written as
\begin{equation}
\left|
\E_\theta[\widehat\tau^{\mathrm{TPG}}_{\theta,k}]
-\tau
\right|
=
O\left(
\frac{\gamma^k}{1-\gamma}\delta M
+
\frac{1-\gamma^k}{(1-\gamma)^2}\delta^2M
+
|1-2\theta|
\frac{1-\gamma^k}{1-\gamma}\delta M
\right),
\label{eq:general-theta-tpg-bias-order}
\end{equation}
where the three terms correspond to the mixing bias from truncation, the second-order Taylor error, and the additional first-order error in $\delta$ induced by unbalanced randomization, respectively.

When $\theta=1/2$, the last term vanishes, and \eqref{eq:general-theta-tpg-bias-order} reduces to the bias order in Theorem~\ref{them:trunc_DQ_bias}. For $\theta\neq1/2$, the additional first-order term remains even when the full history is retained. In particular, when $k\geq T-1$,
\begin{equation*}
\left|
\E_\theta[\widehat\tau^{\mathrm{TPG}}_{\theta,k}]
-\tau
\right|
=
O\left(
|1-2\theta|
\frac{\delta M}{1-\gamma}
+
\frac{\delta^2M}{(1-\gamma)^2}
\right).
\end{equation*}
Therefore, TPG remains an unbiased estimator of the corresponding policy gradient for any interior experimental allocation, but its approximation error for the GATE can be first order in $\delta$ away from balanced randomization.

The next section introduces a reward-corrected estimator that removes this additional first-order term by retaining the endpoint rewards when approximating the treatment and control state distributions.

\section{Reward-Corrected Estimation of the GATE}
\label{app:reward-correction}

In this section, we introduce a reward-corrected (RC) estimator based on a first-order approximation of the treatment and control state distributions. The approximation retains the reward function under each endpoint policy (e.g., total-control and total-treatment), leading to an additional importance weight on future rewards. We first derive the estimator and its connection to DQ~\citep{farias2022markovian} and DN~\citep{peng2025differences}, and then establish its bias and variance bounds under truncation.

Throughout this section, we consider data that are generated under the experimental policy $\pi_\theta$ for a fixed $\theta\in(0,1)$. We use the notation in Appendix~\ref{app:general-randomization} and write $r_z(x)=r(x,z)$. All other model settings remain the same as in the main paper.

\subsection{Reward-Corrected Estimator}
\label{app:rc-construction}
\label{app:rc-geometry}

For clarity, we first consider the full-horizon case with $k=T$. Let $\bar\mu_q=T^{-1}\sum_{t=1}^T\mu_q^t$ be the average state distribution across the horizon, so that $J(q)=\bar\mu_qr_q$. By the product rule, the original policy gradient is
$$
J'(\theta)
=\bar\mu_\theta (r_1 - r_0)+\bar\mu_\theta' r_\theta.
$$
Instead of expanding $J(q)$ (as in the TPG bias analysis in the main paper), we can approximate the treatment and control state distributions around the experimental allocation $\theta$, i.e.,
\[
\bar\mu_1
\approx\bar\mu_\theta+(1-\theta)\bar\mu_\theta',
\qquad
\bar\mu_0
\approx\bar\mu_\theta-\theta\bar\mu_\theta'.
\]
Evaluating these approximations at $r_1$ and $r_0$, respectively, and taking their difference gives
\begin{align}
\bar\mu_1 r_1 - \bar\mu_0 r_0 \approx D(\theta)
&:= \bigl[\bar\mu_\theta+(1-\theta)\bar\mu_\theta'\bigr]r_1
-\bigl[\bar\mu_\theta-\theta\bar\mu_\theta'\bigr]r_0
\label{eq:rc-endpoint-tangent}\\
&=\bar\mu_\theta (r_1 - r_0)
+\bar\mu_\theta' r_{1-\theta},
\label{eq:rc-full-population}
\end{align}
where $r_{1-\theta}=(1-\theta)r_1+\theta r_0$.
The weights $1-\theta$ and $\theta$ arise from the distances between the experimental allocation and the two endpoint policies. Thus, {\em this approximation evaluates the change in state distribution using $r_{1-\theta}$}, whereas the original policy gradient uses $r_\theta$. In particular,
$$
D(\theta)-J'(\theta)
=(1-2\theta)\bar\mu_\theta'(r_1 - r_0).
$$
The two approximations coincide under balanced randomization at $\theta = 1/2$. For other allocations, $D(\theta)$ generally differs from $J'(\theta)$ because the endpoint rewards are kept unchanged in \eqref{eq:rc-endpoint-tangent}.

\paragraph{Reward-corrected estimator.}
To estimate $D(\theta)$ from the experimental trajectory, define the additional importance weight
\begin{equation}
w_\theta(z)
:=\frac{\pi_{1-\theta}(z)}{\pi_\theta(z)}
=\frac{(1-\theta)z}{\theta}
+\frac{\theta(1-z)}{1-\theta}
=1+(1-2\theta)s_\theta(z).
\label{eq:reward-transport-weight}
\end{equation}
Since
\[
\E_\theta[w_\theta(Z_t)Y_t\mid X_t]
=r_{1-\theta}(X_t),
\]
the weighted reward has conditional mean $r_{1-\theta}(X_t)$. The state distribution and the intervening transitions remain those under $\pi_\theta$.

For a truncation size $0\leq k\leq T$, we define the \emph{reward-corrected estimator} as
$$
\widehat\tau^{\mathrm{RC}}_{\theta,k}
:=\frac1T\sum_{u=1}^T s_\theta(Z_u)
\left[Y_u+\sum_{t=u+1}^{\min\{u+k,T\}}w_\theta(Z_t)Y_t\right].
$$
As in TPG, the score $s_\theta(Z_u)$ compares the outcomes following treatment and control at time $u$. The immediate reward $Y_u$ is unchanged, while each future reward receives the additional weight $w_\theta(Z_t)$. Both estimators account for the effect of an assignment on subsequent states over the same truncation window. In particular, the additional weight changes the continuation reward used to evaluate these states from $r_\theta$ to $r_{1-\theta}$. Like TPG, the RC estimator requires only the observed assignments and rewards.

Using \eqref{eq:reward-transport-weight}, we can also write
$$
\widehat\tau^{\mathrm{RC}}_{\theta,k}
=\widehat\tau^{\mathrm{TPG}}_{\theta,k}
+\frac{1-2\theta}{T}\sum_{u=1}^T\sum_{t=u+1}^{\min\{u+k,T\}}
s_\theta(Z_u)s_\theta(Z_t)Y_t.
$$
Note that when $\theta=1/2$, the additional term is zero, so the two estimators coincide on every sample path. When $k=0$, both reduce to the immediate-reward Horvitz--Thompson estimator for every interior allocation.

\subsection{Expected Value and Connection to DQ}
\label{app:rc-comparison}

We next derive the expectation of the RC estimator under the experimental policy $\pi_\theta$. For each reward time $t$, define
$$
A_{\theta,k}^t
:=
\sum_{u=\max\{1, t-k\}}^{t-1}
\mu_\theta^u\Delta^u
P_\theta^{(u+1)\rightarrow t},
\qquad
\Delta^u=P_1^u-P_0^u.
$$
The term $A_{\theta,k}^t$ captures the contribution to the derivative of the state distribution at time $t$ from assignments within the truncation window. To see this, define
\begin{equation}
\nu_q^t
:=
\mu_\theta^{\max\{1, t-k\}}P_q^{\max\{1, t-k\}\rightarrow t}.
\label{eq:rc-window-occupancy}
\end{equation}
Then $\nu_\theta^t=\mu_\theta^t$ and
\[
\left.
\frac{\partial}{\partial q}\nu_q^t
\right|_{q=\theta}
=
A_{\theta,k}^t,
\]
where the prefix distribution $\mu_\theta^{\max\{1, t-k\}}$ is held fixed.

We can now write the expectations of TPG and RC as
\begin{align}
\E_\theta[\widehat\tau^{\mathrm{TPG}}_{\theta,k}]
&=
\frac1T\sum_{t=1}^T
\left[
\mu_\theta^t(r_1-r_0)
+
A_{\theta,k}^t r_\theta
\right],
\notag\\
\E_\theta[\widehat\tau^{\mathrm{RC}}_{\theta,k}]
&=
\frac1T\sum_{t=1}^T
\left[
\mu_\theta^t(r_1-r_0)
+
A_{\theta,k}^t r_{1-\theta}
\right].
\label{eq:rc-population}
\end{align}
To obtain these expressions, first note that
\[
\E_\theta[s_\theta(Z_t)Y_t]
=
\mu_\theta^t(r_1-r_0).
\]
For $u<t$ and any state reward $g$,
\[
\E_\theta[s_\theta(Z_u)g(X_t)]
=
\mu_\theta^u\Delta^u
P_\theta^{(u+1)\rightarrow t}g.
\]
For the continuation rewards, conditioning on $X_t$ gives $r_\theta(X_t)$ for TPG and $r_{1-\theta}(X_t)$ for RC. Reversing the order of summation then gives \eqref{eq:rc-population}.

When $k=T$, all past assignments are included and the product rule gives
\[
A_{\theta,T}^t
=
\partial_\theta\mu_\theta^t.
\]
Therefore, $\E_\theta[\widehat\tau^{\mathrm{RC}}_{\theta,T}]=D(\theta)$, where $D(\theta)$ is the full-horizon approximation derived in \eqref{eq:rc-endpoint-tangent}.

For finite $k$, applying the same endpoint approximation to $\nu_q^t$ gives
$$
\E_\theta[\widehat\tau^{\mathrm{RC}}_{\theta,k}]
=
\frac1T\sum_{t=1}^T
\left[
\bigl(\nu_\theta^t+(1-\theta)A_{\theta,k}^t\bigr)r_1
-
\bigl(\nu_\theta^t-\theta A_{\theta,k}^t\bigr)r_0
\right].
$$
Thus, truncation applies the first-order approximation only to the most recent $k$ transitions. The proof of Theorem~\ref{them:trunc_DQ_bias} uses the same expansion of the treatment and control state distributions. Under balanced randomization, $r_{1-\theta}=r_\theta=r_{1/2}$, so this approximation coincides with the truncated policy gradient.

Therefore, the TPG and RC estimators generally target different first-order approximations to the GATE when $\theta \neq 1/2$. The following simple example illustrates this distinction and shows that the RC approximation can be exact even when the policy value $J(q)$ is nonlinear in $q$.

\paragraph{Unbiased example for the RC estimator.}
Let $X_1=0$, $X_{t+1}=Z_t$, and $Y_t=X_tZ_t$, with $T\geq2$. A reward requires treatment in two consecutive periods. In this case,
\[
J(q)=\frac{T-1}{T}q^2,
\qquad
\tau=\frac{T-1}{T}.
\]
For every $k\geq1$,
\[
\E_\theta[\widehat\tau^{\mathrm{TPG}}_{\theta,k}]
=
2\theta\frac{T-1}{T}
=
J'(\theta),
\qquad
\E_\theta[\widehat\tau^{\mathrm{RC}}_{\theta,k}]
=
\frac{T-1}{T}
=
\tau.
\]
Thus, TPG is unbiased for $J'(\theta)$, while the RC estimator is unbiased for the GATE in this example. This follows because the state distributions are affine in the allocation, even though $J(q)$ is quadratic.

\paragraph{Connection to DQ.}
\label{app:rc-dq}
We next connect the expectation of the RC estimator to DQ. Suppose the transition kernels are time-invariant and satisfy Assumption~\ref{assumption:mixing-time}. Let $\rho_\theta$ denote the stationary state distribution under $\pi_\theta$, and write $\Delta P=P_1-P_0$. For fixed $k$, taking $T\to\infty$ in \eqref{eq:rc-population} gives
$$
\lim_{T\to\infty}\E_\theta[\widehat\tau^{\mathrm{RC}}_{\theta,k}]
=
\Psi_{\theta,k}^{\mathrm{RC}}
:=
\rho_\theta^\top(r_1-r_0)
+
\rho_\theta^\top\Delta P
\sum_{j=0}^{k-1}P_\theta^j r_{1-\theta}.
$$
Here, the effect of the initial state and the shortened truncation windows near the end of the horizon vanish as $T\to\infty$.

We can then take the limit as $k\to\infty$. Since $\Delta P\mathbf{1} = (P_1 - P_0) \mathbf{1}=0$, we have
\begin{equation}
\Psi_{\theta,\infty}^{\mathrm{RC}}
=
\rho_\theta^\top(r_1-r_0)
+
\rho_\theta^\top(P_1-P_0)(I-P_\theta)^\#
\bigl[(1-\theta)r_1+\theta r_0\bigr],
\label{eq:stationary-reward-corrected-dq}
\end{equation}
where $^\#$ is the group inverse operator, i.e.,
\[
(I-P_\theta)^\#
:=
\sum_{j=0}^\infty
\bigl(P_\theta^j-\mathbf{1}\rho_\theta^\top\bigr).
\]
Equation~\eqref{eq:stationary-reward-corrected-dq} is the general-allocation first-order expansion underlying DQ \citep{farias2022markovian,farias2023correcting}. When $\theta=1/2$, we have $r_{1-\theta}=r_\theta=r_{1/2}$, so this expression also coincides with the policy gradient in Appendix~\ref{appedix:dq}. For other allocations, the policy gradient instead uses $r_\theta$ as the continuation reward. We note that this connection concerns the expectation of the estimator. A DQ implementation that estimates the $Q$-function from data (e.g., LSTD or TD$(\lambda)$) introduces an additional source of estimation error. Moreover, without truncation, the RC estimator is closely related to the differences-in-neighbors (DN) estimator proposed by \cite{peng2025differences} in the context of network interference.

\subsection{Bias, Variance, and Comparison with TPG}
\label{app:rc-bias-variance}

We now establish finite-sample bias and variance bounds for the RC estimator. As in the main paper, the bias bound also consists of two components, i.e., the mixing bias from omitting assignments before the truncation window and the Taylor error from approximating the treatment and control state distributions within the window while keeping their respective rewards unchanged.

Define
\[
p_\theta:=\theta(1-\theta),
\qquad
c_\theta:=\E_\theta[w_\theta(Z_t)^2]
=1+\frac{(1-2\theta)^2}{p_\theta},
\]
and
$$
\Psi_k(\gamma):=\sum_{j=0}^{k-2}(j+1)\gamma^j,
\qquad \Psi_0(\gamma)=\Psi_1(\gamma)=0.
$$
We then obtain the following finite-sample bias and variance bounds.

\begin{prop}[Bias and variance bounds for the RC estimator]
\label{prop:rc-bias-variance}
Under Assumption~\ref{assumption:mixing-time}, for every fixed $\theta\in(0,1)$ and integer $0\leq k\leq T$, with $\delta$ defined as in Theorem~\ref{them:trunc_DQ_bias}, we have
\begin{align}
\left|\E_\theta[\widehat\tau^{\mathrm{RC}}_{\theta,k}]-\tau\right|
&\leq\frac{2\delta M\gamma^k}{1-\gamma}
+4\bigl[\theta^2+(1-\theta)^2\bigr]\delta^2M\Psi_k(\gamma),
\label{eq:rc-bias-bound}\\
\operatorname{Var}_\theta(\widehat\tau^{\mathrm{RC}}_{\theta,k})
&\leq\frac{M^2}{Tp_\theta}(1+kc_\theta)
\left(2k+3+\frac{4\gamma}{1-\gamma}\right).
\label{eq:rc-variance-bound}
\end{align}
When $k\geq T-1$, the first term in the bias bound \eqref{eq:rc-bias-bound} can be omitted.
\end{prop}

The first term in the bias bound in \eqref{eq:rc-bias-bound} decreases geometrically in
$k$. The second is quadratic in the kernel deviation and satisfies
\[
\Psi_k(\gamma)
\leq\min\left\{\frac{k(k-1)}2,\frac1{(1-\gamma)^2}\right\}
\qquad(k\geq1).
\]
When $k=0$ or $k=1$, the state distribution within the truncation window is affine in the allocation, and hence the Taylor error is zero.

\paragraph{Comparison of TPG and RC.}
The expectation expressions in the previous section imply
$$
\E_\theta[\widehat\tau^{\mathrm{RC}}_{\theta,k}]-\E_\theta[\widehat\tau^{\mathrm{TPG}}_{\theta,k}]
=
\frac{1-2\theta}{T}
\sum_{t=1}^T A_{\theta,k}^t(r_1-r_0).
$$
Under Assumption~\ref{assumption:mixing-time}, the effect of an assignment at time $u$ on the state distribution at time $t$ decays at rate $\gamma^{t-u-1}$. Together with the definition of $\delta$, this gives
\[
\|A_{\theta,k}^t\|_1
\leq
2\delta
\sum_{u=\max\{1,t-k\}}^{t-1}
\gamma^{t-u-1}
\leq
\frac{2\delta(1-\gamma^k)}{1-\gamma}.
\]
Therefore,
$$
\left|
\E_\theta[\widehat\tau^{\mathrm{RC}}_{\theta,k}]
-\E_\theta[\widehat\tau^{\mathrm{TPG}}_{\theta,k}]
\right|
\leq
\frac{2|1-2\theta|\delta(1-\gamma^k)}{1-\gamma}
\|r_1-r_0\|_\infty.
$$
Thus, when $\theta\neq1/2$, the difference between the expectations of RC and TPG can be of first order in $\delta$ when the immediate reward difference $r_1-r_0$ is not small. The reward correction removes this first-order term from the approximation error. However, this does not imply that RC always has smaller absolute bias, since the remaining approximation errors may have different signs.

For any fixed $\theta\in(0,1)$ and $\gamma<1$, the variance bound is $O(M^2(k+1)^2/T)$. When $\theta=1/2$, we have $p_\theta=1/4$ and $c_\theta=1$, so \eqref{eq:rc-variance-bound} reduces to the TPG variance bound in \eqref{eq:tpg-variance-explicit}. The bias bound remains uniform over $\theta\in(0,1)$, while the variance bound increases as $\theta$ approaches either endpoint (e.g., $\theta \rightarrow 1$ or $\theta \rightarrow 0$) due to the larger importance weights.

\paragraph{Constant-reward example.}
Consider a one-state system with $Y_t\equiv M$. By independence and centering of the assignment scores, one can show that
\[
\operatorname{Var}_\theta(\widehat\tau^{\mathrm{RC}}_{\theta,k})
=
\frac{M^2}{T^2}
\left[
\frac1{p_\theta}
\sum_{u=1}^T
\bigl(1+\min\{k,T-u\}\bigr)^2
+
\frac{(1-2\theta)^2}{p_\theta^2}
\sum_{u=1}^T
\min\{k,T-u\}
\right].
\]
In particular, the first term gives the variance of TPG, while the second captures the additional variance introduced by the reward correction. Both estimators have mean zero in this example. For $0\leq k\leq T/2$, the first term is at least $M^2(k+1)^2/(2Tp_\theta)$, showing that the quadratic dependence on $k+1$ is sharp even in the absence of interference.

}

\section{Proof of Theorems}\label{app:thm}

\subsection{Proof of Theorem~\ref{them:trunc_DQ_bias}}

{\color{black}
\begin{proof}[Proof of Theorem~\ref{them:trunc_DQ_bias}]
The key idea here is to estimate the global treatment effect at each time using the transition kernel under the experimental policy $\pi_{1/2}$, and to show that this estimate corresponds to a truncated policy gradient (or truncated difference in $Q$-values by Proposition~\ref{prop:trunc_policy_gradient}). We first prove the bias bound, followed by the variance bound.

\paragraph{Bias bound.}
{\color{black}We use the $\ell_1$ norm for signed row vectors in this proof. The parameter $\delta$ retains its total-variation definition, so $\sup_x\|\Delta_\pm^u(x,\cdot)\|_1\leq\delta$, whereas $\sup_x\|\Delta^u(x,\cdot)\|_1\leq2\delta$.}
Recall that $\widehat\tau_k$ is unbiased for $\nabla J_k(1/2)$, so it suffices to bound
\[
    \abs{\mathbb{E}[\widehat\tau_k]-\tau}
    =
    \abs{\nabla J_k(1/2)-\tau}.
\]

For each $t\in[T]$, define
\[
    s_t := \max\{1,t-k\},
    \qquad
    \ell_t := t-s_t = \min\{k,t-1\}.
\]

We first record two operator bounds.

Since
\[
    P_{1/2}^u = \frac{1}{2}P_1^u+\frac{1}{2}P_0^u,
\]
Assumption~\ref{assumption:mixing-time} implies that $P_{1/2}^u$ inherits the same contraction coefficient $\gamma$: for any two probability distributions $f,f'$ on $\mathcal X$,
\[
    \|f'P_{1/2}^u-fP_{1/2}^u\|_{\textcolor{black}{1}}
    \le
    \frac12 \|f'P_1^u-fP_1^u\|_{\textcolor{black}{1}}
    +\frac12 \|f'P_0^u-fP_0^u\|_{\textcolor{black}{1}}
    \le
    \gamma \|f'-f\|_{\textcolor{black}{1}}.
\]
By Jordan decomposition, the same bound extends to all signed measures $\mu$ with $\mu(\mathcal X)=0$; hence for all $1\le a\le b\le T$,
\begin{equation}
    \|\mu P_z^{a\to b}\|_{\textcolor{black}{1}}
    \le
    \gamma^{\,b-a}\|\mu\|_{\textcolor{black}{1}},
    \qquad z\in\{0,1,1/2\}.
    \label{eq:trunc_kernel_contraction}
\end{equation}

Next, by the definition of $\delta$ and again by Jordan decomposition, for every signed measure $\mu$,
\begin{equation}
    \|\mu \Delta_+^u\|_{\textcolor{black}{1}}
    \le
    \delta \|\mu\|_{\textcolor{black}{1}},
    \qquad
    \|\mu \Delta_-^u\|_{\textcolor{black}{1}}
    \le
    \delta \|\mu\|_{\textcolor{black}{1}}.
    \label{eq:trunc_delta_operator}
\end{equation}

Fix $1\le a<b\le T$, and let $\mu^\top$ be any probability row vector. {\color{black}The identities and bounds below also apply when $a=b$, with the sums interpreted as empty. This covers windows with no preceding transitions.}
Using the telescoping identity for the difference of two matrix products,
\begin{align}
    \mu^\top P_1^{a\to b} - \mu^\top P_{1/2}^{a\to b}
    &=
    \sum_{u=a}^{b-1}
    \mu^\top P_1^{a\to u}\Delta_+^u P_{1/2}^{(u+1)\to b}.
    \label{eq:telescope_plus}
\end{align}
Subtracting the first-order term
\[
    \sum_{u=a}^{b-1}
    \mu^\top P_{1/2}^{a\to u}\Delta_+^u P_{1/2}^{(u+1)\to b},
\]
we obtain the exact remainder
\begin{align}
R_{a,b}^{(+)}(\mu)
&:=
\mu^\top P_1^{a\to b}
-
\mu^\top P_{1/2}^{a\to b}
-
\sum_{u=a}^{b-1}
\mu^\top P_{1/2}^{a\to u}\Delta_+^u P_{1/2}^{(u+1)\to b}
\notag\\
&=
\sum_{u=a}^{b-1}
\mu^\top \bigl(P_1^{a\to u}-P_{1/2}^{a\to u}\bigr)\Delta_+^u P_{1/2}^{(u+1)\to b}
\notag\\
&=
\sum_{a\le i<u\le b-1}
\mu^\top P_1^{a\to i}\Delta_+^i P_{1/2}^{(i+1)\to u}\Delta_+^u P_{1/2}^{(u+1)\to b},
\label{eq:exact_plus_remainder}
\end{align}
where in the last step we applied \eqref{eq:telescope_plus} again to
$P_1^{a\to u}-P_{1/2}^{a\to u}$.

Now each summand in \eqref{eq:exact_plus_remainder} is bounded, using
\eqref{eq:trunc_kernel_contraction} and \eqref{eq:trunc_delta_operator}, by
\begin{align*}
&\left\|
\mu^\top P_1^{a\to i}\Delta_+^i P_{1/2}^{(i+1)\to u}\Delta_+^u P_{1/2}^{(u+1)\to b}
\right\|_{\textcolor{black}{1}}\le
\delta\,
\gamma^{\,u-i-1}\,
\delta\,
\gamma^{\,b-u-1}
=
\delta^2\gamma^{\,u-i-1}\gamma^{\,b-u-1}.
\end{align*}
Therefore,
\begin{align}
\|R_{a,b}^{(+)}(\mu)\|_{\textcolor{black}{1}}
&\le
\delta^2
\sum_{a\le i<u\le b-1}
\gamma^{\,u-i-1}\gamma^{\,b-u-1}
\notag\\
&=
\delta^2 \sum_{m=0}^{b-a-2}(m+1)\gamma^m
=
\delta^2 \Psi_{\,b-a}(\gamma),
\label{eq:plus_remainder_bound}
\end{align}
where
$$
\Psi_k(\gamma)
    :=
    \sum_{m=0}^{k-2}(m+1)\gamma^m, \quad \left(\Psi_0: = 0\right).
$$
The same argument, with $P_0^u=P_{1/2}^u+\Delta_-^u$, yields the analogous identity and bound
\begin{equation}
\|R_{a,b}^{(-)}(\mu)\|_{\textcolor{black}{1}}
\le
\delta^2 \Psi_{\,b-a}(\gamma),
\label{eq:minus_remainder_bound}
\end{equation}
{\color{black}for the corresponding control-endpoint expansion around $P_{1/2}$.}

Fix $t\in[T]$. We decompose
\begin{align}
    \E_{\mathcal L_1^t}[Y_t]
    &=
    \rho^\top P_1^{1\to t}r_1 \notag\\
    &=
    \underbrace{
    \bigl(\rho^\top P_1^{1\to s_t}-\rho^\top P_{1/2}^{1\to s_t}\bigr)
    P_1^{s_t\to t}r_1
    }_{=:B_{1,t}}
    +
    \underbrace{
    \rho^\top P_{1/2}^{1\to s_t}P_1^{s_t\to t}r_1
    }_{=:L_{1,t}}.
    \label{eq:treatment_decomposition}
\end{align}

We first bound the tail term $B_{1,t}$.
If $s_t=1$, then $B_{1,t}=0$.
Otherwise $s_t=t-k$, and by the same telescoping argument used above,
\[
    \left\|
    \rho^\top P_1^{1\to s_t}-\rho^\top P_{1/2}^{1\to s_t}
    \right\|_{\textcolor{black}{1}}
    \le
    \sum_{u=1}^{s_t-1}\gamma^{\,s_t-u-1}\delta
    \le
    \frac{\delta}{1-\gamma}.
\]
Since this signed measure has total mass zero, another application of
\eqref{eq:trunc_kernel_contraction} gives
\[
    \left\|
    \bigl(\rho^\top P_1^{1\to s_t}-\rho^\top P_{1/2}^{1\to s_t}\bigr)
    P_1^{s_t\to t}
    \right\|_{\textcolor{black}{1}}
    \le
    \gamma^k \frac{\delta}{1-\gamma}.
\]
Using $\|r_1\|_\infty\le M$, we obtain the uniform bound
\begin{equation}
    |B_{1,t}|
    \le
    \frac{\gamma^k}{1-\gamma}\,\delta M.
    \label{eq:B1t_bound}
\end{equation}

Next, apply \eqref{eq:exact_plus_remainder}--\eqref{eq:plus_remainder_bound}
with $\mu^\top=\rho^\top P_{1/2}^{1\to s_t}$, $a=s_t$, and $b=t$.
Since $\rho^\top P_{1/2}^{1\to s_t}$ is a probability row vector,
\begin{align}
    L_{1,t}
    &=
    \rho^\top P_{1/2}^{1\to t}r_1
    +
    \sum_{u=s_t}^{t-1}
    \rho^\top P_{1/2}^{1\to u}\Delta_+^u P_{1/2}^{(u+1)\to t} r_1
    +
    \varepsilon_{1,t},
    \label{eq:L1t_expansion}
\end{align}
where
\begin{equation}
    |\varepsilon_{1,t}|
    \le
    \|R_{s_t,t}^{(+)}(\rho^\top P_{1/2}^{1\to s_t})\|_{\textcolor{black}{1}}\,
    \|r_1\|_\infty
    \le
    \delta^2 M\,\Psi_{\ell_t}(\gamma)
    \le
    \delta^2 M\,\Psi_k(\gamma).
    \label{eq:eps1t_bound}
\end{equation}

Combining \eqref{eq:treatment_decomposition}, \eqref{eq:B1t_bound},
\eqref{eq:L1t_expansion}, and \eqref{eq:eps1t_bound}, we obtain
\begin{equation}
    \E_{\mathcal L_1^t}[Y_t]
    =
    \phi_1^t
    +
    O\!\left(
        \Psi_k(\gamma)\delta^2 M
        +
        \frac{\gamma^k}{1-\gamma}\delta M
    \right),
    \label{eq:phi1_approx_refined}
\end{equation}
where
\[
    \phi_1^t
    :=
    \rho^\top P_{1/2}^{1\to t}r_1
    +
    \sum_{u=\max(t-k,1)}^{t-1}
    \rho^\top P_{1/2}^{1\to u}\Delta_+^u P_{1/2}^{(u+1)\to t} r_1.
\]

Exactly the same argument, using \eqref{eq:minus_remainder_bound}, gives
\begin{equation}
    \E_{\mathcal L_0^t}[Y_t]
    =
    \phi_0^t
    +
    O\!\left(
        \Psi_k(\gamma)\delta^2 M
        +
        \frac{\gamma^k}{1-\gamma}\delta M
    \right),
    \label{eq:phi0_approx_refined}
\end{equation}
where
\[
    \phi_0^t
    :=
    \rho^\top P_{1/2}^{1\to t}r_0
    +
    \sum_{u=\max(t-k,1)}^{t-1}
    \rho^\top P_{1/2}^{1\to u}\Delta_-^u P_{1/2}^{(u+1)\to t} r_0.
\]

\medskip
\noindent
\emph{Identification with the truncated policy gradient.}

Exactly as in the original proof,
\begin{align*}
    \phi_1^t-\phi_0^t
    =
    \sum_{u=\max(t-k,1)}^{t}
    \left(
        \mathbb{E}_{\textcolor{black}{1/2}}[Y_t\mid Z_u=1]
        -
        \mathbb{E}_{\textcolor{black}{1/2}}[Y_t\mid Z_u=0]
    \right).
\end{align*}
Hence,
\begin{align*}
    \frac1T\sum_{t=1}^T(\phi_1^t-\phi_0^t)
    &=
    \frac1T\sum_{t=1}^T \sum_{u=\max(t-k,1)}^{t}
    \left(
        \mathbb{E}_{\textcolor{black}{1/2}}[Y_t\mid Z_u=1]
        -
        \mathbb{E}_{\textcolor{black}{1/2}}[Y_t\mid Z_u=0]
    \right) \\
    &=
    \frac1T\sum_{u=1}^T \sum_{t=u}^{\min(u+k,T)}
    \left(
        \mathbb{E}_{\textcolor{black}{1/2}}[Y_t\mid Z_u=1]
        -
        \mathbb{E}_{\textcolor{black}{1/2}}[Y_t\mid Z_u=0]
    \right) \\
    &=
    \nabla J_k(1/2),
\end{align*}
where the last equality follows from Proposition~\ref{prop:trunc_policy_gradient}.

Then, using \eqref{eq:phi1_approx_refined} and \eqref{eq:phi0_approx_refined},
\begin{align*}
    \abs{\mathbb E[\widehat\tau_k]-\tau}
    &=
    \abs{\nabla J_k(1/2)-\tau} \\
    &=
    \left|
        \frac1T\sum_{t=1}^T(\phi_1^t-\phi_0^t)
        -
        \frac1T\sum_{t=1}^T
        \bigl(
            \E_{\mathcal L_1^t}[Y_t]-\E_{\mathcal L_0^t}[Y_t]
        \bigr)
    \right| \\
    &\le
    \frac1T\sum_{t=1}^T
    \left(
        \abs{\phi_1^t-\E_{\mathcal L_1^t}[Y_t]}
        +
        \abs{\phi_0^t-\E_{\mathcal L_0^t}[Y_t]}
    \right) \\
    &=
    O\!\left(
        \Psi_k(\gamma)\delta^2 M
        +
        \frac{\gamma^k}{1-\gamma}\delta M
    \right).
\end{align*}
Equivalently, for $k\ge 1$,
\[
    \Psi_k(\gamma)
    =
    \frac{1-k\gamma^{k-1}+(k-1)\gamma^k}{(1-\gamma)^2} \leq \frac{1 - \gamma^k}{(1-\gamma)^2}.
\]
Therefore, the bias bound can be simplified as follows
$$
\abs{\mathbb E[\widehat\tau_k]-\tau} = O\pbrac{\frac{1 - \gamma^k}{(1-\gamma)^2}\,\delta^2M + \frac{\gamma^k}{1-\gamma}\,\delta M}.
$$

{\color{black}
\paragraph{Variance bound.}
Write $s_u=2(2Z_u-1)$, $a_t=\max\{1,t-k\}$, and $\ell_t=t-a_t+1$. Reversing the order of summation gives
\begin{equation}
 \widehat\tau_k=\frac1T\sum_{t=1}^T B_t,
 \qquad B_t:=Y_t\sum_{u=a_t}^t s_u.
 \label{eq:tpg-reward-time-sum}
\end{equation}
For any state $x$, let $\E_{a_t,x}$ denote expectation when the experimental process is started at $X_{a_t}=x$. This defines the conditional calculations below even for states with zero probability under the original initial distribution. The assignments from $a_t$ onward are independent Bernoulli($1/2$), so $\E[s_u]=0$ and $\E[s_u^2]=4$. Bounded rewards therefore give the uniform second-moment bound
\begin{equation}
 \E_{a_t,x}[B_t^2]
 \leq M^2\E_{a_t,x}\!\left[\left(\sum_{u=a_t}^t s_u\right)^2\right]
 =4M^2\ell_t\leq V_k,
 \qquad V_k:=4M^2(k+1).
 \label{eq:tpg-local-second-moment}
\end{equation}
In particular, $\Var(B_t)\leq V_k$. Cauchy--Schwarz also gives
$|\operatorname{Cov}(B_i,B_j)|\leq V_k$ for all $i,j$.

For separated windows, we obtain a decaying covariance bound. Suppose $j-i\geq k+2$, so $a_j=j-k\geq i+2$, and define
\begin{equation*}
 g_j(x):=\E_{a_j,x}[B_j],\qquad
 \mathcal F_i:=\sigma\bigl((X_u,Z_u,Y_u)_{1\leq u\leq i},X_{i+1}\bigr).
\end{equation*}
By \eqref{eq:tpg-local-second-moment}, $\|g_j\|_\infty\leq\sqrt{V_k}$. The Markov and conditional reward assumptions imply
\begin{equation*}
 \E[B_j\mid\mathcal F_i]
 =\bigl(P_{1/2}^{(i+1)\rightarrow a_j}g_j\bigr)(X_{i+1}).
\end{equation*}
Assumption~\ref{assumption:mixing-time} bounds the distance between the propagated distributions by $\gamma^{a_j-i-1}$ times their initial total variation distance. Comparing the conditional starting point $X_{i+1}$ with its unconditional distribution thus gives
\begin{equation*}
 \left|\E[B_j\mid\mathcal F_i]-\E[B_j]\right|
 \leq 2\sqrt{V_k}\,\gamma^{j-i-k-1}.
\end{equation*}
Since $B_i$ is $\mathcal F_i$-measurable and
$\E|B_i-\E B_i|\leq\sqrt{V_k}$, it follows that
$$
 |\operatorname{Cov}(B_i,B_j)|
 \leq 2V_k\gamma^{j-i-k-1},\qquad j-i\geq k+2.
$$
Combining this bound with Cauchy--Schwarz for lags at most $k+1$ yields
\begin{align}
 \Var(\widehat\tau_k)
 &\leq\frac1{T^2}\left[TV_k+
 2\sum_{h=1}^{\min\{k+1,T-1\}}(T-h)V_k+
 4V_k\sum_{h=k+2}^{T-1}(T-h)\gamma^{h-k-1}\right]\notag\\
 &\leq\frac{4M^2(k+1)}{T}
 \left(2k+3+\frac{4\gamma}{1-\gamma}\right)\notag\\
 &=O\!\left(\frac{(k+1)^2M^2}{T}
 +\frac{\gamma(k+1)M^2}{T(1-\gamma)}\right).
 \label{eq:tpg-variance-explicit}
\end{align}
Empty sums in this display are zero. The proof uses the second moment of the centered score sum, which is linear in the window length, and mixing to control the dependence between separated windows.

The quadratic dependence on $k+1$ cannot be improved in general for this estimator. Consider a one-state system with $P_0^t=P_1^t$ and $Y_t\equiv M$. It satisfies Assumption~\ref{assumption:mixing-time} for any $\gamma\in(0,1)$ and has zero treatment effect. Independence of the scores gives
\begin{equation}
 \Var(\widehat\tau_k)
 =\frac{4M^2}{T^2}\sum_{u=1}^T
 \bigl(1+\min\{k,T-u\}\bigr)^2.
 \label{eq:tpg-variance-sharpness}
\end{equation}
For $0\leq k\leq T/2$, the first $T-k$ terms equal $(k+1)^2$, so
$\Var(\widehat\tau_k)\geq 2M^2(k+1)^2/T$.
Thus the quadratic order in \eqref{eq:var_bound} is sharp for fixed mixing rate. At $k=T$, \eqref{eq:tpg-variance-sharpness} equals
$2M^2(T+1)(2T+1)/(3T)$, showing that the untruncated variance can grow linearly in $T$ even without interference.
}\end{proof}
}

\subsection{Proof of Theorem~\ref{thm:truncDQ_CLT_HAC_concise}}
\begin{lemma}[CLT for the TPG estimator]\label{lemma:trunc_DQ_CLT}
Fix a finite truncation size $k \geq 0$. 
Under Assumptions~\ref{assumption:mixing-time} and \ref{ass:asymptotic_variance}, define
$$
\sigma_k^2 
:= \lim_{T\to\infty} \frac{1}{T} \,
\Var\!\left( \sum_{u=1}^T 
\left( \frac{\mathbf{1}\{Z_u=1\}}{1/2} - \frac{\mathbf{1}\{Z_u=0\}}{1/2} \right)
\sum_{t=u}^{\min(u+k,\,T)} Y_t \right),
$$
then, as $T \to \infty$,
$
\sqrt{T}\,\big( \widehat{\tau}_k - \E[\widehat{\tau}_k] \big)
\;\Rightarrow\; \mathcal{N}(0,\,\sigma_k^2).
$
\end{lemma}

\begin{lemma}[Consistency of the ideal HAC estimator]\label{lem:HAC_consistency}
Under Assumptions~\ref{assumption:mixing-time} and \ref{ass:asymptotic_variance},
define
$$
V_t := B_t - \mathbb{E}[B_t], \qquad
\widetilde\Gamma_\ell := \frac{1}{T}\sum_{t=1}^{T-\ell}  V_t  V_{t+\ell}, 
\qquad
w^{\mathrm{NW}}_\ell := 1-\frac{\ell}{L_T+1},
$$
{\color{black}where $k$ is fixed and $L_T$ is a deterministic integer bandwidth, $0\leq L_T<T$, with $L_T\to\infty$ and $L_T=O(T^{1/3})$.} 
Construct the ideal HAC estimator as
$$
\widetilde\Omega_T \;:=\; \widetilde\Gamma_0 \,+\, 2\sum_{\ell=1}^{L_T} w^{\mathrm{NW}}_\ell\,\widetilde\Gamma_\ell.
$$
Then,
$\widetilde\Omega_T$ is a consistent estimator of the asymptotic variance $\sigma_k^2$ in Lemma~\ref{lemma:trunc_DQ_CLT}, i.e.,
$
\widetilde\Omega_T \;\xrightarrow{p}\; \sigma_k^2.
$
\end{lemma}

\begin{lemma}[Consistency of the nonstationary HAC estimator under Ces\`aro-$L^2$ mean stability]\label{lemma:HAC_consistency}
Under Assumptions~\ref{assumption:mixing-time}, \ref{ass:asymptotic_variance}, and \ref{ass:mean-drift}, let $\bar B := T^{-1}\sum_{t=1}^T B_t$, 
$$
\widehat V_{t} := B_t - \bar B, \qquad
\widehat\Gamma_\ell := \frac{1}{T}\sum_{t=1}^{T-\ell} \widehat V_{t}\,\widehat V_{t+\ell},
\qquad
w^{\mathrm{NW}}_\ell := 1-\frac{\ell}{L_T+1},
$$
{\color{black}where $k$ is fixed and $L_T$ is a deterministic integer bandwidth, $0\leq L_T<T$, with $L_T\to\infty$ and $L_T=O(T^{1/3})$.}
Construct the nonstationary HAC estimator as
$$
\widehat\Omega_T := \widehat\Gamma_0 + 2\sum_{\ell=1}^{L_T} w^{\mathrm{NW}}_\ell\,\widehat\Gamma_\ell.
$$
Then, as $T\to\infty$,
$
\widehat\Omega_T - \widetilde\Omega_T = o_p(1)
$.
\end{lemma}
\begin{proof}[Proof of Theorem~\ref{thm:truncDQ_CLT_HAC_concise}]
The theorem follows from Lemma~\ref{lemma:trunc_DQ_CLT}, Lemma~\ref{lem:HAC_consistency}, and Lemma~\ref{lemma:HAC_consistency}.
\end{proof}

\section{Proof of Propositions}\label{app:prop}
\subsection{Proof of Proposition~\ref{prop:policy_gradient}}
\begin{proof}[Proof of Proposition~\ref{prop:policy_gradient}]
The value function $J(\theta)$ can be written as
\begin{equation*}
    J(\theta)= \frac{1}{T}\sum_{t=1}^{T} \E_{\mathcal{L}_{\theta}^t}\sbrac{Y_t} = \frac{1}{T}\sum_{t=1}^T \rho^{\top} P_\theta^{1 \rightarrow t} r_\theta.
\end{equation*}
Then, differentiating $J(\theta)$ yields
\begin{align*}
    \frac{d}{d\theta}J(\theta) &= \frac{1}{T}\sum_{t=1}^T \frac{d}{d\theta} \rho^{\top} P_\theta^{1 \rightarrow t} r_\theta\\
    &= \frac{1}{T}\sum_{t=1}^T \rho^{\top}\left[\sum_{u=1}^{t-1} P_\theta^{1 \rightarrow u} \Delta^u P_\theta^{(u+1) \rightarrow t}\right] r_\theta+\frac{1}{T}\sum_{t=1}^T \rho^{\top} P_\theta^{1 \rightarrow t}\left(r_1-r_0\right)\\
    &= \frac{1}{T}\sum_{t=1}^T \left[ \rho^{\top} P_\theta^{1 \rightarrow t}\left(r_1-r_0\right) + \sum_{u=1}^{t-1} \rho^{\top} P_\theta^{1 \rightarrow u} \Delta^u P_\theta^{(u+1) \rightarrow t} r_\theta \right].
\end{align*}
Evaluate at $\theta = 1/2$ (more generally, for any $\theta\in(0,1)$), we have
\begin{align*}
    \nabla J(1/2) 
    &= \left. \frac{d}{d\theta} J(\theta) \right|_{\theta=1/2} \\
    &= \frac{1}{T} \sum_{t=1}^T \left( \sum_{u=1}^{t-1} \rho^\top P_{1/2}^{1 \rightarrow u} \Delta^u P_{1/2}^{(u+1) \rightarrow t} r_{1/2} + \rho^\top P_{1/2}^{1 \rightarrow t} (r_1 - r_0) \right).
\end{align*}
Observe that each term $\rho^\top P_{1/2}^{1 \rightarrow u} \Delta^u P_{1/2}^{(u+1) \rightarrow t} r_{1/2}$ corresponds to the difference in the conditional expected outcome at time $t$ between taking action $Z_u = 1$ versus $Z_u = 0$ at time $u$, under policy $\pi_{1/2}$ and initial state distribution $X_1 \sim \rho$. Thus, we can express the gradient as
\begin{align*}
    \nabla J(1/2)
    &= \frac{1}{T} \sum_{t=1}^T \sum_{u=1}^{t} \left( \mathbb{E}_{\textcolor{black}{1/2}}[Y_t \mid Z_u = 1] - \mathbb{E}_{\textcolor{black}{1/2}}[Y_t \mid Z_u = 0] \right) \\
    &= \frac{1}{T} \sum_{u=1}^T \sum_{t=u}^{T} \left( \mathbb{E}_{\textcolor{black}{1/2}}[Y_t \mid Z_u = 1] - \mathbb{E}_{\textcolor{black}{1/2}}[Y_t \mid Z_u = 0] \right) \\
    &= \frac{1}{T} \sum_{u=1}^T \left( Q_{1/2}^u(1) - Q_{1/2}^u(0) \right) \\
    &= \frac{1}{T} \sum_{t=1}^T \left( Q_{1/2}^t(1) - Q_{1/2}^t(0) \right),
\end{align*}
which completes the proof.
\end{proof}

\subsection{Proof of Proposition~\ref{them:DQ_bias}}

{\color{black}

\begin{proof}[Proof of Proposition~\ref{them:DQ_bias}]
We first show that the untruncated PG estimator $\widehat{\tau}_T$ is an unbiased
estimator of the policy gradient $\nabla J(1/2)$. Recall that
\begin{equation*}
    \widehat \tau_T
    =
    \frac{1}{T}\sum_{u=1}^T
    \left(
        \frac{\mathbf{1}\left\{Z_u=1\right\}}{1/2}
        -
        \frac{\mathbf{1}\left\{Z_u=0\right\}}{1/2}
    \right)
    \sum_{t=u}^T Y_t .
\end{equation*}
Under the experimental policy $\pi_{1/2}$ and initial state distribution
$X_1\sim \rho$, taking expectation gives
\begin{align*}
    \mathbb{E}\sbrac{\widehat \tau_T}
    &=
    \frac{1}{T}\sum_{u=1}^T
    \mathbb{E}\sbrac{
        \frac{\mathbf{1}\left\{Z_u=1\right\}}{1/2}
        \sum_{t=u}^T Y_t
        -
        \frac{\mathbf{1}\left\{Z_u=0\right\}}{1/2}
        \sum_{t=u}^T Y_t
    } \\
    &=
    \frac{1}{T}\sum_{u=1}^T
    \left(
        Q_{1/2}^u(1)-Q_{1/2}^u(0)
    \right)
    =
    \nabla J(1/2),
\end{align*}
where the final equality follows from Proposition~\ref{prop:policy_gradient}.
Hence, $\widehat{\tau}_T$ is an unbiased estimator of $\nabla J(1/2)=J'(1/2)$.

We then prove the bias bound. Recall that
\[
    J(\theta)
    :=
    \frac{1}{T}\sum_{t=1}^T
    \rho^\top P_\theta^{1\to t} r_\theta,
    \qquad
    \tau = J(1)-J(0).
\]
Since $\mathbb{E}\sbrac{\widehat \tau_T}=J'(1/2)$,
\[
    \mathbb{E}\sbrac{\widehat \tau_T}-\tau
    =
    J'(1/2)-\bigl(J(1)-J(0)\bigr).
\]
Using the integral form of the Taylor remainder around $\theta=1/2$,
\begin{equation}
    \mathbb{E}\sbrac{\widehat \tau_T}-\tau
    =
    \int_0^{1/2} \theta J''(\theta)\,d\theta
    -
    \int_{1/2}^1 (1-\theta)J''(\theta)\,d\theta .
    \label{eq:centered_bias_integral}
\end{equation}

It remains to bound this centered remainder. Let
\[
    d_r := r_1-r_0,
    \qquad
    \Delta^u := P_1^u-P_0^u,
    \qquad
    \bar\theta := 1/2.
\]
By differentiating $J(\theta)$ twice, we may write
\begin{equation}
    J''(\theta)
    =
    2D(\theta)+A(\theta)+B(\theta),
    \label{eq:Jpp_decomp_no_mixing_centered}
\end{equation}
where
\begin{align}
D(\theta)
&:=
\frac{1}{T}\sum_{t=1}^T\sum_{u=1}^{t-1}
\rho^\top
P_\theta^{1\to u}
\Delta^u
P_\theta^{(u+1)\to t}
d_r,
\label{eq:D_theta_def}
\\
A(\theta)
&:=
\frac{1}{T}\sum_{t=1}^T\sum_{u=1}^{t-1}\sum_{i=1}^{u-1}
\rho^\top
P_\theta^{1\to i}
\Delta^i
P_\theta^{(i+1)\to u}
\Delta^u
P_\theta^{(u+1)\to t}
r_\theta,
\label{eq:A_theta_def}
\\
B(\theta)
&:=
\frac{1}{T}\sum_{t=1}^T\sum_{u=1}^{t-1}\sum_{j=u+1}^{t-1}
\rho^\top
P_\theta^{1\to u}
\Delta^u
P_\theta^{(u+1)\to j}
\Delta^j
P_\theta^{(j+1)\to t}
r_\theta.
\label{eq:B_theta_def}
\end{align}
Define the first-order term at the centered design by
\[
    \bar D := D(\bar\theta)=D(1/2).
\]
Then
\begin{equation}
    J''(\theta)
    =
    2\bar D + R(\theta),
    \qquad
    R(\theta)
    :=
    2\pbrac{D(\theta)-\bar D}+A(\theta)+B(\theta).
    \label{eq:Jpp_centered_linear_remainder}
\end{equation}

We now show that, without any mixing assumption,
\[
    \sup_{\theta\in[0,1]}\abs{R(\theta)}
    =
    O\pbrac{T^2\delta^2M}.
    \label{eq:R_bound_no_mixing_centered}
\]
{\color{black}We use the $\ell_1$ norm for signed row vectors, while $\delta$ retains its definition in total variation. Markov kernels are nonexpansive under the $\ell_1$ norm:}
$$
    \|\mu P_\theta^{a\to b}\|_{\textcolor{black}{1}}
    \le
    \|\mu\|_{\textcolor{black}{1}},
$$
together with
\[
    \|\mu\Delta^u\|_{\textcolor{black}{1}}
    \le
    \textcolor{black}{2}\delta \|\mu\|_{\textcolor{black}{1}},
    \qquad
    \|r_\theta\|_\infty\le M,
    \qquad
    \|d_r\|_\infty\le 2M.
\]

First, every term in $A(\theta)$ contains two $\Delta$ insertions. Therefore,
uniformly over $\theta\in[0,1]$,
\begin{align}
\abs{A(\theta)}
&\le
\frac{\textcolor{black}{4}\delta^2 M}{T}
\sum_{t=1}^T\sum_{u=1}^{t-1}\sum_{i=1}^{u-1} 1
=
O\pbrac{T^2\delta^2M}.
\label{eq:A_bound_no_mixing_centered}
\end{align}
Similarly,
\begin{align}
\abs{B(\theta)}
&\le
\frac{\textcolor{black}{4}\delta^2 M}{T}
\sum_{t=1}^T\sum_{u=1}^{t-1}\sum_{j=u+1}^{t-1} 1
=
O\pbrac{T^2\delta^2M}.
\label{eq:B_bound_no_mixing_centered}
\end{align}

It remains to bound $D(\theta)-\bar D$. By adding and subtracting
$P_{\bar\theta}^{1\to u}\Delta^u P_\theta^{(u+1)\to t}$,
\begin{align}
D(\theta)-\bar D
&=
\frac{1}{T}\sum_{t=1}^T\sum_{u=1}^{t-1}
\rho^\top
\left(
    P_\theta^{1\to u}
    -
    P_{\bar\theta}^{1\to u}
\right)
\Delta^u
P_\theta^{(u+1)\to t}
d_r
\notag\\
&\quad
+
\frac{1}{T}\sum_{t=1}^T\sum_{u=1}^{t-1}
\rho^\top
P_{\bar\theta}^{1\to u}
\Delta^u
\left(
    P_\theta^{(u+1)\to t}
    -
    P_{\bar\theta}^{(u+1)\to t}
\right)
d_r .
\label{eq:D_minus_Dbar_split}
\end{align}
For the first difference, the product telescoping identity gives
\[
    P_\theta^{1\to u}
    -
    P_{\bar\theta}^{1\to u}
    =
    \left(\theta-\bar\theta\right)
    \sum_{i=1}^{u-1}
    P_\theta^{1\to i}
    \Delta^i
    P_{\bar\theta}^{(i+1)\to u}.
\]
For the second difference,
\[
    P_\theta^{(u+1)\to t}
    -
    P_{\bar\theta}^{(u+1)\to t}
    =
    \left(\theta-\bar\theta\right)
    \sum_{j=u+1}^{t-1}
    P_\theta^{(u+1)\to j}
    \Delta^j
    P_{\bar\theta}^{(j+1)\to t}.
\]
Since $\abs{\theta-\bar\theta}\le 1/2$, each term in
\eqref{eq:D_minus_Dbar_split} contains two $\Delta$ insertions. Hence,
uniformly over $\theta\in[0,1]$,
\begin{align}
\abs{D(\theta)-\bar D}
&\le
\frac{C\delta^2M}{T}
\sum_{t=1}^T\sum_{u=1}^{t-1}
\left\{
    \sum_{i=1}^{u-1}1
    +
    \sum_{j=u+1}^{t-1}1
\right\}
\notag\\
&=
O\pbrac{T^2\delta^2M},
\label{eq:D_minus_Dbar_bound_no_mixing_centered}
\end{align}
for a universal constant $C>0$.

Combining
\eqref{eq:A_bound_no_mixing_centered},
\eqref{eq:B_bound_no_mixing_centered}, and
\eqref{eq:D_minus_Dbar_bound_no_mixing_centered} gives
\[
    \sup_{\theta\in[0,1]}\abs{R(\theta)}
    =
    O\pbrac{T^2\delta^2M}.
\]

Substituting the decomposition
\eqref{eq:Jpp_centered_linear_remainder} into
\eqref{eq:centered_bias_integral}, we obtain
\begin{align*}
\mathbb{E}\sbrac{\widehat \tau_T}-\tau
&=
2\bar D
\left(
    \int_0^{1/2}\theta\,d\theta
    -
    \int_{1/2}^1(1-\theta)\,d\theta
\right)
+
\int_0^{1/2}\theta R(\theta)\,d\theta
-
\int_{1/2}^1(1-\theta)R(\theta)\,d\theta.
\end{align*}
The first term is exactly zero because
\[
    \int_0^{1/2}\theta\,d\theta
    =
    \int_{1/2}^1(1-\theta)\,d\theta
    =
    \frac18.
\]
Therefore,
\begin{align*}
\abs{\mathbb{E}\sbrac{\widehat \tau_T}-\tau}
&\le
\left(
    \int_0^{1/2}\theta\,d\theta
    +
    \int_{1/2}^1(1-\theta)\,d\theta
\right)
\sup_{\theta\in[0,1]}\abs{R(\theta)}
\\
&=
O\pbrac{T^2\delta^2M}.
\end{align*}
This proves the bias bound. Since the argument used only the identity
$\mathbb{E}[\widehat{\tau}]=J'(1/2)$, the same bound applies to any unbiased
estimator of $\nabla J(1/2)$.

{\color{black}
For the variance, reversing the summations as in \eqref{eq:tpg-reward-time-sum} with $k=T$, we have
\begin{equation*}
 \widehat\tau_T=\frac1T\sum_{t=1}^T C_t,
 \qquad C_t:=Y_t\sum_{u=1}^t s_u,\qquad s_u=2(2Z_u-1).
\end{equation*}
The scores are independent and centered, and $|Y_t|\leq M$, so
\begin{equation*}
 \E[C_t^2]\leq M^2\E\!\left[\left(\sum_{u=1}^t s_u\right)^2\right]=4M^2t.
\end{equation*}
Cauchy--Schwarz applied to the finite sum, without any assumption on the dependence among the $C_t$, yields
\begin{equation*}
 \Var(\widehat\tau_T)
 \leq\E[\widehat\tau_T^2]
 \leq\frac1T\sum_{t=1}^T\E[C_t^2]
 \leq\frac{4M^2}{T}\sum_{t=1}^T t
 =2M^2(T+1).
\end{equation*}
This proves \eqref{eq:untrunc_variance_bound} without the mixing assumption. The constant-reward example in \eqref{eq:tpg-variance-sharpness} shows that linear growth in $T$ is attainable.
}\end{proof}
}

\subsection{Proof of Proposition~\ref{prop:trunc_policy_gradient}}
\begin{proof}[Proof of Proposition~\ref{prop:trunc_policy_gradient}]
The truncated policy value function $J_k(\theta)$ can be written as
\begin{equation*}
    J_k(\theta)= \frac{1}{T}\sum_{t=1}^{T} \E_{\mathcal{L}_\theta^{t, k}}\sbrac{Y_t}= \frac{1}{T}\sum_{t=1}^T \rho^{\top} P_{1/2}^{1 \rightarrow \max(t-k,1)} P_{\theta}^{\max(t-k,1) \rightarrow t}r_\theta.
\end{equation*}

Then, differentiating $J_k(\theta)$ yields
\begin{align*}
    &\textcolor{white}{\,\,=} \frac{d}{d\theta}J_k(\theta) \\
    &= \frac{1}{T}\sum_{t=1}^T \frac{d}{d\theta} \rho^{\top} P_{1/2}^{1 \rightarrow \max(t-k,1)} P_{\theta}^{\max(t-k,1) \rightarrow t}r_\theta\\
&= \frac{1}{T}\sum_{t=1}^T \rho^{\top} P_{1/2}^{1 \rightarrow \max(t-k,1)}\Bigg[\sum_{u=\max(t-k,1)}^{t-1} P_\theta^{\max(t-k,1) \rightarrow u} \Delta^u P_\theta^{(u+1) \rightarrow t}r_\theta +\; P_{\theta}^{\max(t-k,1) \rightarrow t}(r_1 - r_0)\Bigg].
\end{align*}

Evaluate at $\theta = 1/2$, one can observe
\begin{align*}
    \nabla J_k(1/2) &= \left.\frac{d}{d\theta} J_k(\theta)\right|_{\theta = \frac{1}{2}} \\
    &= \frac{1}{T}\sum_{t=1}^T \left[ \rho^{\top} P_{1/2}^{1 \rightarrow t}\left(r_1-r_0\right) + \sum_{u=\max(t-k,1)}^{t-1} \rho^{\top} P_{1/2}^{1 \rightarrow u} \Delta^u P_{1/2}^{(u+1) \rightarrow t} r_{1/2} \right]\\
    & = \frac{1}{T} \sum_{t=1}^T  \sum_{u=\max(t-k,1)}^{t} \left(\mathbb{E}_{\textcolor{black}{1/2}}\left[Y_t\mid Z_u=1\right] - \mathbb{E}_{\textcolor{black}{1/2}}\left[Y_t\mid Z_u=0\right] \right)\\
    & = \frac{1}{T}\sum_{u=1}^T  \sum_{t=u}^{\min (u+k, T)} \left(\mathbb{E}_{\textcolor{black}{1/2}}\left[Y_t\mid Z_u=1\right] - \mathbb{E}_{\textcolor{black}{1/2}}\left[Y_t\mid Z_u=0\right] \right)\\
    &= \frac{1}{T}\sum_{t=1}^T \left(Q_{{1/2}}^{t, k}(1) - Q_{{1/2}}^{t, k}(0) \right).
\end{align*}
\end{proof}

{\color{black}
\subsection{Proof of Proposition~\ref{prop:general-theta-tpg}}
\begin{proof}[Proof of Proposition~\ref{prop:general-theta-tpg}]
For any integrable random variable $H$,
\[
\E_\theta[s_\theta(Z_u)H]
=\E_\theta[H\mid Z_u=1]
-\E_\theta[H\mid Z_u=0].
\]
Applying this identity to the return in
\eqref{eq:general-theta-tpg} gives the first equality in
\eqref{eq:general-tpg-q}. For each $t$, differentiating the likelihood
of the assignments from $\max\{1, t-k\}$ through $t$ gives
\[
\left.\frac{\partial}{\partial q}
\bigl(\mu_\theta^{\max\{1, t-k\}}P_q^{\max\{1, t-k\}\rightarrow t}r_q\bigr)
\right|_{q=\theta}
=\sum_{u=\max\{1, t-k\}}^t\E_\theta[s_\theta(Z_u)Y_t].
\]
Averaging over $t$ and reversing the sums proves the second equality.
When $k=T$, $\max\{1, t-k\}=1$ for every $t$, so the truncated policy value is
$J(q)$.\end{proof}

\subsection{Proof of Proposition~\ref{prop:general-theta-tpg-bias-variance}}
\begin{proof}[Proof of Proposition~\ref{prop:general-theta-tpg-bias-variance}]
We first consider the bias. For each reward time $t$, define
$
h_t:=t-\max\{1,t-k\}.
$
Relative to the experimental kernel $P_\theta^u$, the treatment and control kernels satisfy
\[
P_1^u-P_\theta^u
=
(1-\theta)(P_1^u-P_0^u),
\qquad
P_0^u-P_\theta^u
=
-\theta(P_1^u-P_0^u).
\]
Applying the same endpoint expansion as in the proof of Theorem~\ref{them:trunc_DQ_bias}, we can write
$$
\E_{\mathcal L_1^t}[Y_t]
=
\phi_{1,\theta}^t
+
\varepsilon_{1,t}, \qquad 
\E_{\mathcal L_0^t}[Y_t]
=
\phi_{0,\theta}^t
+
\varepsilon_{0,t},
$$
where
\begin{align*}
\phi_{1,\theta}^t
&:=
\mu_\theta^t r_1
+
(1-\theta)
\sum_{u=\max\{1,t-k\}}^{t-1}
\mu_\theta^u
(P_1^u-P_0^u)
P_\theta^{(u+1)\rightarrow t}r_1,\\
\phi_{0,\theta}^t
&:=
\mu_\theta^t r_0
-
\theta
\sum_{u=\max\{1,t-k\}}^{t-1}
\mu_\theta^u
(P_1^u-P_0^u)
P_\theta^{(u+1)\rightarrow t}r_0.
\end{align*}
The same contraction argument as in the proof of Theorem~\ref{them:trunc_DQ_bias} gives
\begin{align*}
|\varepsilon_{1,t}|
&\leq
4(1-\theta)^2\delta^2M\Psi_{h_t}(\gamma)
+
\frac{2(1-\theta)\delta M\gamma^k}{1-\gamma},\\
|\varepsilon_{0,t}|
&\leq
4\theta^2\delta^2M\Psi_{h_t}(\gamma)
+
\frac{2\theta\delta M\gamma^k}{1-\gamma}.
\end{align*}

Subtracting the two first-order approximations gives
\begin{align*}
\phi_{1,\theta}^t-\phi_{0,\theta}^t
&=
\mu_\theta^t(r_1-r_0)+
\sum_{u=\max\{1,t-k\}}^{t-1}
\mu_\theta^u
(P_1^u-P_0^u)
P_\theta^{(u+1)\rightarrow t}
r_{1-\theta},
\end{align*}
where
$
r_{1-\theta}
=
(1-\theta)r_1+\theta r_0.
$

On the other hand, Proposition~\ref{prop:general-theta-tpg} gives
\begin{align*}
\E_\theta[
\widehat\tau^{\mathrm{TPG}}_{\theta,k}]
=
\frac1T\sum_{t=1}^T
\Bigg[
\mu_\theta^t(r_1-r_0)
+
\sum_{u=\max\{1,t-k\}}^{t-1}
\mu_\theta^u
(P_1^u-P_0^u)
P_\theta^{(u+1)\rightarrow t}
r_\theta
\Bigg].
\end{align*}
Since
$
r_{1-\theta}-r_\theta
=
(1-2\theta)(r_1-r_0),
$
the difference between the first-order endpoint (i.e., $\theta=1$ and $\theta = 0$) approximation and the TPG expectation is
\begin{align}
&\frac1T\sum_{t=1}^T
(\phi_{1,\theta}^t-\phi_{0,\theta}^t)
-
\E_\theta[
\widehat\tau^{\mathrm{TPG}}_{\theta,k}]
=
\frac{1-2\theta}{T}
\sum_{t=1}^T
\sum_{u=\max\{1,t-k\}}^{t-1}
\mu_\theta^u
(P_1^u-P_0^u)
P_\theta^{(u+1)\rightarrow t}
(r_1-r_0).
\label{eq:general-theta-tpg-extra-term}
\end{align}
Under Assumption~\ref{assumption:mixing-time},
\[
\left\|
\sum_{u=\max\{1,t-k\}}^{t-1}
\mu_\theta^u
(P_1^u-P_0^u)
P_\theta^{(u+1)\rightarrow t}
\right\|_1
\leq
\frac{2\delta(1-\gamma^k)}{1-\gamma}.
\]
Therefore, the absolute value of \eqref{eq:general-theta-tpg-extra-term} is bounded by
\[
\frac{2|1-2\theta|\delta(1-\gamma^k)}
{1-\gamma}
\|r_1-r_0\|_\infty.
\]
Combining this bound with the treatment and control approximation errors gives \eqref{eq:general-theta-tpg-bias-bound}. When $k\geq T-1$, all past assignments are included and the mixing-bias term vanishes.

We next consider the variance. Regrouping \eqref{eq:general-theta-tpg} by reward time gives
\[
\widehat\tau^{\mathrm{TPG}}_{\theta,k}
=
\frac1T\sum_{t=1}^T B_t,
\qquad
B_t
:=
Y_t
\sum_{u=\max\{1,t-k\}}^t
s_\theta(Z_u).
\]
The assignment scores are independent and centered, with
\[
\E_\theta[s_\theta(Z_u)^2]
=
\frac{1}{p_\theta}.
\]
Since $|Y_t|\leq M$,
\[
\E_\theta[B_t^2]
\leq
\frac{M^2(k+1)}{p_\theta}
=:V_{\theta,k}.
\]
By Cauchy--Schwarz,
\[
|\operatorname{Cov}_\theta(B_i,B_j)|
\leq
V_{\theta,k}
\qquad
\text{for all }i,j.
\]

For $j-i\geq k+2$, the two reward windows are separated. The same mixing argument as in the proof of Theorem~\ref{them:trunc_DQ_bias} gives
\[
|\operatorname{Cov}_\theta(B_i,B_j)|
\leq
2V_{\theta,k}
\gamma^{j-i-k-1}.
\]
Summing the covariance bounds for overlapping and separated windows yields
\begin{align*}
\operatorname{Var}_\theta
\left(\widehat\tau^{\mathrm{TPG}}_{\theta,k}\right)
&\leq
\frac{V_{\theta,k}}{T}
\left[
1+2(k+1)
+
4\sum_{h=k+2}^\infty
\gamma^{h-k-1}
\right]\\
&=
\frac{V_{\theta,k}}{T}
\left(
2k+3+\frac{4\gamma}{1-\gamma}
\right).
\end{align*}
Substituting the definition of $V_{\theta,k}$ gives
\eqref{eq:general-theta-tpg-var-bound}.
\end{proof}

\subsection{Proof of Proposition~\ref{prop:rc-bias-variance}}

\begin{proof}[Proof of Proposition~\ref{prop:rc-bias-variance}]
We use the $\ell_1$ norm for signed row vectors. By Assumption~\ref{assumption:mixing-time}, any zero-mass signed row $\xi$ satisfies
\[
\|\xi P_q^{a\rightarrow b}\|_1
\leq
\gamma^{b-a}\|\xi\|_1,
\qquad q\in[0,1].
\]
Moreover, by the definition of $\delta$,
\[
\|\eta\Delta^i\|_1
\leq
2\delta\|\eta\|_1
\]
for any signed row $\eta$, where $\Delta^i=P_1^i-P_0^i$.

\emph{Mixing bias.}
Recall $\nu_q^t$ from \eqref{eq:rc-window-occupancy}, and let
$
h_t:=t-\max\{1,t-k\}
$
denote the number of retained transitions for reward time $t$. Define the {\em hypothetical truncated contrast}
\[
\tau_{\theta,k}^{\mathrm{hypo}}
:=
J_{k,\theta}(1)-J_{k,\theta}(0)
=
\frac1T\sum_{t=1}^T
\left(\nu_1^t r_1-\nu_0^t r_0\right).
\]
For each reward time $t$, this contrast applies the treatment and control policies within the retained window, starting from the state distribution under the experimental policy at time $\max\{1,t-k\}$.

For $z\in\{0,1\}$, telescoping the transition products gives
\[
\left\|
\mu_z^{\max\{1,t-k\}}
-
\mu_\theta^{\max\{1,t-k\}}
\right\|_1
\leq
2|z-\theta|\delta
\sum_{i=1}^{\max\{1,t-k\}-1}
\gamma^{\max\{1,t-k\}-i-1}.
\]
Hence,
\[
\left\|
\mu_z^{\max\{1,t-k\}}
-
\mu_\theta^{\max\{1,t-k\}}
\right\|_1
\leq
\frac{
2|z-\theta|\delta
\left(1-\gamma^{\max\{1,t-k\}-1}\right)
}{1-\gamma}.
\]
The remaining $h_t$ transitions under the endpoint policy contract this difference by $\gamma^{h_t}$. Since $\|r_z\|_\infty\leq M$, averaging over $t$ gives
\begin{equation}
\left|
\tau-\tau_{\theta,k}^{\mathrm{hypo}}
\right|
\leq
\frac{2\delta M}{T(1-\gamma)}
\sum_{t=1}^T
\gamma^{h_t}
\left(1-\gamma^{\max\{1,t-k\}-1}\right)
\leq
\frac{2\delta M\gamma^k}{1-\gamma}.
\label{eq:rc-prefix-bias}
\end{equation}
For the last inequality, the summand is zero when $\max\{1,t-k\}=1$. Otherwise, $h_t=k$. When $k\geq T-1$, all summands are zero, so the mixing bias vanishes.

\emph{Taylor error.}
The prefix distribution $\mu_\theta^{\max\{1,t-k\}}$ is held fixed when differentiating $\nu_q^t$ with respect to $q$. Applying Taylor's formula to the treatment and control state distributions gives
\begin{equation}
\tau_{\theta,k}^{\mathrm{hypo}}-\E_\theta[\widehat\tau^{\mathrm{RC}}_{\theta,k}]
=
\frac1T\sum_{t=1}^T
\left[
\int_\theta^1
(1-q)(\nu_q^t)''r_1\,dq
-
\int_0^\theta
q(\nu_q^t)''r_0\,dq
\right].
\label{eq:rc-window-remainder}
\end{equation}
The second derivative is
\[
(\nu_q^t)''
=
2
\sum_{\max\{1,t-k\}\leq i<j<t}
\mu_\theta^{\max\{1,t-k\}}
P_q^{\max\{1,t-k\}\rightarrow i}
\Delta^i
P_q^{(i+1)\rightarrow j}
\Delta^j
P_q^{(j+1)\rightarrow t}.
\]
Applying the contraction property after each kernel difference gives
\[
\sup_{q\in[0,1]}
\|(\nu_q^t)''\|_1
\leq
8\delta^2
\sum_{\max\{1,t-k\}\leq i<j<t}
\gamma^{t-i-2}
=
8\delta^2\Psi_{h_t}(\gamma).
\]
Using $\|r_1\|_\infty,\|r_0\|_\infty\leq M$ in \eqref{eq:rc-window-remainder}, we obtain
\begin{align*}
\left|
\tau_{\theta,k}^{\mathrm{hypo}}-\E_\theta[\widehat\tau^{\mathrm{RC}}_{\theta,k}]
\right|
&\leq
\frac{4\delta^2M}{T}
\left[\theta^2+(1-\theta)^2\right]
\sum_{t=1}^T\Psi_{h_t}(\gamma)\\
&\leq
4\left[\theta^2+(1-\theta)^2\right]
\delta^2M\Psi_k(\gamma).
\end{align*}
Combining this bound with \eqref{eq:rc-prefix-bias} proves \eqref{eq:rc-bias-bound}.

\emph{Variance.}
Write
\[
s_t:=s_\theta(Z_t),
\qquad
w_t:=w_\theta(Z_t).
\]
Regrouping the estimator by reward time gives
$$
\widehat\tau^{\mathrm{RC}}_{\theta,k}
=
\frac1T\sum_{t=1}^T B_t^{\mathrm{RC}},
\qquad
B_t^{\mathrm{RC}}
:=
Y_t(s_t+w_tS_t),
\qquad
S_t
:=
\sum_{u=\max\{1,t-k\}}^{t-1}s_u.
$$
Let $\E_{\theta,\max\{1,t-k\},x}$ denote expectation under the experimental process started from
\[
X_{\max\{1,t-k\}}=x.
\]
The assignment scores within this window are independent and centered, with
\[
\E_\theta[s_t^2]
=
p_\theta^{-1}.
\]
Therefore,
\[
\E[S_t]=0,
\qquad
\E[S_t^2]
=
\frac{h_t}{p_\theta}.
\]
Moreover, $(s_t,w_t)$ is independent of $S_t$. Using $|Y_t|\leq M$ and $\E_\theta[w_t^2]=c_\theta$, we obtain
\begin{align*}
\E_{\theta,\max\{1,t-k\},x}
\left[(B_t^{\mathrm{RC}})^2\right]
&\leq
M^2
\E_{\theta,\max\{1,t-k\},x}
\left[(s_t+w_tS_t)^2\right]\notag\\
&=
\frac{M^2}{p_\theta}
\left(1+h_tc_\theta\right)
\leq
V_{\theta,k},
\end{align*}
where
\[
V_{\theta,k}
:=
\frac{M^2}{p_\theta}(1+kc_\theta).
\]
This bound is uniform over starting states. By Cauchy--Schwarz,
\[
\left|
\operatorname{Cov}_\theta
(B_i^{\mathrm{RC}},B_j^{\mathrm{RC}})
\right|
\leq
V_{\theta,k}
\qquad
\text{for all }i,j.
\]

We obtain a sharper bound when the two reward windows are separated. Suppose $j-i\geq k+2$, and define
\[
g_j(x)
:=
\E_{\theta,\max\{1,j-k\},x}
[B_j^{\mathrm{RC}}],
\qquad
\mathcal F_i
:=
\sigma\bigl((X_u,Z_u,Y_u)_{1\leq u\leq i},X_{i+1}\bigr).
\]
The moment bound above implies
\[
\|g_j\|_\infty
\leq
\sqrt{V_{\theta,k}}.
\]
Since $j-i\geq k+2$, we have $\max\{1,j-k\}=j-k\geq i+2$. The Markov and reward assumptions therefore give
\[
\E_\theta[B_j^{\mathrm{RC}}\mid\mathcal F_i]
=
\left(
P_\theta^{(i+1)\rightarrow j-k}g_j
\right)(X_{i+1}).
\]
Applying Assumption~\ref{assumption:mixing-time} to the conditional and unconditional distributions of $X_{j-k}$ yields
\[
\left|
\E_\theta[B_j^{\mathrm{RC}}\mid\mathcal F_i]
-
\E_\theta[B_j^{\mathrm{RC}}]
\right|
\leq
2\sqrt{V_{\theta,k}}
\,\gamma^{j-i-k-1}.
\]
Since $B_i^{\mathrm{RC}}$ is $\mathcal F_i$-measurable and
\[
\E_\theta
\left|
B_i^{\mathrm{RC}}
-
\E_\theta B_i^{\mathrm{RC}}
\right|
\leq
\sqrt{V_{\theta,k}},
\]
we have
\[
\left|
\operatorname{Cov}_\theta
(B_i^{\mathrm{RC}},B_j^{\mathrm{RC}})
\right|
\leq
2V_{\theta,k}
\gamma^{j-i-k-1},
\qquad
j-i\geq k+2.
\]

Finally, summing the covariance bounds for overlapping and separated reward windows gives
\begin{align*}
\operatorname{Var}_\theta
(\widehat\tau^{\mathrm{RC}}_{\theta,k})
&\leq
\frac{V_{\theta,k}}{T}
\left[
1+2(k+1)
+
4\sum_{h=k+2}^\infty
\gamma^{h-k-1}
\right]\\
&=
\frac{V_{\theta,k}}{T}
\left(
2k+3+\frac{4\gamma}{1-\gamma}
\right).
\end{align*}
Substituting the definition of $V_{\theta,k}$ proves \eqref{eq:rc-variance-bound}.
\end{proof}
}

{\color{black}
\section{Proof of Lemmas}
\label{app:lemma}
\subsection{Proof of Lemma~\ref{lemma:trunc_DQ_CLT}}
\begin{definition}[Dobrushin, ergodic, and minimal ergodic coefficients]\label{def:Dob_coef}
Let $K$ be a Markov transition kernel on a measurable state space $(\mathcal{X},\mathcal{B})$.
The \emph{Dobrushin contraction coefficient} of $K$ is
$$
\tilde \gamma(K) \;:=\; \sup_{x_1,x_2\in\mathcal{X}}\ \sup_{B\in\mathcal{B}}
\big|\,K(x_1,B)-K(x_2,B)\,\big|,
$$
and the associated \emph{ergodic coefficient} is $\tilde \alpha(K) := 1 - \tilde \gamma(K)$.
For each $n\geq 1$, let $\{X_{n,i}: 1 \leq i \leq N_n\}$ be a finite temporally nonhomogeneous Markov chain
with one-step kernels $K^{(n)}_{i,i+1}(x,\cdot)$ satisfying
$$
\mathbb{P}\!\left( X_{n,i+1} \in \cdot \,\middle|\, X_{n,1},\ldots,X_{n,i} \right)
= K^{(n)}_{i,i+1}(X_{n,i}, \cdot).
$$
The \emph{minimal ergodic coefficient} of the $n$-th row is then
$$
\tilde \alpha_n \;:=\; \min_{1 \leq i < N_n} \tilde \alpha\!\left(K^{(n)}_{i,i+1}\right)
\;=\; 1 - \max_{1 \leq i < N_n} \tilde \gamma\!\left(K^{(n)}_{i,i+1}\right).
$$
\end{definition}

\begin{lemma}[CLT for temporally nonhomogeneous Markov chains {\normalfont in \cite{arlotto2016central}}]
\label{thm:GO-CLT}
Fix $m \geq 0$.
Let $(\mathcal{X},\mathcal{B})$ be a Borel state space.
For each $n\geq 1$, let $\{X_{n,i}\}_{i=1}^{n+m}$
be as in Definition~\ref{def:Dob_coef},
and let $\{f_{n,i}\}_{i=1}^n$ be bounded measurable functions
(i.e., $f_{n,i}: \mathcal{X}^{1+m} \rightarrow \mathbb{R}$) with
$$
S_n\;:=\;\sum_{i=1}^n f_{n, i}\left(X_{n, i}, \ldots, X_{n, i+m}\right),\qquad
\max_{1\leq i\leq n}\ \|f_{n,i}\|_\infty \;\le\; C_n,
$$
where $C_n$ may depend on $n$. Suppose the minimal ergodic coefficients satisfy $\tilde \alpha_n>0$ and
$$
C_n^2\,\tilde \alpha_n^{-2}\;=\;{o}\left(\operatorname{Var}(S_n)\right),\quad\text{as }n\to\infty.
$$
Then, we have
$$
\frac{S_n-\E[S_n]}{\sqrt{\operatorname{Var}(S_n)}}
\ \Rightarrow\ \mathcal{N}(0,1),
\quad\text{as }n\to\infty.
$$
\end{lemma}

\begin{proof}[Proof of Lemma~\ref{lemma:trunc_DQ_CLT}]
    The proof follows by recasting our nonstationary MDP under the experimenting policy $\pi_{1/2}$ as a temporally nonhomogeneous Markov chain and verifying the conditions of Lemma~\ref{thm:GO-CLT}.

    Let
    $$
    G(x,z,u) := F^{-1}_{R(\cdot\mid x,z)}(u),
    $$
    where $F^{-1}_{R(\cdot\mid x,z)}$ denotes a generalized inverse CDF. {\color{black}Choose its values at $u=0,1$ in $[-M,M]$, so $G$ is bounded everywhere.} Using independent
    $U_t\sim \operatorname{Unif}[0,1]$, the reward can be represented as
    $$
    Y_t = G(X_t,Z_t,U_t).
    $$

    For each horizon $T$, define the shifted augmented row
    $$
    \widetilde X_{T,t} := (X_t,Z_{t-1},U_{t-1}), \qquad 1\leq t\leq T+1,
    $$
    where $Z_0$ and $U_0$ are arbitrary dummy variables independent of everything else. For
    $t>T+1$, extend the row to length $T+k+1$ using any fixed transition kernel that is independent
    of the current state; the functions below do not use these additional variables.

    For $1\leq t\leq T$, the one-step transition kernel of $\widetilde X_{T,t}$ is, for
    $w=(x,z_-,u_-)\in \mathcal X\times\{0,1\}\times[0,1]$ and measurable
    $A\subseteq \mathcal X\times\{0,1\}\times[0,1]$,
    $$
    \widetilde K_t(w,A)
    :=
    \frac12\sum_{z=0}^1 \int_0^1
    \sum_{y\in\mathcal X}
    \mathbf 1\{(y,z,u)\in A\} P^t_z(x,y)\,du .
    $$
    This shifted augmentation is Markov because, conditional on $X_t$, the fresh action $Z_t$ and
    reward noise $U_t$ are drawn independently, and then $X_{t+1}$ is generated from
    $P^t_{Z_t}(X_t,\cdot)$.

    We now verify the Dobrushin coefficient of $\widetilde K_t$. Fix
    $w=(x,z_-,u_-)$ and $w'=(x',z'_-,u'_-)$, and for each $z,u$ define the section
    $$
    A_{z,u}:=\{y\in\mathcal X:(y,z,u)\in A\}.
    $$
    Then, for $1\leq t\leq T$,
    $$
    \begin{aligned}
    \big|\widetilde K_t(w,A)-\widetilde K_t(w',A)\big|
    &\leq
    \frac12\sum_{z=0}^1\int_0^1
    \big|P^t_z(x,A_{z,u})-P^t_z(x',A_{z,u})\big|\,du  \\
    &\leq
    \frac12\sum_{z=0}^1 \gamma
    \;=\;\gamma,
    \end{aligned}
    $$
    where the second inequality follows from Assumption~\ref{assumption:mixing-time} applied to point
    masses. For the artificially extended kernels after time $T$, the Dobrushin coefficient is zero
    because those kernels are independent of the current state. Hence the minimal ergodic coefficient
    of the row satisfies
    $$
    \tilde \alpha_T \geq 1-\gamma>0.
    $$

    We next define the local functions. For local arguments
    $w_\ell=(x_\ell,a_\ell,u_\ell)\in \mathcal X\times\{0,1\}\times[0,1]$,
    $\ell=0,\ldots,k+1$, define
    $$
    f_{T,i}(w_0,\ldots,w_{k+1})
    :=
    2(2a_1-1)
    \sum_{j=0}^k
    \mathbf 1\{i+j\leq T\}
    G(x_j,a_{j+1},u_{j+1}).
    $$
    Since $w_\ell$ corresponds to $\widetilde X_{T,i+\ell}$, this gives
    $$
    f_{T,i}\left(\widetilde X_{T,i},\ldots,\widetilde X_{T,i+k+1}\right)
    =
    2(2Z_i-1)\sum_{j=0}^k \mathbf 1\{i+j\leq T\}Y_{i+j}.
    $$
    Therefore, with
    $$
    S_T := \sum_{i=1}^T
    f_{T,i}\left(\widetilde X_{T,i},\ldots,\widetilde X_{T,i+k+1}\right),
    $$
    we have
    $
    S_T = T\widehat \tau_k.
    $

    We now verify the remaining conditions of Lemma~\ref{thm:GO-CLT}, applying it with
    $n=T$ and $m=k+1$. Since $|Y_t|\leq M$ almost surely,
    $$
    \max_{1\leq i\leq T}\|f_{T,i}\|_\infty \leq 2(k+1)M.
    $$
    Moreover, by the assumed existence of the asymptotic variance,
    $$
    \frac{1}{T}\operatorname{Var}(S_T)\to \sigma_k^2 > 0.
    $$
    Then
    $\operatorname{Var}(S_T)\sim T\sigma_k^2\to\infty$, and
    $$
    C_T^2\tilde\alpha_T^{-2}
    \leq
    4(k+1)^2M^2(1-\gamma)^{-2}
    =
    o\!\left(\operatorname{Var}(S_T)\right).
    $$
    Lemma~\ref{thm:GO-CLT} therefore yields
    $$
    \frac{S_T-\E[S_T]}{\sqrt{\operatorname{Var}(S_T)}}
    \Rightarrow
    \mathcal N(0,1).
    $$
    Since $\operatorname{Var}(S_T)/T\to\sigma_k^2$, Slutsky's theorem gives
    $$
    \sqrt{T}\,\big(\widehat\tau_k-\E[\widehat\tau_k]\big)
    =
    \frac{S_T-\E[S_T]}{\sqrt{\operatorname{Var}(S_T)}}
    \sqrt{\frac{\operatorname{Var}(S_T)}{T}}
    \Rightarrow
    \mathcal N(0,\sigma_k^2).
    $$
\end{proof}

\subsection{Proof of Lemma~\ref{lem:HAC_consistency}}
{\color{black}\begin{proof}[Proof of Lemma~\ref{lem:HAC_consistency}]
Reindexing gives $T\widehat\tau_k=\sum_{t=1}^T B_t$. Set
\[
 K=4(k+1)M,\qquad
 a(h)=\gamma^{\max\{0,h-k-1\}}\quad(h\geq0),\qquad
 A=\sum_{h=1}^{\infty}a(h)<\infty.
\]
Then $\E[V_t]=0$, $|V_t|\leq K$, and the strong-mixing coefficients of
$\{V_t\}$ are bounded by $a(h)$, uniformly over the horizon. Indeed,
$B_t$ depends only on the observations from $\max\{1,t-k\}$ through $t$,
so Lemma~\ref{lemma:alpha-mixing} applies to two separated observation windows.
Constants below may depend on the fixed $k$, $M$, and $\gamma$, but not on
$T$ or the lag.

First, we prove a uniform second-moment bound for the ideal sample
autocovariances. For $0\leq\ell<T$, write $U_{t,\ell}=V_tV_{t+\ell}$.
For $\ell\geq1$ and $1\leq h<\ell$, the four time indices are
$t<t+h<t+\ell<t+h+\ell$. Centering the first factor and applying
Lemma~\ref{lemma3_mixing} give
\[
 \left|\E[V_tV_{t+\ell}V_{t+h}V_{t+h+\ell}]\right|
 \leq4K^4 a(h).
\]
Also $|\E[V_tV_{t+\ell}]|\leq4K^2a(\ell)$. Consequently,
\[
 |\operatorname{Cov}(U_{t,\ell},U_{t+h,\ell})|
 \leq4K^4a(h)+16K^4a(\ell)^2\qquad(1\leq h<\ell).
\]
For $h>\ell$, the two pairs are separated by $h-\ell$, and hence
\[
 |\operatorname{Cov}(U_{t,\ell},U_{t+h,\ell})|
 \leq4K^4a(h-\ell).
\]
This last bound remains valid at $h=\ell$ by the elementary covariance
bound and the convention $a(0)=1$. Since $a$ is nonincreasing and bounded
by one,
\[
 \ell a(\ell)^2\leq\sum_{h=1}^{\ell}a(h)\leq A.
\]
The sum over positive lags $h$ of the absolute off-diagonal covariances is therefore bounded
by $K^4(4+24A)$, uniformly in $t$ and $\ell\geq1$.
At $\ell=0$, apply the mixing inequality directly to $V_t^2$ and
$V_{t+h}^2$ to obtain the same conclusion. Including the diagonal term
$\Var(U_{t,\ell})\leq K^4$ yields
\[
 \sup_{0\leq\ell<T}\Var(\widetilde\Gamma_\ell)
 \leq\frac{K^4(9+48A)}{T}.
\]
By the triangle inequality in $L^2$ and $0\leq w^{\mathrm{NW}}_\ell\leq1$,
\[
 \big\|\widetilde\Omega_T-\E[\widetilde\Omega_T]\big\|_2
 \leq (2L_T+1)\frac{K^2\sqrt{9+48A}}{\sqrt T}
 =o(1),
\]
where $L_T=O(T^{1/3})$ is used in the last step.

It remains to identify the expectation. Define
\[
 \Omega_T=\frac1T\Var\!\left(\sum_{t=1}^T B_t\right)
 =\E[\widetilde\Gamma_0]
  +2\sum_{\ell=1}^{T-1}\E[\widetilde\Gamma_\ell].
\]
For all sufficiently large $T$, $L_T<T$. Since
$|\E[\widetilde\Gamma_\ell]|\leq4K^2a(\ell)$ for $\ell\geq1$,
\[
 |\E[\widetilde\Omega_T]-\Omega_T|
 \leq8K^2\left[
  \frac1{L_T+1}\sum_{\ell=1}^{L_T}\ell a(\ell)
  +\sum_{\ell>L_T}a(\ell)
 \right]\longrightarrow0.
\]
Here the first term vanishes because $\sum_{\ell\geq1}\ell a(\ell)<\infty$,
and the second because $L_T\to\infty$. Assumption~\ref{ass:asymptotic_variance}
gives $\Omega_T\to\sigma_k^2$. Combining the stochastic and deterministic
bounds proves $\widetilde\Omega_T\xrightarrow{p}\sigma_k^2$.
\end{proof}
}

\subsection{Proof of Lemma~\ref{lemma:HAC_consistency}}
\begin{proof}[Proof of Lemma~\ref{lemma:HAC_consistency}]
Recall that $\bar B := T^{-1}\sum_{t=1}^T B_t$, $\mu_t := \E[B_t]$, and $\bar\mu_T := T^{-1}\sum_{t=1}^T \mu_t$.
Write
$$
\bar V_T:=\frac1T\sum_{t=1}^T V_{t}=\bar B-\bar\mu_T,\qquad
\delta_{t}:=\bar\mu_T-\mu_t.
$$
Then
$$
\widehat V_{t}=B_t-\bar B=(B_t-\mu_t)-(\bar B-\bar\mu_T)-(\bar\mu_T-\mu_t)
=V_{t}-\bar V_T-\delta_{t}.
$$
Fix $\ell\geq 0$. Expanding $\widehat\Gamma_\ell-\widetilde\Gamma_\ell$ gives
\begin{align}
\widehat\Gamma_\ell-\widetilde\Gamma_\ell
&=\frac1T\sum_{t=1}^{T-\ell}\Big[(V_{t}-\bar V_T-\delta_{t})(V_{t+\ell}-\bar V_T-\delta_{t+\ell})-V_{t}V_{t+\ell}\Big]\notag\\
&=\underbrace{-\frac{\bar V_T}{T}\!\sum_{t=1}^{T-\ell}V_{t}
              -\frac{\bar V_T}{T}\!\sum_{t=1}^{T-\ell}V_{t+\ell}
              +\frac{T-\ell}{T}\bar V_T^2}_{=:A_\ell}\label{eq:A_term}\\
&\quad\textcolor{black}{+}\underbrace{-\frac1T\!\sum_{t=1}^{T-\ell}\delta_{t}\,V_{t+\ell}
                   -\frac1T\!\sum_{t=1}^{T-\ell}\delta_{t+\ell}\,V_{t}
                   +\frac{\bar V_T}{T}\!\sum_{t=1}^{T-\ell}(\delta_{t}+\delta_{t+\ell})}_{=:B_\ell}\label{eq:B_term}\\
&\quad+\underbrace{\frac1T\!\sum_{t=1}^{T-\ell}\delta_{t}\,\delta_{t+\ell}}_{=:C_\ell}.\label{eq:C_term}
\end{align}

We bound these uniformly over $0\leq \ell\leq L_T$.

\emph{(i) Bound for $A_\ell$.}
Since $\sqrt{T}\,\bar V_T=O_p(1)$ by Lemma~\ref{lemma:trunc_DQ_CLT} and $|V_{t}|\leq 4(k+1) M$ deterministically for any $\textcolor{black}{1}\leq t \leq T$,
$$
\frac1T\sum_{t=1}^{T-\ell}V_{t}=\bar V_T+O\!\left(\frac{\ell}{T}\right),
\qquad
\frac1T\sum_{t=1}^{T-\ell}V_{t+\ell}=\bar V_T+O\!\left(\frac{\ell}{T}\right),
$$
hence $A_\ell=O_p\left(T^{-1}\right)+O_p\left(\frac{L_T}{T^{3 / 2}}\right)+O_p\left(\frac{L_T}{T^2}\right)$. Since $L_T = O(T^{1/3})$, $\sup_{\ell\leq L_T}|A_\ell|=O_p(T^{-1})$.

\emph{(ii) Bound for $B_\ell$.}
For the two mixed terms use Cauchy--Schwarz and $\sum_{t=1}^T V_{t}^2\leq 16(k+1)^2 M^2 T$:
$$
\left|\frac1T\sum_{t=1}^{T-\ell}\delta_{t}V_{t+\ell}\right|
\leq \frac1T\Big(\sum_{t=1}^{T}\delta_{t}^2\Big)^{1/2}\Big(\sum_{t=1}^{T}V_{t}^2\Big)^{1/2}
\leq 4(k+1)M\,\Delta_{T},
$$
and similarly for $\frac1T\sum \delta_{t+\ell}V_{t}$. For the part with $\bar V_T$,
$$
\left|\frac{\bar V_T}{T}\sum_{t=1}^{T-\ell}(\delta_{t}+\delta_{t+\ell})\right|
\leq \frac{|\bar V_T|}{T}\,\bigg(\Big|\sum_{t=1}^{T-\ell}\delta_{t}\Big|+\Big|\sum_{t=1}^{T-\ell}\delta_{t+\ell}\Big|\bigg).
$$
Because $\sum_{t=1}^T\delta_{t}=0$ and $|\delta_{t}|\leq 4(k+1)M<\infty$, each partial sum is $O(L_T)$, so this term is $O_p\left(\frac{L_T}{T^{3 / 2}}\right)$. Therefore,
$$
\sup_{\ell\leq L_T}|B_\ell| \;\le\; 8(k+1)M\,\Delta_T \;+\; O_p\!\left(\frac{L_T}{T^{3/2}}\right)
\;=\; O_p\left(\Delta_T\right) + O_p\left(T^{-7 / 6}\right).
$$

\emph{(iii) Bound for $C_\ell$.}
Again by Cauchy--Schwarz and $\sum_{t=1}^{T-\ell}\delta_{t+\ell}^2\leq \sum_{t=1}^{T}\delta_{t}^2$,
$$
|C_\ell|\leq \frac{1}{T}\Big(\sum_{t=1}^{T}\delta_{t}^2\Big)^{1/2}\Big(\sum_{t=1}^{T}\delta_{t}^2\Big)^{1/2}
=\Delta_T^2,
$$
so $\sup_{\ell\leq L_T}|C_\ell|= O_p(\Delta_T^2)$.
Combining (i)--(iii), uniformly for $0\leq \ell\leq L_T$,
$$
\widehat\Gamma_\ell-\widetilde\Gamma_\ell
= O_p(T^{-1}) \;+\; O_p(\Delta_T) \;+\; O_p(\Delta_T^2).
$$
Using $0\leq w^{\mathrm{NW}}_\ell\leq 1$,
$$
\widehat\Omega_T-\widetilde\Omega_T
=\big(\widehat\Gamma_0-\widetilde\Gamma_0\big)
+2\sum_{\ell=1}^{L_T}w^{\mathrm{NW}}_\ell(\widehat\Gamma_\ell-\widetilde\Gamma_\ell)
= O_p\!\Big(\frac{L_T}{T}\Big)\;+\;O_p\!\big(L_T\Delta_T\big)\;+\;O_p\!\big(L_T\Delta_T^2\big).
$$
Together with $L_T = O(T^{1/3})$ and $\Delta_T^2 = o(T^{-2/3})$ in Assumption~\ref{ass:mean-drift}, we have $L_T/T\to0$, $L_T\Delta_T\to0$, and $L_T\Delta_T^2\to0$. 
Hence $\widehat\Omega_T-\widetilde\Omega_T = o_p(1)$.
\end{proof}

}
\section{Experimental details}
\label{app:exp}

This section provides a detailed description of the simulation setup for Section~\ref{sec:numerical}; additional implementation details are available in the accompanying code. All experiments were run on a PC with a 12-core CPU and 18 GB of RAM.

\subsection{Two-state non-stationary MDP}
\label{Appendix:two_state}

{\color{black}
In the experiment with a two-state non-stationary MDP, we consider a finite-horizon setting with state space $\mathcal{X}=\{0,1\}$, action space $\mathcal{Z}=\{0,1\}$, and horizon length $T=5000$.
The environment evolves according to time-varying transition kernels $P_z^t$ for $z\in\{0,1\}$ and $t\in[T]$, constructed through a mean-reverting autoregressive process with additive i.i.d. Gaussian noise and a specified uniform mixing bound $0\leq\gamma<1$.
Specifically, for each state $x\in\{0,1\}$, we first form the control-kernel proposal
\begin{equation*}
    \widetilde P_0^t(x,\cdot)
    = C\left(\alpha P_0^{t-1}(x,\cdot)
    +(1-\alpha)\mu_x+\varepsilon_{t,x}\right),
    \qquad t\geq 1,
\end{equation*}
where $\alpha\in[0,1]$ denotes the mean reversion rate, $\mu_x$ is the target distribution of next-state transitions from state $x$, and the noise vectors $\varepsilon_{t,x}\sim\mathcal{N}(0,\sigma_\epsilon^2 I_2)$ are independent across times and states. The operator $C$ clips each component to $[0.01,0.99]$ and normalizes the resulting vector to sum to one.

To ensure the uniform mixing bound $\gamma$, we further project the two control rows so that their total variation distance is at most $\gamma$.
Specifically, letting $p_x=\widetilde P_0^t(x,1)$ and $\bar d=\max\{-\gamma,\min\{\gamma,p_1-p_0\}\}$, we set
\begin{equation*}
    P_0^t(0,1)=\frac{p_0+p_1-\bar d}{2},
    \qquad
    P_0^t(1,1)=\frac{p_0+p_1+\bar d}{2},
\end{equation*}
with $P_0^t(x,0)=1-P_0^t(x,1)$. This projection preserves valid probabilities and leaves the transition rows unchanged when their total variation distance is already at most $\gamma$.
The projected kernels are used in the next autoregressive update. The treatment transition kernel is constructed by shifting the control kernel toward state 1 or state 0, depending on the sign of $b$:
\begin{equation*}
    P_1^t(x,\cdot)
    =C\left(P_0^t(x,\cdot)+(-b,b)\right).
\end{equation*}
Equivalently, $P_1^t(x,1)=\min\{0.99,\max\{0.01,P_0^t(x,1)+b\}\}$.
Since scalar clipping is $1$-Lipschitz, this treatment transformation preserves the row-separation bound. Thus both action kernels, and every state-independent Bernoulli mixture $P_\theta^t=(1-\theta)P_0^t+\theta P_1^t$ with $\theta\in[0,1]$, satisfy Assumption~\ref{assumption:mixing-time} with bound $\gamma$. The bound need not be tight. The maximum kernel deviation satisfies
\begin{equation*}
    \delta
    :=\max_{t\in[T],\,x\in\{0,1\}}
    |P_1^t(x,1)-P_0^t(x,1)|
    \leq |b|.
\end{equation*}
The computed maximum equals $|b|$ for every reported kernel sequence. Rewards are generated as $Y_t=r(X_t,Z_t)+\eta_t$, where $\eta_t\overset{\mathrm{i.i.d.}}{\sim}\mathcal{N}(0,\sigma_r^2)$ independently of the initial state, treatment assignments, and transition randomness.

In the simulation, we set the mean reversion rate to $\alpha=0.5$, the standard deviations of the Gaussian innovations and reward noise to $\sigma_\epsilon=\sigma_r=0.1$, and the mean rewards to $r(x,z)=5x+z$.
The target distributions $\mu_0$ and $\mu_1$ are a random ordering of $[(1+\gamma)/2,(1-\gamma)/2]$ and $[(1-\gamma)/2,(1+\gamma)/2]$, ensuring a total variation distance of $\gamma$ between the two target distributions. We initialize $P_0^0(x,\cdot)=\mu_x$ and generate each kernel sequence with random seed 42, holding it fixed across trials.
Due to clipping and projection, the target distributions do not necessarily coincide with the long-run mean transition probabilities. The initial state $X_1$ is uniformly distributed over $\{0,1\}$, and treatments are independently assigned according to $Z_t\sim\operatorname{Bernoulli}(1/2)$.
After observing the reward at time $t$, we draw $X_{t+1}\sim P_{Z_t}^t(X_t,\cdot)$. The true GATE is computed exactly by propagating the state distributions under the all-treatment and all-control policies and averaging their expected reward difference over the horizon. All experiments are repeated over 1,000 independent trials conditional on the kernel sequence, using the same experimental trajectory for all truncation sizes. We note that the Gaussian rewards do not satisfy the almost-sure boundedness assumption imposed in our theoretical analysis.

For the main experiment, we set $b=0.1$ and vary the uniform mixing bound $\gamma$ over a grid obtained by rounding 20 equally spaced values between $0.01$ and $0.99$ to two decimal places. Figure~\ref{fig:numerical} further illustrates estimator performance under $\gamma\in\{0.1,0.3,0.5,0.7\}$ and different signed kernel shifts $b\in\{0.1,0.3,-0.1,-0.3\}$.
}

\begin{figure}[p]
\centering
\begin{minipage}{\linewidth}
\centering
\subfloat[{\color{black}$b=0.1, \gamma=0.1$}\label{fig:two_states_sub_a}]{\includegraphics[width=0.225\linewidth]{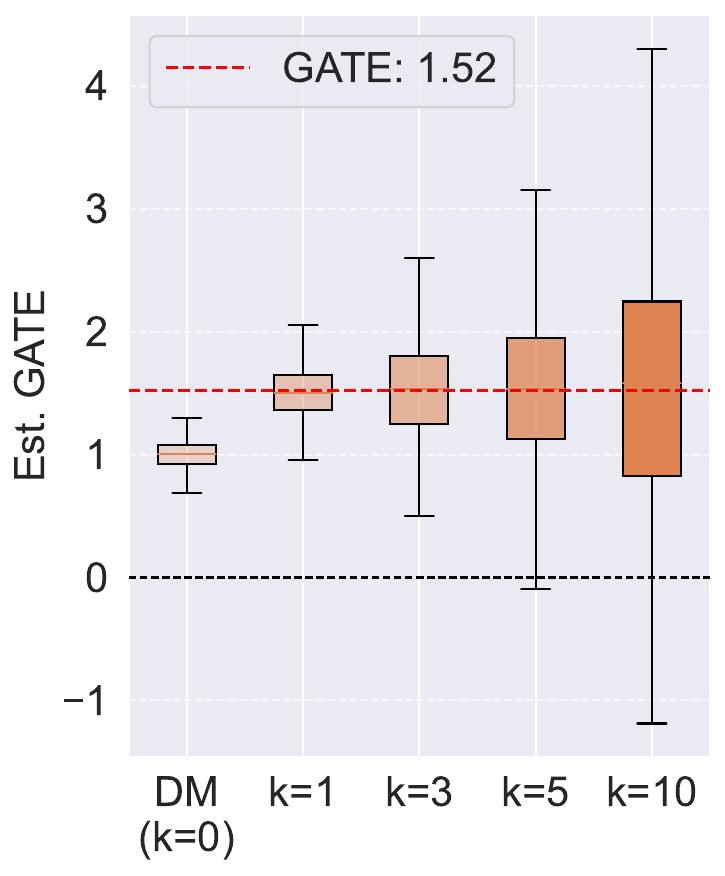}}
    \subfloat[{\color{black}$b=0.1, \gamma=0.3$}\label{fig:two_states_sub_b}]{\includegraphics[width=0.225\linewidth]{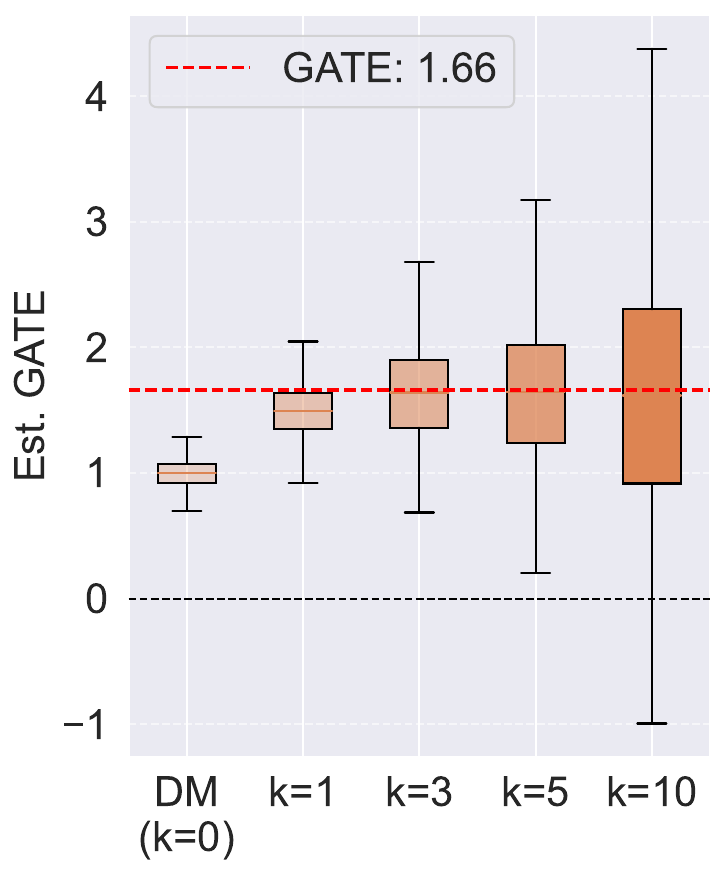}}
    \subfloat[{\color{black}$b=0.1, \gamma=0.5$}\label{fig:two_states_sub_c}]{\includegraphics[width=0.225\linewidth]{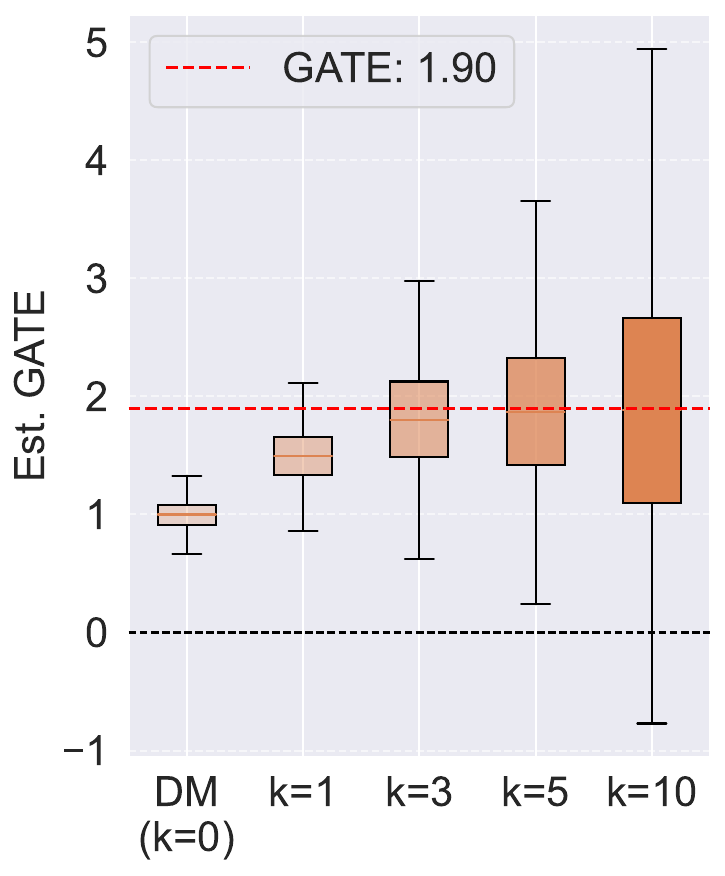}}
    \subfloat[{\color{black}$b=0.1, \gamma=0.7$}\label{fig:two_states_sub_d}]{\includegraphics[width=0.225\linewidth]{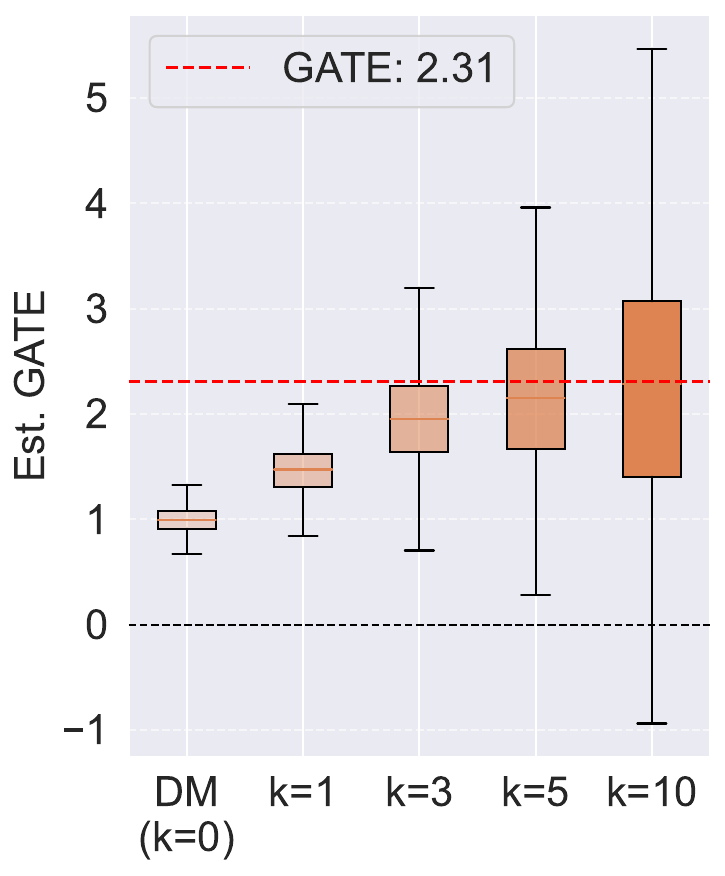}}\\[2pt]
    \subfloat[{\color{black}$b=0.3, \gamma=0.1$}\label{fig:two_states_sub_e}]{\includegraphics[width=0.225\linewidth]{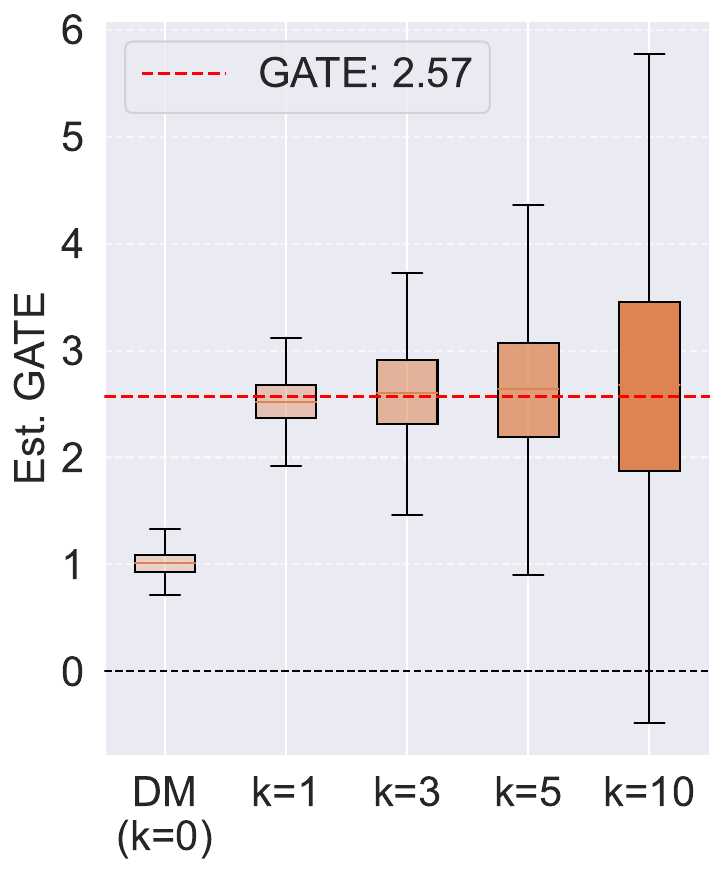}}
    \subfloat[{\color{black}$b=0.3, \gamma=0.3$}\label{fig:two_states_sub_f}]{\includegraphics[width=0.225\linewidth]{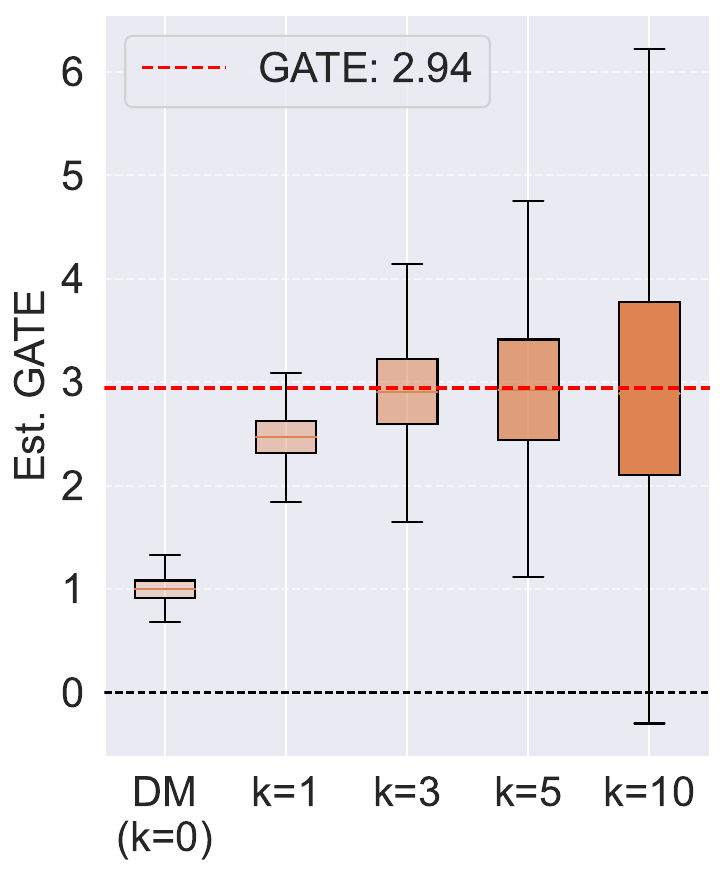}}
    \subfloat[{\color{black}$b=0.3, \gamma=0.5$}\label{fig:two_states_sub_g}]{\includegraphics[width=0.225\linewidth]{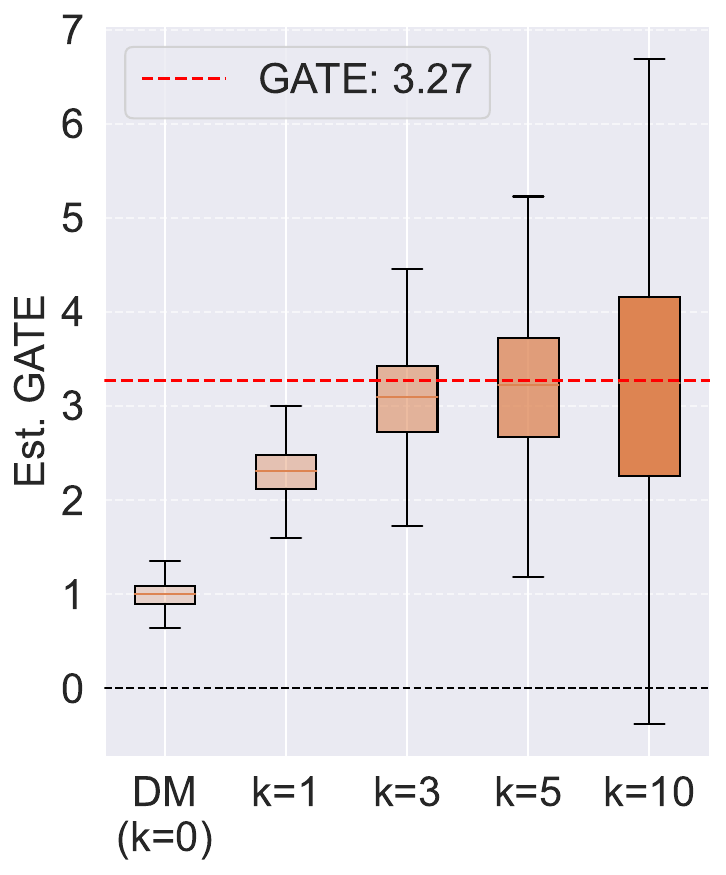}}
    \subfloat[{\color{black}$b=0.3, \gamma=0.7$}\label{fig:two_states_sub_h}]{\includegraphics[width=0.225\linewidth]{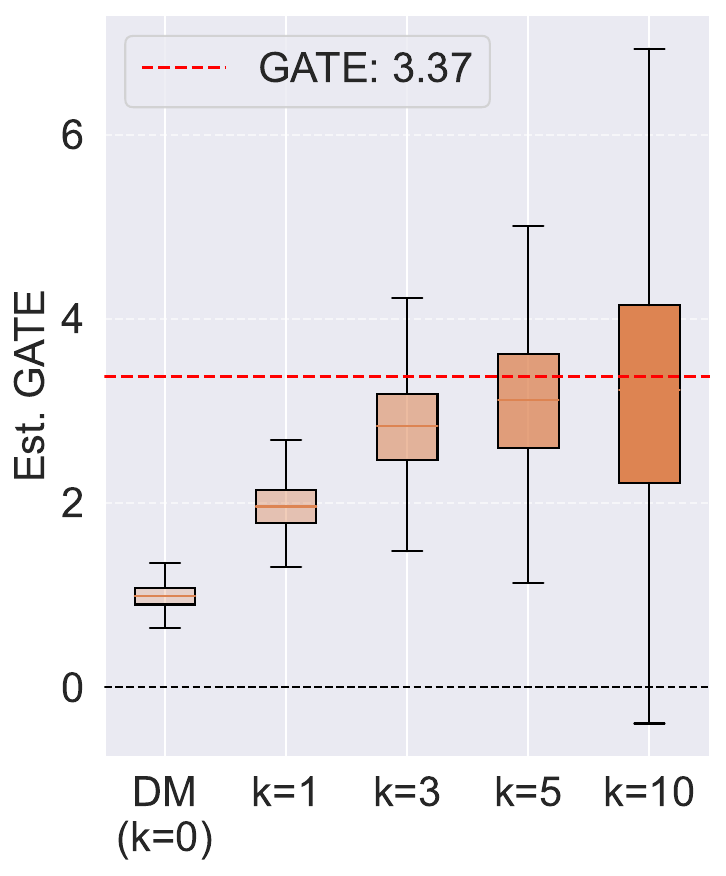}}\\[2pt]
        \subfloat[{\color{black}$b=-0.1, \gamma=0.1$}\label{fig:two_states_sub_a_ap}]{\includegraphics[width=0.225\linewidth]{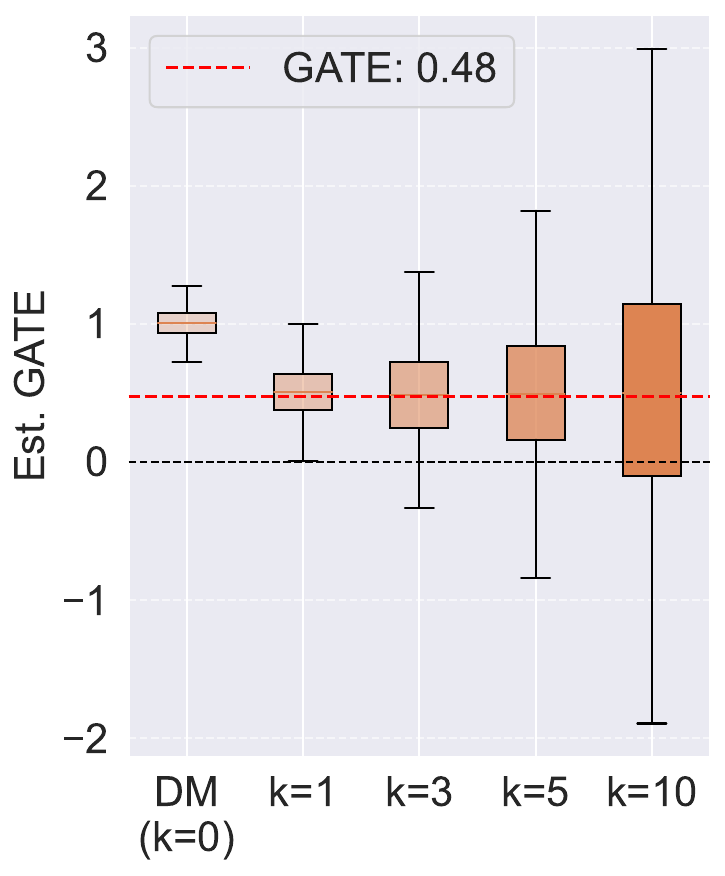}}
    \subfloat[{\color{black}$b=-0.1, \gamma=0.3$}\label{fig:two_states_sub_b_ap}]{\includegraphics[width=0.225\linewidth]{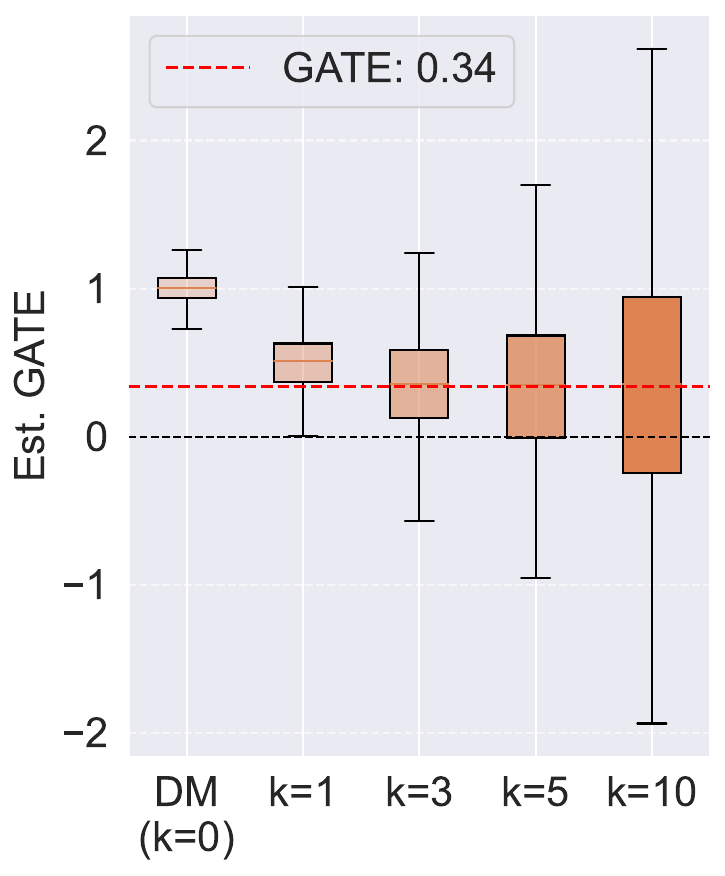}}
    \subfloat[{\color{black}$b=-0.1, \gamma=0.5$}\label{fig:two_states_sub_c_ap}]{\includegraphics[width=0.225\linewidth]{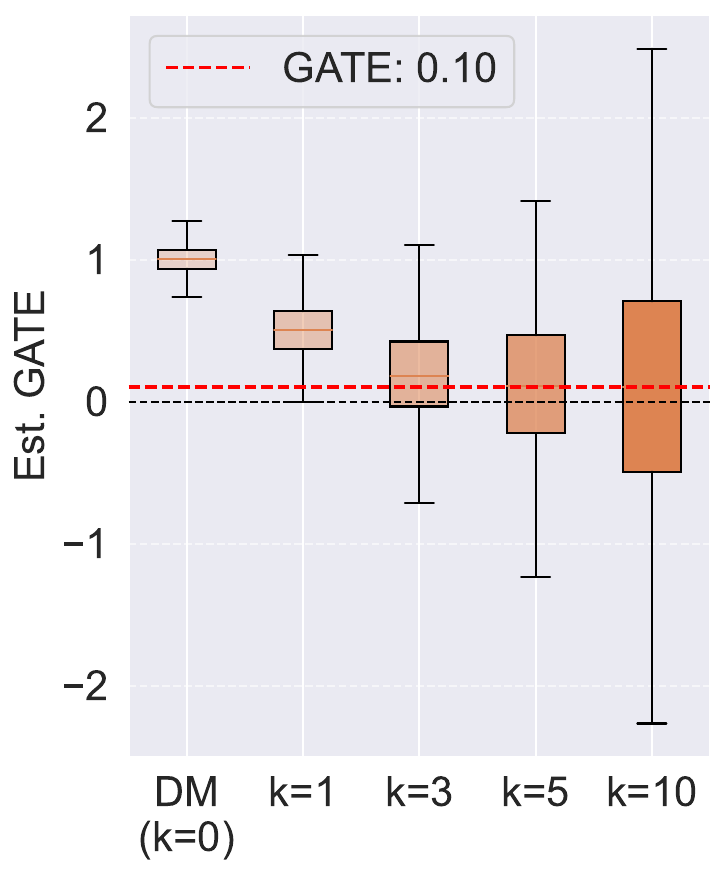}}
    \subfloat[{\color{black}$b=-0.1, \gamma=0.7$}\label{fig:two_states_sub_d_ap}]{\includegraphics[width=0.225\linewidth]{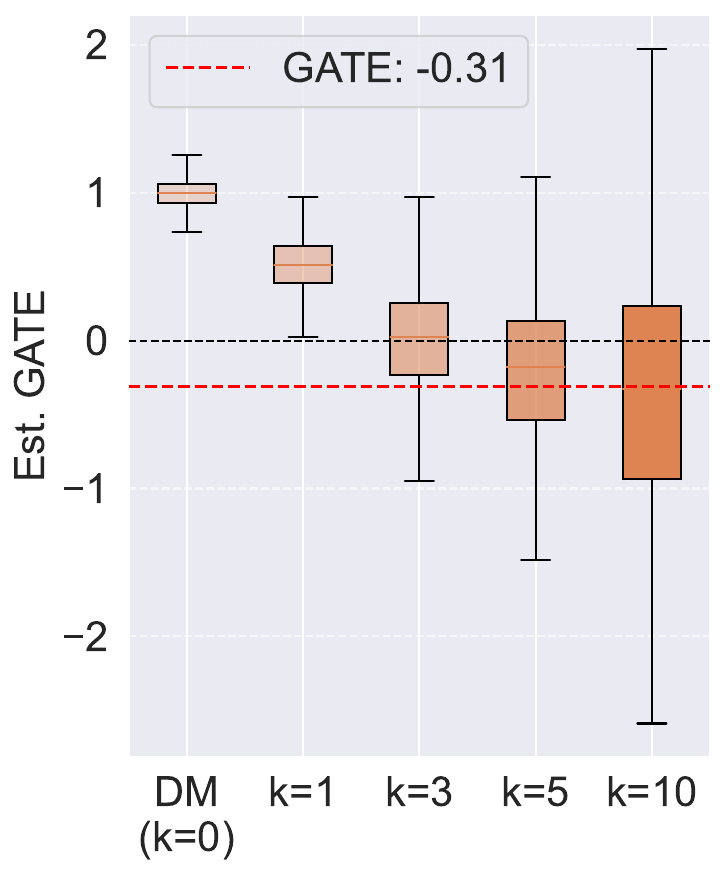}}\\[2pt]
    \subfloat[{\color{black}$b=-0.3, \gamma=0.1$}\label{fig:two_states_sub_e_ap}]{\includegraphics[width=0.225\linewidth]{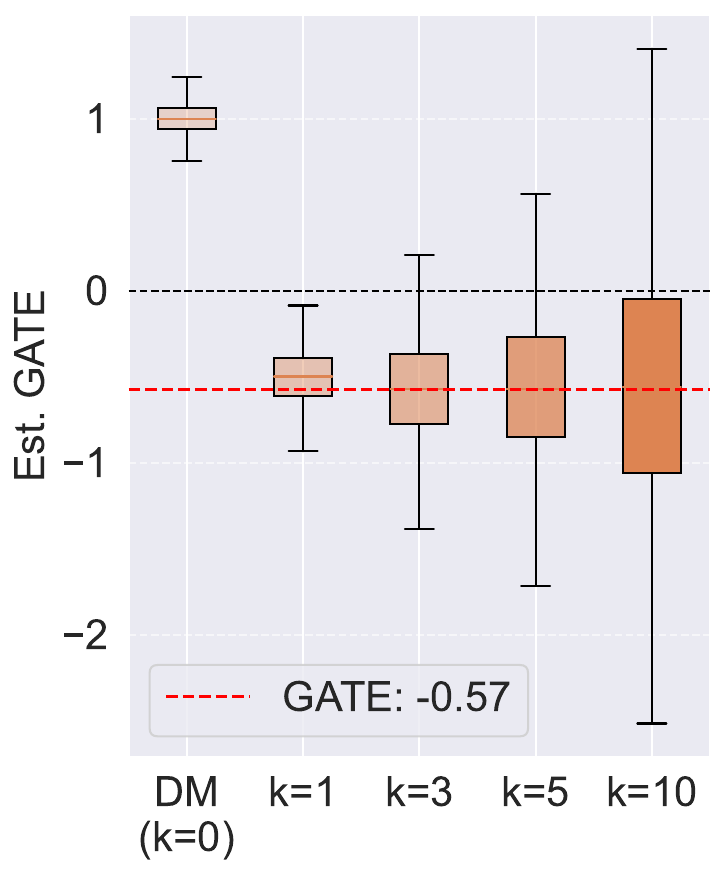}}
    \subfloat[{\color{black}$b=-0.3, \gamma=0.3$}\label{fig:two_states_sub_f_ap}]{\includegraphics[width=0.225\linewidth]{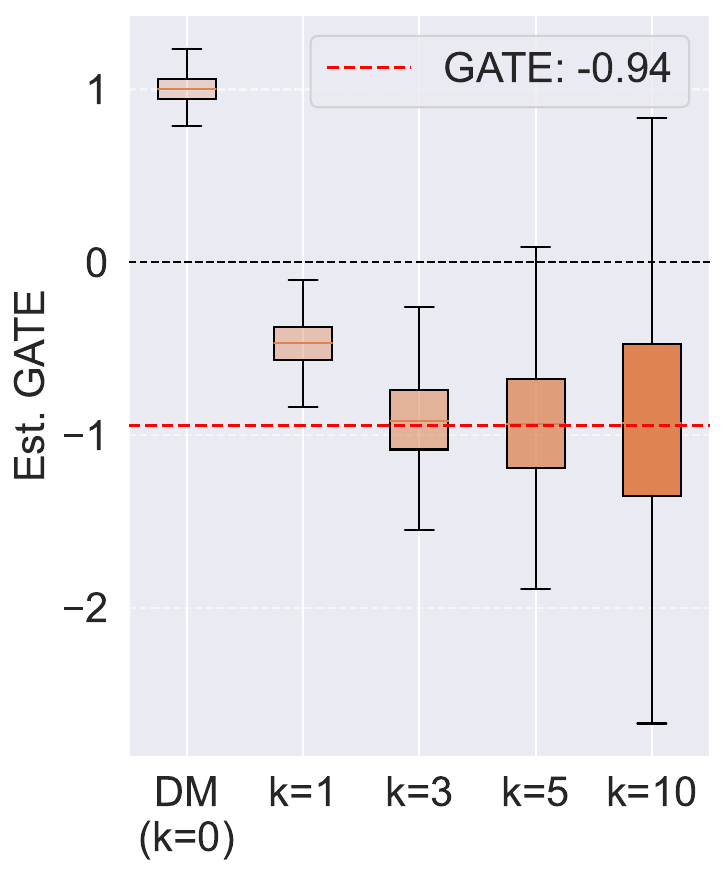}}
    \subfloat[{\color{black}$b=-0.3, \gamma=0.5$}\label{fig:two_states_sub_g_ap}]{\includegraphics[width=0.225\linewidth]{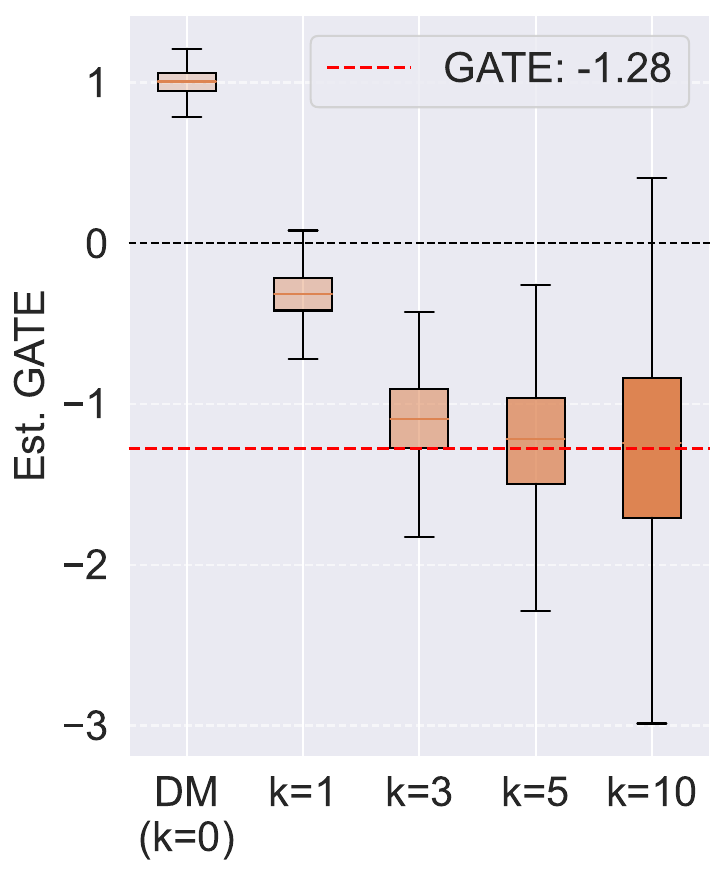}}
    \subfloat[{\color{black}$b=-0.3, \gamma=0.7$}\label{fig:two_states_sub_h_ap}]{\includegraphics[width=0.225\linewidth]{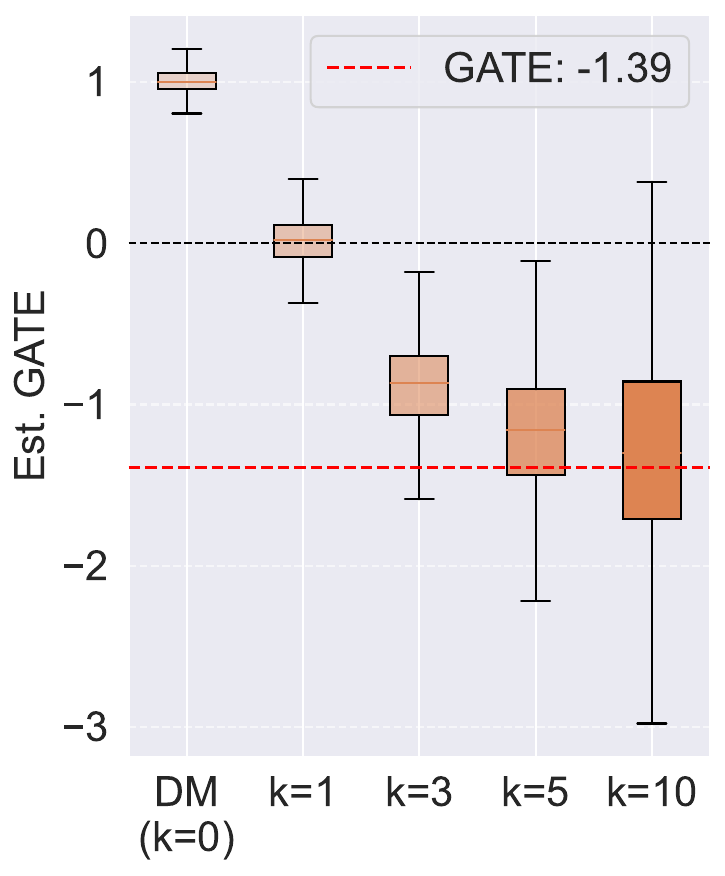}}
\end{minipage}
\caption{\color{black}Evaluation of TPG estimators with various $b$ and $\gamma$ in a 2-state MDP.\label{fig:numerical}}
\end{figure}

\subsection{Nonstationary queueing simulator}
\label{Appendix:real_case}

{In this section we detail our nonstationary queueing simulator, based on a nonstationary estimate of patient arrival rates to an emergency department (both based on day of week, and time of day), obtained using data from the SEEStat Home Hospital database \citep{seestat_data}; see Figure 9, Section 5.2, and Appendix B.1 of \cite{li2026experimenting} for further details.

The patient arrival rates are estimated using data from the Home Hospital database on admissions to an Israeli hospital emergency department over 2004, for every half hour of each day, at the day-of-week level (Sunday-Saturday).  Let $b_{d,t}$ denote this estimated arrival rate at time $t$, on a day of week $d$.  We simulate a four-week experiment, similar to \cite{li2026experimenting}. The arrival rate to the queueing system in week $w = 1, 2, 3, 4$, on day of week $d$, and time $t$, in state $k$, is specified as:
$$ \lambda_{w,d,t,k} = \frac{8 - 4p}{1 + k/5} a_w b_{d,t}, $$
where $a_1 = 0.9, a_2 = 1.0, a_3 = 1.1, a_4 =1.2$.  The service rate is $\mu = 20$ in all states.  In treatment, $p = 1.75$; in control, $p = 0.25$.

To simulate the system and map it to the discrete-time setting in our paper, we {\em uniformize} time by considering one-minute time steps.  This gives rise to $T = 40,320$ time steps per simulation.  We run 500 simulations of each type of experiment (Bernoulli randomized experiment, and switchback experiment). 
In the switchback experiment, we set the switchback interval to one hour (i.e., 60 time steps), and switch between treatment and control with probability $1/2$ at each interval, following the Bernoulli switchback design of \cite{hu2022switchback}.

\subsection{NYC ride-sharing simulator}
\label{Appendix:real_case_NYC}

We adopt a variant of the large-scale NYC ride-sharing simulator from \cite{peng2025differences} for the pricing experiment. The system observes a sequence of ``eyeballs'' (i.e., potential ride requests), where each rider specifies pickup and dropoff locations and is shown a posted price along with an estimated time of arrival (ETA). 
Based on this information, the rider decides whether to accept the trip. If accepted, the nearest available vehicle is dispatched and the system collects the trip fare as a reward; if rejected, no dispatch occurs and the reward is zero. The rider's acceptance behavior follows a simple logistic choice model:
$$
\mathbb{P}(\text {accept})=\sigma\left(w_{\text {price}} \cdot \text {price}+w_{\text {eta}} \cdot \text {ETA}+w_0\right),
$$
where $\sigma(\cdot)$ is the sigmoid function. Specifically, we set $w_{\text {price}}=-0.3, w_{\text {eta}}=-0.005$, and $w_0=4$. The price is computed as $\text {price}=\alpha \cdot \tau(\text {pickup location}, \text {dropoff location})$, where $\tau$ denotes the travel time matrix derived from real Manhattan traffic data. The ETA is computed as the driver's earliest arrival time at the pickup location minus the request time.
The pricing coefficient $\alpha$ is set to $0.01$ under the control policy and $0.02$ under the treatment policy. A similar modeling setup is described in Section~6.3 of \cite{peng2025differences}. As shown in Figure~\ref{fig:events}, rider arrivals in the NYC ride-sharing data from September 2024 exhibit complex, non-stationary patterns \citep{TLC}, presenting a significant challenge for treatment effect estimation.
}